\documentclass[manuscript, nonacm]{acmart}
\usepackage[none]{hyphenat}
\usepackage{cleveref}

\crefname{equation}{Eq.}{Eqs.}
\crefname{section}{Section}{Sections}
\crefname{figure}{Fig.}{Figs.}
\crefname{table}{Table}{Tables}
\crefname{appendix}{Appendix}{Appendices}
\crefname{theorem}{Theorem}{Theorems}
\crefname{thm}{Theorem}{Theorems}
\crefname{defi}{Def.}{Defs.}
\crefname{conjecture}{Conjecture}{Conjectures}
\crefname{proposition}{Prop.}{Props.}
\crefname{lemma}{Lemma}{Lemmas}
\crefname{corollary}{Corollary}{Corollaries}

\newcommand{\z}{\mathbf{0}}
\newcommand{\e}{\mathbf{e}}

\usepackage{blindtext}

\usepackage{caption}
\usepackage{subcaption}
\usepackage{algorithm}
\usepackage{algpseudocode}
\usepackage{circuitikz}
\usepackage{tikz}
\usetikzlibrary{calc, intersections, through, backgrounds, arrows, positioning, automata, shapes, shapes.geometric, fit}
\tikzset{elliptic state/.style={draw,ellipse}}

\usepackage{cleveref}
\usepackage{geometry}
\usepackage{array}
\usepackage{booktabs}
\usepackage{graphicx}
\usepackage{listings}

\usepackage{mathrsfs}
\usepackage[mathscr]{euscript}
\usepackage{relsize}
\usepackage{bigints}
\usepackage{verbatim}
\usepackage[english]{babel}
\usepackage{csquotes}
\usepackage{slantsc}
\usepackage{amsthm}
\usepackage{bbm}
\usepackage{placeins}
\usepackage{cancel}
\usepackage{setspace}
\usepackage{mathtools}
\usepackage{blkarray}
\usepackage{multirow}
\usepackage{bigdelim}
\usepackage{xcolor}
\usepackage{changepage}
\usepackage{enumitem}

\newcommand{\Z}{\mathbb{Z}}
\newcommand{\R}{\mathbb{R}}
\newcommand{\C}{\mathbb{C}}

\newcommand{\abs}[1]{\left\lvert #1 \right\rvert}

\newcommand{\m}{\medskip}

\newcommand{\bra}[1]{\left\langle\,#1\,\right\rvert}
\newcommand{\ket}[1]{\left\lvert\,#1\,\right\rangle}
\newcommand{\ketbra}[2]{\left\lvert\,#1\,\right\rangle\!\left\langle\,#2\,\right\rvert}

\newcommand{\dg}{\dagger}

\newtheoremstyle{mystyle}
  {}
  {}
  {\raggedright\sffamily}
  {}
  {\raggedright\normalsize\bfseries\sffamily}
  {:}
  {1em}
  {\thmname{#1}\thmnumber{ #2}\thmnote{ (#3)}}
\makeatother
\theoremstyle{mystyle}

\newtheorem{thm}{Theorem}[section]
\newtheorem{cor}{Corollary}[thm]
\newtheorem{lem}[thm]{Lemma}

\newtheorem{res}{Result}

\newtheoremstyle{propositional}
{10pt}
{10pt}
{
	\addtolength{\linewidth}{-2.0em}
	\parshape 1 2.0em \linewidth}
{}
{\sc}
{:}
{.5em}
{}
\makeatother

\theoremstyle{propositional}

\newtheoremstyle{definitive}
{10pt}
{10pt}
{
	\addtolength{\linewidth}{-2.0em}
	\parshape 1 0.0em \linewidth}
{}
{\raggedright\normalsize\bfseries\sffamily}
{:}
{.5em}
{}
\makeatother

\theoremstyle{definitive}
\newtheorem{defi}{Definition}

\newtheorem*{question}{Question}
\newtheorem*{remark}{Remark}

\newcounter{numb}

\newcounter{bean}
\newcounter{count}
\newcounter{count2}
\newcounter{count3}



\definecolor{gainsboro}{RGB}{220,220,220}
\definecolor{textclr}{RGB}{131,148,150}
\definecolor{monokai}{HTML}{272822}
\definecolor{darkb}{RGB}{0, 43, 54}

\newlength{\ointextsep}
\newcommand{\MG}{\mathfrak{J}_{\mathcal{G}}}

\makeatletter
\renewcommand\@authorsaddresses{}

\makeatother

\usepackage{xcolor}
\usepackage{etoolbox}

\newbool{showcomments}
\setbool{showcomments}{true}

\DeclareRobustCommand{\HL}[1]{\ifbool{showcomments}{{\color{orange}{#1}}}{}}

\begin{document}
\microtypesetup{expansion=false}

\raggedbottom

\title{Imposing Constraints on Driver Hamiltonians and Mixing Operators: From Theory to Practical Implementation}

\author{Hannes Leipold}
\authornote{Correspondence: \texttt{hleipold@fujitsu.com}}
\affiliation{
  \institution{Fujitsu Research of America}
  \city{Santa Clara}
  \state{CA}
  \country{USA}
}
\affiliation{
  \institution{Information Sciences Institute, University of Southern California}
  \city{Marina del Rey}
  \state{CA}
  \country{USA}
}
\affiliation{
  \institution{Department of Computer Science, University of Southern California}
  \city{Los Angeles}
  \state{CA}
  \country{USA}
}
\affiliation{
  \institution{Quantum Artificial Intelligence Laboratory (QuAIL), NASA Ames Research Center}
  \city{Moffett Field}
  \state{CA}
  \country{USA}
}
\affiliation{
  \institution{USRA Research Institute for Advanced Computer Science (RIACS)}
  \city{Mountain View}
  \state{CA}
  \country{USA}
}

\author{Federico Spedalieri}
\affiliation{
  \institution{Information Sciences Institute, University of Southern California}
  \city{Marina del Rey}
  \state{CA}
  \country{USA}
}
\affiliation{
  \institution{Department of Electrical and Computer Engineering, University of Southern California}
  \city{Los Angeles}
  \state{CA}
  \country{USA}
}

\author{Stuart Hadfield}
\affiliation{
  \institution{Quantum Artificial Intelligence Laboratory (QuAIL), NASA Ames Research Center}
  \city{Moffett Field}
  \state{CA}
  \country{USA}
}
\affiliation{
  \institution{USRA Research Institute for Advanced Computer Science (RIACS)}
  \city{Mountain View}
  \state{CA}
  \country{USA}
}

\author{Eleanor Rieffel}
\affiliation{
  \institution{Quantum Artificial Intelligence Laboratory (QuAIL), NASA Ames Research Center}
  \city{Moffett Field}
  \state{CA}
  \country{USA}
}

\begin{abstract}
Driver Hamiltonians and Mixing Operators that satisfy constraints is an important part of ansatz construction for many quantum algorithms. In this manuscript, we give general algebraic expressions for finding Hamiltonian terms and analogously unitary primitives, that satisfy constraint embeddings and use these to give complexity characterizations of the related problems. We prove that knowing if operators exist that enforce classical constraints is NP-Complete in the general case, but give algorithmic procedures with worse case polynomial runtime to find any operators with a constant locality bound; a useful result since many constraints imposed admit local operators to enforce them in practice. We then give algorithmic procedures to turn these algebraic primitives into Hamiltonian drivers and unitary mixers that can be used for Constrained Quantum Annealing (CQA) and Quantum Alternating Operator Ansatz (QAOA) constructions by tackling practical problems related to finding an appropriate set of reduced generators and defining corresponding drivers and mixers accordingly.  We consider a new QAOA approach based on the maximally disjoint subset as well as higher order constraint satisfaction terms for 1-in-3 SAT, which dramatically outperform the X-mixer.
\end{abstract}

\maketitle

\textbf{Key Words: } Quantum Optimization, Variational Quantum Algorithms, Quantum Alternating Operator Ansatz, Quantum Approximate Optimization Algorithm, Quantum Constraint Satisfaction

\begin{figure}[H]
\centering
\begin{tabular}{@{}c@{}c@{}c@{}}
\includegraphics[trim=0.2cm 0 0.2cm 0, clip, height=0.12\textheight]{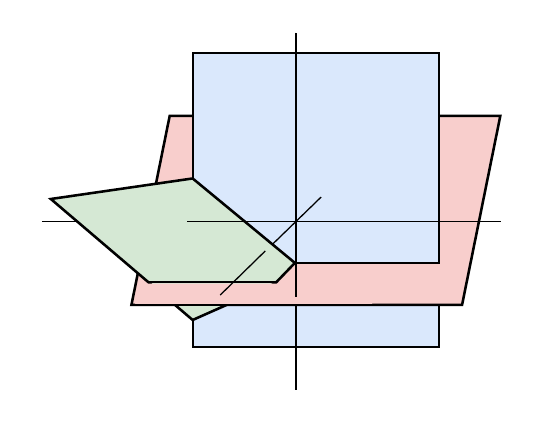} &
\includegraphics[height=0.12\textheight]{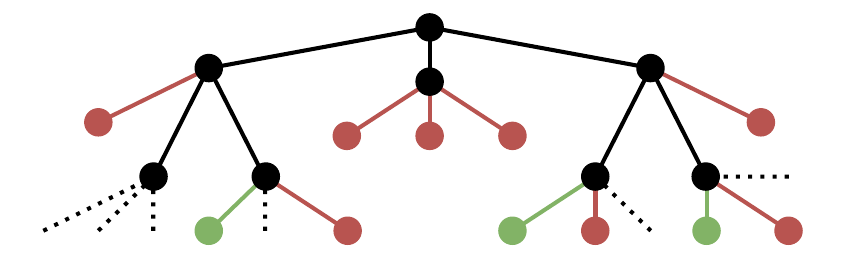} & \includegraphics[trim=0.2cm 0 0.2cm 0, clip, height=0.12\textheight]{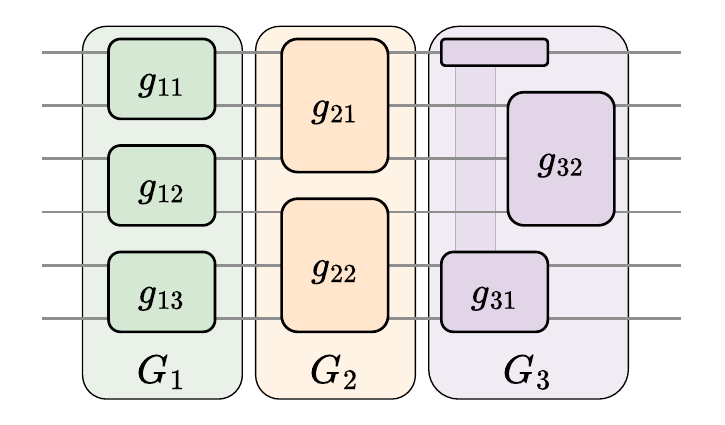}
\end{tabular}
\caption*{}
\label{fig:bt}
\end{figure}
\pagebreak
\tableofcontents

\section{Introduction}

The potential for achieving quantum advantages in the broadly important but computationally challenging field of classical optimization remains a tantalizing prospect for quantum technologies~\cite{abbas2023quantum}.
Quantum annealing (QA)~\cite{kadowaki1998quantum,farhi2000quantum} and Quantum Alternating Operator Ansatz (QAOA)~\cite{farhi2014quantum,hadfield2019quantum,wang2018quantum} are general paradigms for utilizing quantum systems to solve a broad class of optimization problems~\cite{farhi2000quantum,farhi2001quantum,santoro2006optimization,farhi2014quantum,hadfield2019quantum,abbas2023quantum}, including but not limited to Noisy Intermediate-Scale Quantum (NISQ)~\cite{preskill2018quantum} devices~\cite{ronnow2014defining, venturelli2015quantum,harrigan2021quantum,job2018test,dupont2023quantum,maciejewski2023design}.
While much effort has focused on the setting of unconstrained optimization,
hard constraints - conditions which solutions must strictly satisfy - are ubiquitous in real-world optimization problems and must be accounted for. Such constraints can be enforced in QA and QAOA by introducing penalty terms, thereby transforming the problem into an unconstrained formalism, such as unconstrained quadratic binary optimization (QUBO) problems~\cite{kochenberger2014unconstrained}, which can be mapped to Ising spin systems~\cite{santoro2006optimization,venturelli2015quantum,hadfield2021representation}.
However, penalty-based approaches typically come with various trade-offs and challenges~\cite{lucas2014ising,hadfield2019quantum} such as significantly increased quantum resource requirements and high energy coefficients to enforce feasibility.
As an alternative approach, recent developments in both QA and QAOA have focused on finding specialized drivers (in CQA)~\cite{hen2016driver,hen2016quantum,leipold2022quantum} and tailored ans{\"a}tze (in QAOA)~\cite{hadfield2019quantum,hadfield2021analytical,wang2020x,shaydulin2020classical,leipold2022tailored} to better exploit the underlying classical symmetries of the specific optimization problems that are being solved as well as problems in other domains~\cite{kremenetski2021quantum,cerezo2021variational,ragone2022representation}. These specialized drivers and mixers offer an alternative framework for enforcing hard constraints natively in quantum systems, thereby avoiding the use of penalty terms. They share the central feature that each commutes with an embedded constraint operator—representing the classical invariance enforced on the operator space. The terms of the driver Hamiltonian or the generators of the mixer must be selected from the operator space defined by this invariance operator.

\quad Various works have shown these specialized tools can have substantial benefits for practical quantum computers, including (1) reduced number of long-range qubit interactions~\cite{hen2016quantum, hen2016driver}, (2) enhanced noise resistance and mitigation~\cite{streif2021quantum, shaydulin2021error}, (3) improved trainability for QAOA~\cite{holmes2022connecting, shaydulin2021exploiting}, and (4) better performance on time-to-solution metrics~\cite{leipold2022tailored}. In this work we give a general formula for reasoning about enforcing several constraints on the evolution of quantum systems through the driver Hamiltonian (in CQA) or the mixing operators (in QAOA).

\begin{res}[Necessary and Sufficient Condition for Mixers]\label{result1}
A mixer written over the overcomplete basis provided commutes with the embedded constraint operator if and only if every term in the basis satisfies an efficiently checkable algebraic condition. In particular, we can exhaustively search for local mixers using this condition; no such condition exists for other local bases such as the Pauli basis.
\end{res}

\quad Finding operators that respect classical invariances and provide sufficient expressivity to explore the associated feasibility space are pivotal to developing effective, problem-tailored quantum algorithms. Results in Ref.~\cite{leipold2021constructing} are sufficient to show that the decision problem of knowing whether a Hamiltonian (or a quantum operator) commutes with a set of constraint embedding operators (i.e. imposing the constraints on the quantum evolution through the Hamiltonian or operator) is NP-Complete for linear constraints. From this, it is clearly NP-Hard for polynomial constraints, but we show it is NP-Complete using Result~\ref{result1}.

\begin{res}[NP-Completeness of Mixer Existence (Informal)]\label{result2}
Finding a mixer that is invariant for the embedded constraint operators associated with a collection of polynomial constraints is NP-Complete.
\end{res}

However, there is an important separation for local operators, i.e. when the number of qubits each operator acts upon is bounded by a constant. We show this problem is in \texttt{P}, and provide algorithmic prescriptions to tackle it efficiently.

\begin{res}[Efficient and Complete Polynomial Time Algorithm for Local Mixers (Informal)]\label{result3}
Finding a local mixer with a general collection of polynomial constraints is in \texttt{P}. We provide an efficient yet complete polynomial time backtracking algorithm to find local mixers.
\end{res}

\quad \cref{sec:complexity} proves that for polynomial constraints, knowing the existence of a driver Hamiltonian or mixing operator that imposes the constraints is in general NP-Complete and in \texttt{P} for local Hamiltonians or unitary operators by utilizing the algebraic framework developed in \cref{sec:algecond}. We then develop practical algorithms, that provide significant speed-up compared to the brute force approaches, for finding the symmetry terms up to a specified locality and selecting a set of generators for them as well as compiling the generators into unitaries.

\quad We then consider the question of ansatz construction for QAOA applied to random instances of 1-in-3 SAT, comparing three different ans{\"a}tze that are trained using finite difference parameter shift gradient descent on random instances of size 12 with QAOA-depth $ p = 14 $. We then juxtapose their performance on random instances of sizes between $ 12 $ and $ 22 $. The first ansatz is the X-mixer with the standard phase-separating operator. Since the constraints are local and relatively sparse, the second ansatz is based on imposing the maximum disjoint subset (MDS) of constraints on the mixer while the rest of the constraints are represented in the phase-separating operator. The third approach then considers mixers associated with each constraint and all constraints that share variables with that particular constraint. Curve fitting on random 1-in-3 SAT instances indicates improved empirical scaling for the problem-tailored ansatz constructions, consistent with an approximately quadratic improvement relative to the X-mixer baseline, with the best performance obtained using the automated ansatz construction.

\begin{res}[Empirical Scaling Advantage for Tailored Ans{\"a}tze against the X-mixer on 1-in-3 SAT]
By tailoring mixers to the structure of 1-in-3 SAT, we observe improved empirical time-to-solution scaling for QAOA with $p=14$ on instance sizes between $12$ and $22$ relative to the X-mixer baseline, consistent with an approximately quadratic improvement in the fitted exponential scaling.
\end{res}

\section{Driver Hamiltonians and Mixing Operators in Quantum Computing}\label{sec:driver_mixer}

\begin{figure}
\centering
\includegraphics[width=0.90\textwidth]{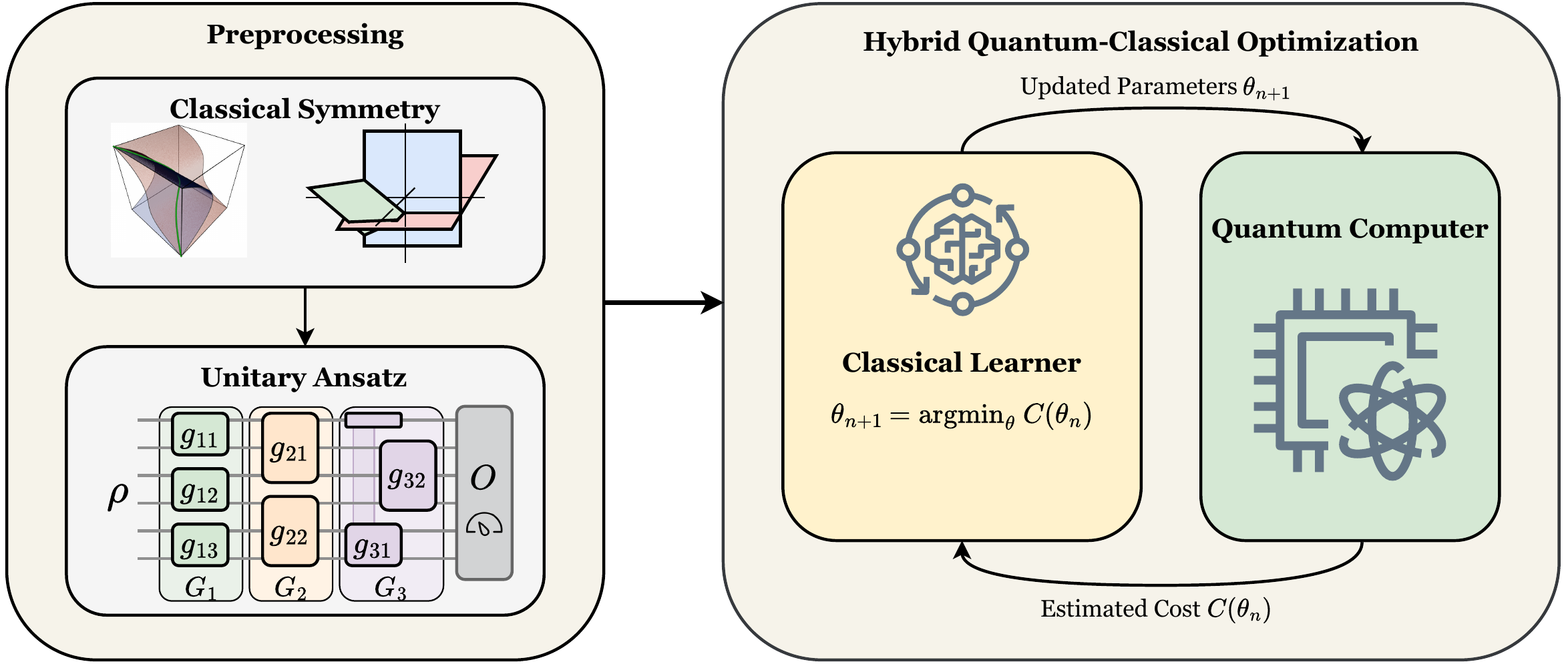}
\caption{\textbf{Symmetry Imposed Ans{\"a}tze in Training Variational Quantum Algorithms.} Control parameters of a quantum computer are trained via a classical learner to minimize a cost function. Before the training can begin, a parameterized unitary $ U(\theta) $, must be selected during the preprocessing. Classical symmetries inform the selection of an ansatz; that is, the classical invariance should be imposed on the ansatz. The ansatz determines what gates will be applied during a computation on the quantum computer and how these gates are parameterized.}
\label{fig:sym_proc_diagram}
\end{figure}

In this manuscript, we delineate the mathematical foundations for a large class of ans{\"a}tze: driver and mixers that at sufficient depth can explore a subspace constrained to a classical symmetry. Such drivers and mixers have been studied in the context of optimization problems since at least Ref.~\cite{hen2016quantum}. Determining an appropriate driver Hamiltonian~\cite{somma2012quantum,del2013shortcuts,lanting2017experimental} or circuit ansatz~\cite{wecker2015progress,cong2019quantum,wiersema2020exploring,nakaji2021expressibility,ragone2022representation} is a central task for achieving maximum utility of quantum computers. Especially in the context of variational quantum algorithms (VQAs)~\cite{cerezo2021variational,hadfield2019quantum,saleem2023approaches} and NISQ~\cite{preskill2018quantum} quantum annealing~\cite{hen2016quantum,hen2016driver,leipold2022quantum} (QA), ansatz construction has become a major focal point for modern quantum algorithms. Despite many exciting avenues and important foundational results on the value proposition of different ansatz constructions, their selection largely remains an art within a broad mathematical framework. Given a specific hardware set and problem of interest, selecting the best ansatz based on some criteria is one of the most important challenges in quantum heuristics today. \cref{fig:sym_proc_diagram} depicts how classical symmetries determine the unitary ansatz of a VQA.

\quad Consider a general constrained optimization problem over binary variables $ x = (x_{1}, \ldots, x_{n} ) \in \Z_{2}^{n} $ with a set of constraints $ \mathcal{C} = \{ C_{1}, \ldots, C_{m} \} $, where $ C_{i}(x) = b_{i} $ is an equality polynomial constraint function for some $ b_{i} $. Then we desire to minimize an objective polynomial function $ f(x) $ while satisfying each constraint. The constraint optimization problem can then be written as:
\begin{align}
\min \; & f(x) \nonumber\\
\text{subject to } \; & C_{i}(x) = b_{i},\; \text{for} \; i \in [1,\ldots,m]
\end{align}

\quad Inequality constraints can be handled in a number of ways through this formulation; such constraints can be turned into soft constraints through penalty terms~\cite{bottarelli2025inequality}, tracked with ancilla qubits~\cite{bucher2025penalty}, transformed into equality constraints with slack variables, enforced through other quantum effects like Zeno processes~\cite{herman2023portfolio}, or cast into high(er) order polynomial equality constraints (to name a few).

\quad

Ref.~\cite{lucas2014ising} considers quadratic unconstrained binary optimization (QUBO) formulations for many NP-Hard problems, with inequality conditions enforced by several penalty terms~\cite{kuroiwa2021penalty}. Ref.~\cite{gabbassov2025lagrangian} considered Lagrangian-relaxed QUBO formulation obtained from the Lagrangian dual of the constrained problem. Many problems can also be described by higher order unconstrained binary optimization (HUBO) formulations~\cite{djidjev2023quantum}.

\m

\quad Quantum Alternating Operator Ansatz (QAOA) and Constrained Quantum Annealing (CQA) share a similar assumption for solving optimization problems. Each bit $ x_{i} \in \{ 0, 1 \} $ of $ x $ is mapped to $ \ket{x_{i}} \in \{ \ket{0}, \ket{1} \} $, such that $ \ket{x} = \ket{x_{1}} \cdots \ket{x_{n}} \in \{ \ket{0}, \ket{1} \}^{\otimes n} $. Then each constraint is mapped to a constraint embedded operator $ \hat{C} $ that is an observable over the computational basis such that $ C(x) = \bra{x} \hat{C} \ket{x} $ for any $ x $.

\quad Under nonrelativistic closed dynamics governed by a Hamiltonian $ H(t) $, if $ \bra{\psi(0)} \hat{C} \ket{\psi(0)} = b $ at time $ t = 0 $ and $ [ \hat{C}, H ] = 0 $, then $ \bra{\psi(t)} \hat{C} \ket{\psi(t)} = b $ for all time $ t $. Hence, $ \hat{C} $ is an invariant observable over $ \ket{\psi(t)} $ and $ H $ maintains the symmetry imposed by $ C $. CQA attempts to (approximately) minimize the optimization function $ f(x) $ \textit{within} the constraint space. Let
\begin{align}
\hat{F} = \sum_{ x \in F } \ketbra{x}{x}, \; F = \{ x \, : \, x \in \{0,1\}^{n}, \;  C_{i}(x) = b_{i} \; \forall C_{i} \in \mathcal{C} \}
\end{align}

represent the projection operator into the embedded feasible space for a collection of constraints $ \mathcal{C} $. If the optimization function, $ f(x) $, is a constant function, the problem is a constraint satisfaction problem. Notice that \textit{any} operator written in the diagonal computational basis - the space $ \text{span}(\{ \mathbbm{1}, \sigma^{z} \}^{\otimes n}) $ - will commute with $ \hat{C} $ since they are mutually diagonalizable in the computational basis.

\quad Standard CQA then attempts to solve a constrained optimization problem by constructing a dynamic Hamiltonian $ H(t) = A(t) \, H_{d} + B(t) \, H_{f} $, where $ H_{d} $ is a \textit{driver} Hamiltonian and $ H_{f} $ is the embedded cost operator for $ f(x) $. If $ \ket{ \psi(0) } $ is in the feasibility space, $ \bra{\psi(0)} \hat{F} \ket{\psi(0)} = 1 $, and $ [ \hat{C}_{i}, H_{d} ] = 0 $ for all $ i $, then $ \bra{\psi(t)} \hat{F} \ket{\psi(t)} = 1 $ for any time $ t $. As such, $ \ket{\psi(t)} $ will evolve within the constraint space.  If $ H_{d} $ irreducibly commutes with $ \mathcal{C} $, $ \ket{\psi(0)} $ is a (approximate) ground state of $ H_{d} $, and there is a smooth evolution from $ A(0) = A, B(0) = 0 $ to $ A(T) = 0, B(T) = B $, then the adiabatic theorem states that $ \ket{\psi(T)} $ is a (approximate) ground state of $ H_{f} $ for sufficiently large total time $ T $.

\quad An analogous discussion can be found in QAOA. Standard QAOA of $p$-depth is described by the application of two unitaries for $ p $ rounds. H-QAOA~\cite{hadfield2019quantum} refers to the closest relative to CQA, $ U(\vec{\alpha}, \vec{\beta}) = U(\beta_{p},\alpha_{p}) \cdots U(\beta_{1}, \alpha_{1}) $ where $ U(\alpha_{j}, \beta_{j}) = U_{m}(\beta_{j}) U_{c}(\alpha_{j}) $ with $ U_{c}(\alpha_{j}) = e^{-i \, \alpha_{j} \, H_{f}} $ and $ U_{m}(\beta_{j}) = e^{-i \, \beta_{j} \, H_{d}} $ for $ \vec{\alpha}, \vec{\beta} \in [0, 2\,\pi]^{p} $ (or more generally $ \vec{\alpha}, \vec{\beta} \in \R_{+}^{p} $ depending on energy scalings). If angles $ \vec{\alpha}, \vec{\beta} $ are selected small enough, this becomes a Trotterization of CQA~\cite{yi2022spectral}. QAOA more generally considers alternative formulations for what is often called the mixing operator $ U_{m} $, in place of what $ e^{-i \, \beta \, H_{d}} $ accomplishes in H-QAOA, since $ e^{-i \, \beta \, H_{d}} $ can be difficult to implement and can require high depth circuits.

\quad The phase-separation operator, which introduces phases in the computational basis based on the quality of a state written in that basis.
\begin{defi}[Phase-separation Operator]
Given a cost function $f(x)$, the phase-separation operator applies phasing based on the cost:
\begin{align}\label{def:phase_separate_op}
U_{c}(\alpha) = e^{i \, \alpha \, H_{f}},
\end{align}
such that the phase of a computational basis state $ \ket{x} $ is updated based on its cost, $ U_{c}(\alpha) \ket{x} = e^{i \, \alpha \, f(x) } \ket{x} $.
\end{defi}

$ H_{f} $ is composed of terms written in the diagonal of the computational basis states and therefore each term commutes with one another. Typically, as for NP-Hard optimization problems, $ H_{f} $ is made of several local body terms, $ H_{f} = \sum_{t} H_{t} $, where each $ H_{t} $ is a Hermitian operator acting on a fixed number of qubits (not dependent on $ n $) and so the phase-separation operator can be written as $ U_{c}(\alpha) = e^{i \, \alpha \, \sum_{t} H_{t} } = \prod_{t} e^{-i \, \alpha \, H_{t}} $.

\begin{defi}[Mixing Operator]
The mixing operator nontrivially \textit{mixes} states written in the computational basis with their \textit{neighbors} under the dynamics of each driver term in the mixing operator. A driver term $g$ is any Hamiltonian that is not diagonal in the computational basis. Then given the collection $G$ of, in general, noncommuting and therefore ordered driver terms, the mixing operator is
\begin{align}\label{eq:def_mixing_op}
    U_{m}(\beta) = \prod_{g \in G} e^{-i \, \beta \, g} .
\end{align}
\end{defi}

For example, any driver Hamiltonian $H_{d}$ could be a mixer as $U_{d} = e^{i \beta H_{d}}$. However, $ H_{d} $ could be composed of many noncommuting (and typically acting on only a fixed number of qubits) terms: $ H_{d} = - \sum_{g \in G} g $ for a set $ G $ of indices or mappings. Then $ e^{i \, H_{d}} $ may be impractical to implement or require high depth circuits due to this noncommutativity. However, the set of generators $ G $ can be partitioned into $ \mathcal{G} = \{ G_{1}, \ldots, G_{k} \} $ such that for all $ g_{ij}, g_{ik} \in G_{i} $, $ [ g_{ij}, g_{ik} ] = 0 $ and the disjoint union $ \coprod_{G_{i} \in \mathcal{G}} G_{i} = G $. Then the mixer is transpiled for \cref{eq:def_mixing_op}:
\begin{align}
U_{m}(\beta) = \prod_{g \in G} e^{-i \, \beta \, g } \equiv \prod_{G_{i} \in \mathcal{G}} \prod_{g \in G_{i}} e^{-i \, \beta \, g},
\end{align}
which can have a significantly lower circuit depth to implement than $ e^{i \, H_{d}} $.

\begin{defi}[Generator (of Mixing or Driving)]\label{defi:genetator_mixing_driving}
Given a unitary $ U = e^{i \, \theta \, g} $, $ g $ is the generator of the unitary $ U $. Given a Hamiltonian $ H = \sum_{i} g_{i} $, each $ g_{i} $ is a generator of $ H $.
\end{defi}

\quad A central question in both approaches is then how to find a suitable collection of generators $ g $ such that $ [ \hat{C_{i}}, g ] = 0 $ for all $ g $ and all $ i $.

\begin{question}[Mixer Selection for Constrained Optimization (Informal)]
What is a suitable ordered collection of generators $ \mathcal{G} = \{ g_{1}, \ldots, g_{|\mathcal{G}|} \} $ such that $[ \hat{C}_{i}, g ] = 0 $ for all constraints $C$?
\end{question}

In the literature this problem has been handled through specialized analysis for problem domains~\cite{hadfield2019quantum,torggler2019quantum,hodson2019portfolio,chancellor2019domain}. In the case that the feasible space happens to decompose into a simple cartesian sum $ F = B_1 \times B_2 \times \cdots \times B_{r} $ with $|B_{j}| = \text{poly}(n)$, such as with disjoint constraints with each constraint over a few variables, we can build mixers through hopping terms between each entry~\cite{fuchs2022constraint} of each set $B_{j}$. For simple structured problems, preparing a superposition over the embedded feasible subspace, such as the Dicke state for cardinality constrained optimization, is possible and so a Grover diffusor can be constructed~\cite{bartschi2020grover}, such as an alternative to the popular XY mixer~\cite{wang2020x,kordonowy2026lie} for cardinality constrained optimization.

\quad We developed this into an algebraic question in \cref{sec:algecond}, developed into an algorithmic problem from a complexity theory perspective in \cref{sec:complexity} and then practical algorithms to solve this problem are provided in \cref{sec:theory_to_prac}. Typically, we are interested in \textit{$k$-local} operators, that only require $k=\mathcal{O}(1)$ interactions to implement as $ n $ grows. For many optimization problems of interest, phase-separating and mixer operators that are local are sufficient.

\begin{defi}[$k$-Locality]\label{def:klocality}
A Hamiltonian $ H $ is $k$-local over $n$ qubits if $ H = \sum_{t} h_{t} $, where $h_{t}$ acts non-trivially on at most $k$ qubits. A unitary is $k$-local if there exists a $k$-local Hamiltonian $ H $ such that $ U = e^{i \, H} $.
\end{defi}

\begin{defi}[Locality Weight]\label{def:localityweight}
The \textit{locality weight} of an operator $H$ (or $U=e^{iH}$) is the \textit{minimum} $k$ such that such $H$ is $k$-local.
\end{defi}

Our approach yields several important conclusions. We are able to classify the hardness of this question in \cref{sec:complexity}. Local operators constructed through careful analysis of specific problems are a simple subset of those our polynomial time algorithm yields.
Constructing hopping terms over brute force enumeration of feasible states is in general intractable, but even when tractable, can require extraordinary resources; our approach helps resolve this concern~\cite{fuchs2022constraint} by limiting the locality weight of terms. Indeed, for many practical problems of interest, from a global linear constraint~\cite{hen2016quantum,hen2016driver} to highly structured overlapping local constraints~\cite{leipold2022quantum} to graph-based quadratic constraints~\cite{hadfield2019quantum,saleem2023approaches,leipold2024train}, only $\mathcal{O}(n)$ mixers are necessary to be able to explore the associated embedded feasible subspace for problems of size $n$.

\section{Algebraic Conditions for Imposing Constraints}\label{sec:algecond}

\begin{table}[t]
\centering
\begin{tabular}{c|c|c|c|c|c|c|c}
 & Basis & N & S & L & H & O \\
\hline
1 & $ \{ \mathbbm{1}_2, \sigma^{x}, \sigma^{y}, \sigma^{z} \}^{\otimes n} $ & False & False & True & True & False \\
2 & $ \{ \mathbbm{1}_2, \sigma^{z}, \sigma^{-}, \sigma^{+} \}^{\otimes n} $ & False & True & True & False & False \\
3 & $ \{ \sigma^{0}, \sigma^{1}, \sigma^{-}, \sigma^{+} \}^{\otimes n} $ & True & True & False & False & False \\
4 & $ \{ \mathbbm{1}_2, \sigma^{0}, \sigma^{1}, \sigma^{-}, \sigma^{+} \}^{\otimes n} $ & True & True & True & False & True
\end{tabular}
\vspace{0.5em}
\caption{{\bf Properties of Different Qubit Basis. } Classification of the properties of different bases for complex matrices over $ \mathbbm{C}^{2^{n}} $ with the mapping (N) for necessarily zero ((S) for sufficient to show zero) if a basis term is used for commuting with an embedded constraint operator (\cref{eq:defconop}). (L) is for admitting local operators, (H) is if the basis terms are all Hermitian, and (O) is if the basis is overcomplete (the number of matrices is larger than the dimension of the matrix space). (1) is the Pauli basis. (4) is the focus of this section, which we show satisfy conditions (N), (S), and (L) that are important for algorithmic primitives used for complexity results in \cref{sec:complexity} and application in \cref{sec:theory_to_prac}. It is not (H) and it is (O). Hermiticity is enforced by requiring a term's transpose with the coefficient set as its complex conjugate. Overcompleteness is resolved by always selecting the most local representation. (3) satisfies (N) and (S), but is not (L), (H), or (O), while (2) satisfies (S) and (L), but is not (N), (H), or (O). However, (2) is (N) for linear constraints and discussed in App.~\ref{app:suffquad} for enforcing (2-local and higher) Ising spin-z Hamiltonians.} \label{tab:comparebasis}
\end{table}

In this section we derive a necessary and sufficient condition for a Hamiltonian to commute with a polynomial embedded constraint operator. This condition, \cref{thm:gencom}, is utilized in \cref{sec:complexity} to show that finding local Hamiltonians that commute with embedded constraints is a task in polynomial time and \cref{sec:theory_to_prac} discusses practical algorithms for utilizing \cref{thm:gencom} to find a basis for nontrivial operators over the embedded feasible space. \cref{tab:comparebasis} shows the properties of several basis for writing complex matrices in $ \mathbbm{C}^{2^{n} \times 2^{n}} $, with the unique property of the basis discussed in this section being that it admits a description of local operators and that commutation with an embedded constraint operator for a Hamiltonian written over the basis is zero if and only if commutation with each individual basis is zero. This locality condition and the necessity and sufficiency of the commutation condition is unique to this basis, see \cref{tab:comparebasis} for descriptions of other candidate basis.

\subsection{A Frame for Classical Symmetries of Quantum Operators}

\quad Let $ \sigma^{x}, \sigma^{z}, \sigma^{y} $ represent the standard Pauli matrices. Then let:
\begin{align}
\sigma^{0} &= \ketbra{0}{0} = \left( \mathbbm{1} + \sigma^{z} \right) / 2,      \nonumber\\
\sigma^{1} &= \ketbra{1}{1} = \left( \mathbbm{1} - \sigma^{z} \right) / 2,      \nonumber\\
\sigma^{-} &= \ketbra{0}{1} = \left( \sigma^{x} + i \, \sigma^{y} \right) / 2,  \nonumber\\
\sigma^{+} &= \ketbra{1}{0} = \left( \sigma^{x} - i \, \sigma^{y} \right) / 2.
\end{align}

Recall the standard definition of the Kronecker tensor product for matrices.
\begin{defi}[Tensor Product for Matrices]
Given $ A \in \C^{p \times p}, B \in \C^{q \times q} $, the tensor product $ A \otimes B \in \C^{p \, q \,  \times p \, q} $ has entries given by:
\begin{align}
(A \otimes B)_{(i-1)\,q+k, (j-1) \, q + \ell} = A_{ij} B_{k \ell} \; \text{ for } 1 \leq i, j \leq p, 1 \leq k, \ell \leq q .
\end{align}
\end{defi}

For example, given two matrices over $\mathbb{C}^{2 \times 2}$
\begin{align}
A = \begin{pmatrix}
a_{11} & a_{12} \\
a_{21} & a_{22} \\
\end{pmatrix}, \\
B = \begin{pmatrix}
b_{11} & b_{12} \\
b_{21} & b_{22}
\end{pmatrix},
\end{align}
the Kronecker product $A \otimes B $ over $\mathbb{C}^{4 \times 4}$ is
\begin{align}
A \otimes B &= \begin{pmatrix}
a_{11} B & a_{21} B \\
a_{21} B & a_{22} B
\end{pmatrix} \\
&= \begin{pmatrix}
a_{11} b_{11} & a_{11} b_{12} & a_{12} b_{11} & a_{12} b_{12} \\
a_{11} b_{21} & a_{11} b_{21} & a_{12} b_{21} & a_{12} b_{22} \\
a_{21} b_{21} & a_{21} b_{21} & a_{22} b_{21} & a_{22} b_{22} \\
a_{21} b_{21} & a_{21} b_{21} & a_{22} b_{21} & a_{22} b_{22} \\
\end{pmatrix} .
\end{align}

We add subscripts based on which qubit the operator is applied, with the identity operator $\mathbbm{1}_{2} = \ketbra{0}{0} + \ketbra{1}{1} $ applied to every other qubit:
\begin{align}
\sigma_{j}^{k} = \bigotimes_{\ell=1}^{n} \left( \sigma^{k} \right)^{\delta_{\ell j}} = \underbrace{\mathbbm{1}_2 \otimes \mathbbm{1}_2 \otimes \cdots \otimes \mathbbm{1}_2}_{1:j-1} \otimes \, \sigma^{k} \, \otimes \underbrace{\mathbbm{1}_2 \otimes \mathbbm{1}_2 \otimes \cdots \otimes \mathbbm{1}_2}_{j+1:n} .
\end{align}

\quad Left and right multiplication identities for these matrices are found in Appendix~\ref{app:single_qubit_relations} and are useful for derivations within this section. A useful definition used throughout this manuscript is for a single term $ T $ over the overcomplete operator basis $ \{ \mathbbm{1}_{2}, \sigma^{0}, \sigma^{1}, \sigma^{-}, \sigma^{+} \}^{\bigotimes n} $ (for $ \mathbbm{C}^{2^n \times 2^n} $):
\begin{align}\label{eq:defi_basis_term}
T(\vec{x}, \vec{y}, \vec{v}, \vec{w}) &= \bigotimes_{j=1}^{n} \bigg( \sigma^{0} \bigg)^{x_j} \bigg( \sigma^{1} \bigg)^{y_{j}} \bigg( \sigma^{+} \bigg)^{v_{j}} \bigg( \sigma^{-} \bigg)^{w_{j}} \nonumber \\
&= \prod_{j=1}^{n} \bigg( \sigma_{j}^{0} \bigg)^{x_{j}} \bigg( \sigma_{j}^{1} \bigg)^{y_{j}} \bigg( \sigma_{j}^{+} \bigg)^{v_{j}} \bigg( \sigma_{j}^{-} \bigg)^{w_{j}},
\end{align}

where $ \vec{x}, \vec{y}, \vec{v}, \vec{w} \in \{0,1\}^{n} $ are assumed to be orthogonal binary vectors: $ x_{j} + y_{j} + v_{j} + w_{j} \leq 1 $ for any $ j \in \{1,\ldots,n\} $. Throughout the paper, we prefer the product form over the product form due to its utility in clean derivations. By the facts $ {\sigma^{-}}^{\dg} = \sigma^{+}$, $ {\sigma^{0}}^{\dg} = \sigma^{0} $, and ${\sigma^{1}}^{\dg} = \sigma^{1} $ as well as the distribution of the adjoint operator over tensor products, we have:
\begin{align}\label{eq:conjsymcond}
T^{\dg}(\vec{x}, \vec{y}, \vec{v}, \vec{w}) = T(\vec{x}, \vec{y}, \vec{w}, \vec{v}) .
\end{align}

For example, given $\z = (0,0,0,0), \e_{1} = (1,0,0,0), \e_{3} = (0,0,1,0), \e_{4} = (0,0,0,1) $ we have:
\begin{align}
    T(e_{1}, \z, e_{3},e_{4}) = \underbrace{\sigma^{0}}_{1} \otimes \underbrace{\mathbbm{1}_{2}}_{2} \otimes \underbrace{\sigma^{+}}_{3} \otimes \underbrace{\sigma^{-}}_{4} = \sigma_1^0 \, \sigma_3^+ \, \sigma_4^-
\end{align}

This example term is non-Hermitian, but as described in Eq.\ref{eq:conjsymcond} the following is Hermitian:
\begin{align}
T(e_{1}, \z, e_{3},e_{4}) + T(e_{1}, \z, e_{4},e_{3}) &= \sigma_1^0 \, \sigma_3^+ \, \sigma_4^- + \sigma_1^0 \, \sigma_3^- \, \sigma_4^+ .
\end{align}

\subsection{A Single Constraint, its Feasible Subspaces, and their Embeddings}

Any polynomial constraint over binary variables $ x \in \{0,1\}^{n} $ can be written over selections of $x_{i}$ and $(1-x_{i})$:
\begin{defi}[Polynomial Constraint]
Over binary variables $x\in \{0,1\}^n$, a polynomial constraint has the form:
\begin{align}\label{eq:def_constraint}
C(x) = \sum_{ (\vec{a_{J}}, \vec{b_{J}}) \in \mathcal{J} } \beta_{J} \, \prod_{k=1}^{n} (1 - x_{k})^{a_{Jk}} \, x_{k}^{b_{Jk}},
\end{align}
where each $ (\vec{a_{J}}, \vec{b_{J}}) \in \mathcal{J} $ is a pair of orthogonal binary vectors satisfying $a_{Ji} + b_{Ji} \leq 1 $ for any $ i \in \{1, \ldots, n\} $ and each $ \beta_{J} $ is assumed nonzero.
\end{defi}

Under the standard bit to qubit mapping $ \ket{x} = \ket{ x_{1} } \ldots \ket{ x_{n} } \in \{ \ket{0}, \ket{1} \}^{\otimes n} $, we consider the natural embedding to be $ \hat{C} $ such that $ \bra{x} \hat{C} \ket{x} = C(x) = b $.
\begin{defi}[Embedded Constraint Operator]
Given a constraint written as \cref{eq:def_constraint}, the embedded constraint operator is:
\begin{align}\label{eq:defconop}
\hat{C} &= \sum_{ (\vec{a_{J}}, \vec{b_{J}}) \in \mathcal{J} } \beta_{J} \, \prod_{k=1}^{n} \left(\sigma_{k}^{0}\right)^{a_{Jk}} \left(\sigma_{k}^{1}\right)^{b_{Jk}} \nonumber \\
&= \sum_{ (\vec{a_{J}}, \vec{b_{J}}) \in \mathcal{J} } \beta_{J} \, T(\vec{a_{J}}, \vec{b_{J}}, \z, \z),
\end{align}
where every $ (\vec{a_{J}}, \vec{b_{J}}) \in \mathcal{J} $ are orthogonal binary vectors and each $ \beta_{J} $ is assumed nonzero.
\end{defi}

In particular, the embedded constraint operator has the constraint value for computational basis states as the constraint it embeds.
\begin{thm}[Eigenvalue Correspondence for Constraints]\label{thm:eig_cor_for_con}
Given a constraint $C(x)$ written as \cref{eq:def_constraint} and an embedded constraint operator $\hat{C}$ written as \cref{eq:defconop}, for any computational basis state $\ket{x}$ for $x \in \{0,1\}^n$, $ \bra{x} \hat{C} \ket{x} = C(x) $.
\end{thm}
\begin{proof}
\begin{align}
\bra{x} \hat{C} \ket{x} &= \sum_{ (\vec{a_{J}}, \vec{b_{J}}) \in \mathcal{J} } \beta_{J} \, \prod_{k=1}^{n} \bra{x} \left(\sigma_{k}^{0}\right)^{a_{Jk}} \left(\sigma_{k}^{1}\right)^{b_{Jk}} \ket{x} \nonumber \\
&= \sum_{ (\vec{a_{J}}, \vec{b_{J}}) \in \mathcal{J} } \beta_{J} \, \prod_{k=1}^{n} \bra{x_{k}} \left( \sigma^{0} \right)^{a_{Jk}} \left( \sigma^{1} \right)^{b_{Jk}} \ket{x_{k}} \nonumber \\
&= \sum_{ (\vec{a_{J}}, \vec{b_{J}}) \in \mathcal{J} } \beta_{J} \, \prod_{k=1}^{n} \Theta(1 - x_{k})^{a_{Jk}} \Theta(x_{k})^{b_{Jk}} \nonumber \\
&= \sum_{ (\vec{a_{J}}, \vec{b_{J}}) \in \mathcal{J} } \beta_{J} \, \prod_{k=1}^{n} (1 - x_{k})^{a_{Jk}} x_{k}^{b_{Jk}},
\end{align}
\end{proof}

where $ \Theta(x) $ is the Heaviside step function such that $ \Theta(x) = 1 $ if $ x > 0 $ and $ \Theta(x) = 0 $ if $ x \leq 0 $.

\quad For any polynomial constraint, $ C(x) $, we can find the corresponding eigenspace decomposition for different equality values $ b $ representing the associated feasibility space of the constraint problem with $ C(x) $ fixed to the value $ b $.
\begin{defi}[Feasible Subspace (single equality constraint and constraint value)]
Given a constraint $ C(x) $ and a constant $ b $, the associated feasibility space is
\begin{align}\label{eq:deffeasible}
F_{b} = \{ x \, : \, C(x) = b, \, x \in \{0,1\}^{n} \}.
\end{align}
\end{defi}

Then the associated projection operator of the embedded eigenspace (embedded feasibility subspace) assuming $ \hat{C} $ has the constraint value as the eigenvalue for each eigenspace:
\begin{defi}[Feasible Subspace Projector]
Given a feasible subspace $F_{b}$ written as \cref{eq:deffeasible}, the feasible subspace projector is:
\begin{align}\label{eq:defembedfeasible}
\hat{F}_{b} &= \sum_{x \in F_{b}} \ketbra{x}{x} .
\end{align}
\end{defi}

\begin{figure}
    \centering
    \begin{tabular}{c c}
    \includegraphics[width=0.48\linewidth]{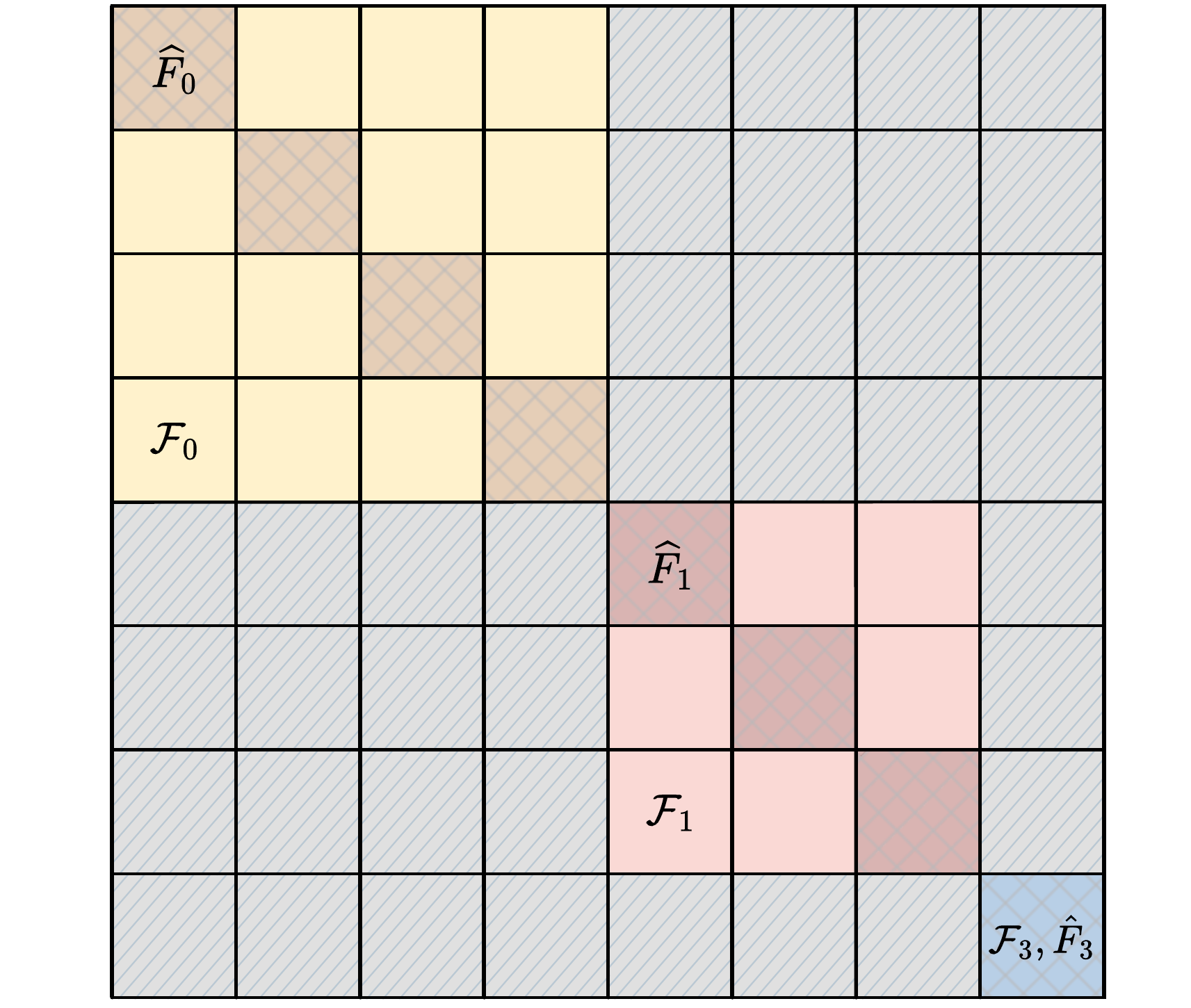} & \includegraphics[width=0.48\linewidth]{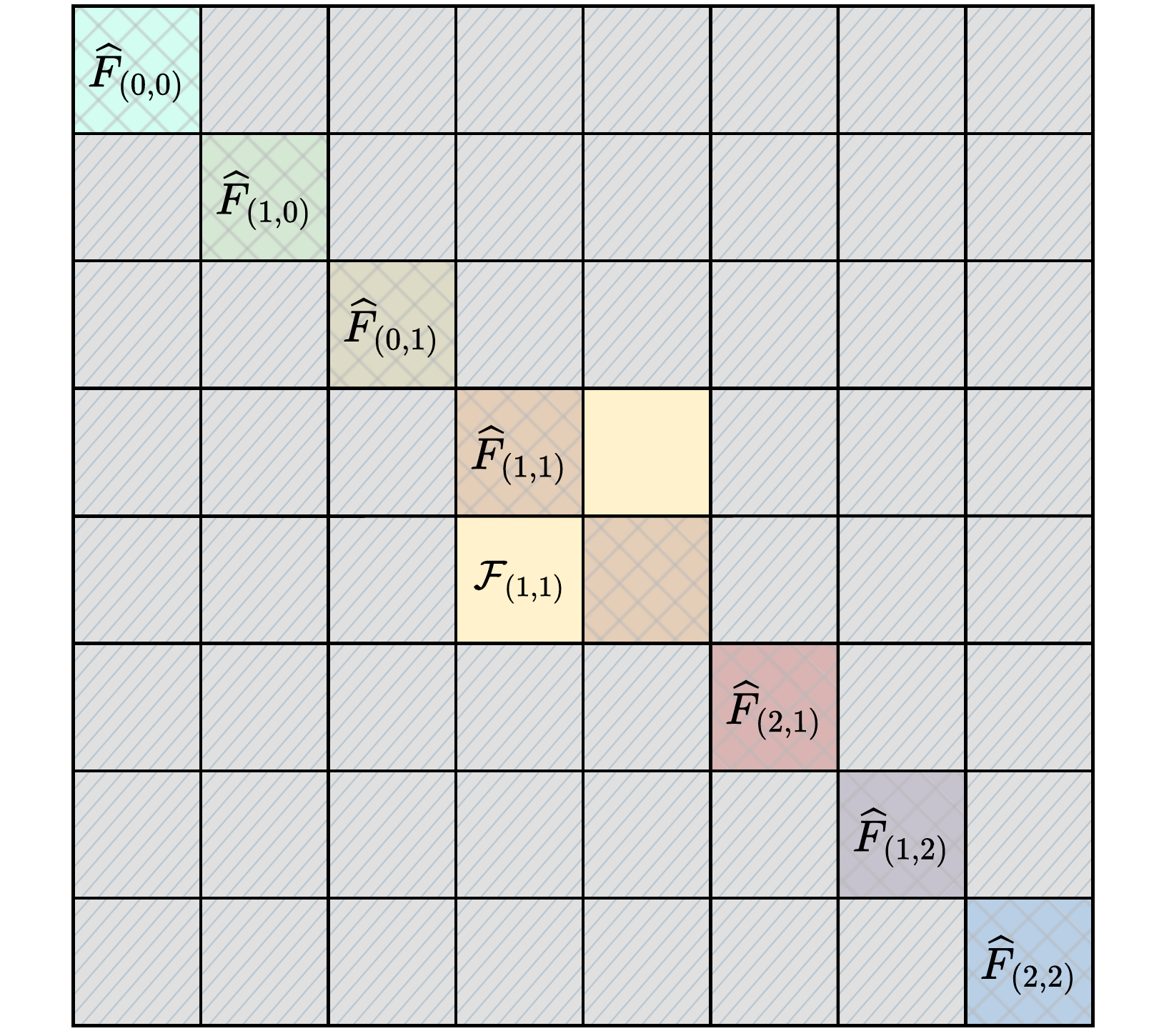} \\
    (a) & (b)
    \end{tabular}
    \caption{\textbf{Constraint's Feasibility Decomposition and Commutation Space.} Two example decompositions, each visualizing the embedded constraint operator down the diagonal and the associated commutation space. (a) visualizes for a constraint $C = x_{1} x_{2} + x_{2} x_{3} + x_{1} x_{3}$ leading to $ \hat{C} = \sigma_{1}^{1} \sigma_{2}^{1} + \sigma_{2}^{1} \sigma_{3}^{1} + \sigma_{1}^{1} \sigma_{3}^{1} = 0 \, \hat{F}_{0} + 1 \, \hat{F}_{1} + 3 \, \hat{F}_{3} $. Then the associated commutation space is $\mathcal{F}_{C} = \mathcal{F}_{0} \oplus \mathcal{F}_{1} \oplus \mathcal{F}_{3} $. (b) visualizes for two constraints $C_1 = x_1 + x_2$, $C_2 = x_2 + x_3$ with $ \hat{C}_{1} = \sigma_{1}^{1} + \sigma_{2}^{1} = 0 \, \hat{F}_{0}^{1} + 1 \, \hat{F}_{1}^{1} + 2 \, \hat{F}_{2}^{1} $, $ \hat{C}_{2} = \sigma_{2}^{1} + \sigma_{3}^{1} = 0 \, \hat{F}_{0}^{2} + 1 \, \hat{F}_{1}^{2} + 2 \, \hat{F}_{2}^{2} $ that has the joint spectrum $ (0,0) \, \hat{F}_{(0,0)} + (0,1) \, \hat{F}_{(0,1)} + (1,0) \, \hat{F}_{(1,0)} + (1,1) \, \hat{F}_{(1,1)} + (2,1) \, \hat{F}_{(2,1)} + (1,2) \, \hat{F}_{(1,2)} + (2,2) \, \hat{F}_{(2,2)} $ associated with the commutation space $\mathcal{F}_{C_1,C_2} = \mathcal{F}_{(0,0)} \oplus \mathcal{F}_{(1,0)} \oplus \mathcal{F}_{(0,1)} \oplus \mathcal{F}_{(1,1)} \oplus \mathcal{F}_{(2,1)} \oplus \mathcal{F}_{(1,2)} \oplus \mathcal{F}_{(2,2)} $. For readability we suppress $\mathcal{F}_{(b_1,b_2)} $ for (b) when it is equivalent to $\text{span}\left( \hat{F}_{(b_1,b_2)} \right)$.}
    \label{fig:constraint_decomp}
\end{figure}

\cref{eq:defconop} is the natural embedding since the correspondents between \cref{eq:deffeasible} and \cref{eq:defembedfeasible} is that $ \bra{x} \hat{C} \ket{x} = b $ for $ x $ in $ F_{b} $ precisely for this operator. While this is the natural embedding, other embeddings have been studied more in part due to the prevalence of Ising spin formalisms~\cite{lucas2014ising,hen2016driver,hen2016quantum}.

\quad A prototypical, widely studied case is that of a single global constraint $ C(x) = \sum_{i}^{n} x_{i} = b $ for some $ b $. This constraint is utilized in graph partitioning and portfolio optimization among others. Then the natural embedding constraint operator is $ \hat{C} = \sum_{i}^{n} \sigma_{i}^{1} = b $. Other embeddings would be $ \hat{C} = \sum_{i}^{n} \sigma_{i}^{0} = n - b $ and $ \hat{C} = \sum_{i}^{n} \sigma_{i}^{0} - \sigma_{i}^{1} = \sum_{i}^{n} \sigma_{i}^{z} = n - 2 \, b $. Such embeddings are equivalent in their eigenspaces, but not their corresponding eigenvalues, so finding commutative terms for one is associated with the same (eigenspace) invariance as the other embedding, but the expected value of the associated observables differ.

\quad An example of this type of mapping with linear constraints is discussed in \cref{sec:qaoa_1in3} and appears in the early Ref.~\cite{farhi2000quantum} (among many other papers) in the context of quantum annealing. We assume that $ \hat{C} $ is always simplified to use the minimum number of non-zero entries (e.g. simplify cases where an identity can replace two like terms that differ only by having $ \sigma^{0} $ versus $ \sigma^{1} $ on some qubit shared on an entry).

\begin{thm}[Spectral Decomposition of Embedded Constraint over Feasible Subspaces]\label{thm:spectrum_embedded_constraint}
Given an embedded constraint operator $ \hat{C}(x) $ as written in \cref{eq:defconop} with coefficients $\beta \in \R^{|\mathcal{J}|}$, it can be spectrally decomposed as:
\begin{align}
\hat{C} = \sum_{b = -|\beta|}^{|\beta|} b \, \hat{F}_{b}.
\end{align}
We often drop the limits for $b$, noting that for typical problems the range of $b$ will be more limited, and we assume null operators (associated with the empty set as a feasible subspace) are pruned.
\end{thm}

\begin{proof}
Suppose $ x \in F_{b} $, then $ \bra{x} \hat{C} \ket{x} = C(x) = b $. Since the embedded constraint operator is diagonal over the computational basis:
\begin{align}
\hat{C} &= \sum_{x \in \{0,1\}^n} C(x) \ketbra{x}{x} \nonumber \\
&= \sum_{b=-\infty}^{\infty} \sum_{x \in F_{b}} b \ketbra{x}{x} \nonumber \\
&= \sum_{b = -|\beta|}^{|\beta|} b \, \hat{F}_{b}.
\end{align}
In the last step, we recognized that there are no feasible solutions beyond the (generic) lower bound $-|\beta|$ and upper bound $|\beta| = \sum_{J} | \beta_{J} | $ and so $ F_{b} = \{\} $ for $ b < |\beta| $ and $ b > |\beta| $.
\end{proof}

\subsection{From Singular Constraints to Collection of Constraints}

In the case of multiple constraints, we can utilize the fact that each embedded constraint operator commutes to formulate a joint measurement with the associated energy being a tuple of each constraint energy based on the product of feasible subspace projectors.

\begin{defi}[Feasibility Space (Multiple Constraints)]
$ F_{(b_1,\ldots,b_m)} = \left\{ x : \bigwedge_{j=1}^{m} C_{j}(x) = b_{j} \right\} = \bigcap_{j=1}^{m} F_{b_{j}}^{j} . $
\end{defi}

\begin{thm}[Feasibility Subspace Projector (Multiple Constraints)]
$ \hat{F}_{(b_1,\ldots,b_m)} = \prod_{j=1}^{m} \hat{F}_{b_{j}}^{j} . $
\end{thm}
\begin{proof}
\begin{align}
\hat{F}_{(b_1,\ldots,b_m)} &= \prod_{k=1}^{m} \hat{F}_{b_{k}}^{k}  \nonumber \\
&= \prod_{k=1}^{m} \sum_{x \in F_{b_{k}^{k}}} \ketbra{x}{x} \nonumber \\
&= \left( \sum_{x \in F_{b_{1}}^{1}, y \in F_{b_{2}^{2}}} \ketbra{x}{x} \ketbra{y}{y} \right) \prod_{k=3}^{m} \hat{F}_{b_{k}}^{k} \nonumber \\
&= \left( \sum_{x \in F_{b_{1}^{1}} \cap F_{b_{2}^{2}}} \ketbra{x}{x} \right) \prod_{k=3}^{m} \hat{F}_{b_{k}}^{k} \nonumber \\
&\vdots \nonumber \\
&= \sum_{x \in F_{(b_{1}, \ldots, b_{m})}} \ketbra{x}{x} .
\end{align}
\end{proof}

\begin{defi}[Joint Spectral Decomposition of Multiple Constraints Operator]
Let $ \mathcal{C} = \{ C_{1}, \ldots, C_{m} \} $, and $ \hat{C}_{i} = \sum_{b_{i}} b_{i} \, \hat{F}_{b_{i}}^{i} $, then we have the following \textit{joint} spectral decomposition:
\begin{align}
\hat{\mathcal{C}} = \sum_{b_1,\ldots,b_m} (b_1,\ldots, b_m) \hat{F}_{b_1}^{1} \hat{F}_{b_2}^{2} \cdots \hat{F}_{b_m}^{m} .
\end{align}
Notice that $\hat{\mathcal{C}}$ has a vector valued observable, since each $\hat{C}_{i}$ commutes (each is diagonal in the computational basis), and so this would represent a \textit{joint} (sequential) measurement.
\end{defi}

\subsection{The Commutation Space of Constraints}

Linear operators that commute with $ \hat{C} $ are operators over the matrix span we call the \textit{commutation space}, which is a subspace of linear operators over $ \mathbbm{C}^{2^{n} \times 2^{n}} $.
\begin{defi}[Commutation Space of Feasible Subspace]
\begin{align}
\mathcal{F}_{b} = \text{span} \left( \left\{ \ketbra{x}{y} \, : \, x, y \in F_{b} \right\} \right),
\end{align}
\end{defi}

\begin{thm}\label{thm:in_one_comspace}
$ H \in \mathcal{F}_{b} $ iff $ \hat{F}_{b} H \hat{F}_{b} = H $.
\end{thm}

\begin{proof}
Suppose $ H \in \mathcal{F}_{b} $, then $ H = \sum_{xy} c_{xy} \, \ketbra{x}{y} $ and so
\begin{align}
\hat{F}_{b} \, H \, \hat{F}_{b} &= \sum_{xy} c_{xy} \, \left( \sum_{x} \ketbra{x}{x} \right) \ketbra{x}{y} \left( \sum_{y} \ketbra{y}{y} \right) \nonumber \\
&= \sum_{xy} \delta_{xx} \, \delta_{yy} \, c_{xy} \, \ketbra{x}{y} \nonumber \\
&= H.
\end{align}
And if not, then clearly the same steps show $ \hat{F}_{b} \, H \, \hat{F}_{b} \neq H $.
\end{proof}

\begin{defi}[Commutation Space of a Single Constraint]
\begin{align}
\mathcal{F}_{C} = \bigoplus_{b} \mathcal{F}_{b} .
\end{align}
\end{defi}

To decompose Hamiltonians over the commutation space of constraints, we utilize that left-right multiplication of projectors will project matrices into corresponding matrix subspaces.

\begin{lem}[Matrix Orthogonality over Projectors]\label{lem:mat_orth}
$ \text{Tr}\left( \left( \hat{F}_{a} \, H \, \hat{F}_{b} \right)^\dg \hat{F}_{c} \, H \, \hat{F}_{d} \right) = 0 $ unless $a=c$ and $b=d$.
\end{lem}

\begin{proof}
\begin{align}
\text{Tr}\left( \left(\hat{F}_{a} \, H \, \hat{F}_{b} \right)^\dg \hat{F}_{c} \, H \, \hat{F}_{d} \right) &= \text{Tr}\left( \hat{F}_{b}^\dg \, H^{\dg} \, \hat{F}_{a}^{\dg} \, \hat{F}_{c} \, H \, \hat{F}_{d} \right) \nonumber \\
&= \text{Tr}\left( \hat{F}_{b} \, H^{\dg} \, \hat{F}_{a} \, \hat{F}_{c} \, H \, \hat{F}_{d} \right) \nonumber \\
&= \text{Tr}\left( \hat{F}_{d} \, \hat{F}_{b} \, H^{\dg} \, \hat{F}_{a} \, \hat{F}_{c} \, H \right) \nonumber \\
&= \delta_{a,c} \, \delta_{b,d} \, \text{Tr}\left( \hat{F}_{b} \, H^{\dg} \, \hat{F}_{a} \, H \right),
\end{align}
due to distributivity of the conjugate transpose, Hermiticity of projectors, and orthogonality of projectors.
\end{proof}

A matrix is in the commutation space of a constraint if it commutes with the embedded constraint operator.

\begin{thm}[Invariance under a Single Constraint]\label{thm:invar_single_con}
$ \left[H, \hat{C}\right] = 0 $ iff $ H = \sum_{b} \hat{F}_{b} \, H \, \hat{F}_{b} $ iff $ H \in \mathcal{F}_{C} $.
\end{thm}
\begin{proof}
Suppose there exists $\hat{F}_{a} \, H \, \hat{F}_{b} \neq 0 $ for some $a \neq b$. Then:
\begin{align}
\left[\hat{F}_{a} \, H \, \hat{F}_{b}, \hat{C}\right] &= \sum_{c} \left[\hat{F}_{a} \, H \, \hat{F}_{b}, c \, \hat{F}_{c}\right]  \nonumber \\
&= \sum_{c} c \left(\hat{F}_{a} \, H \, \hat{F}_{b} \, \hat{F}_{c} - \hat{F}_{c} \, \hat{F}_{a} \, H \, \hat{F}_{b} \right) \nonumber \\
&= (b-a) \, \hat{F}_{a} \, H \, \hat{F}_{b} \nonumber \\
&\neq 0 .
\end{align}
Due to orthogonality of $ \hat{F}_{a} \, H \, \hat{F}_{b} $ with other possible terms in $ H $ (from \cref{lem:mat_orth}), $ H $ cannot be the null matrix. By the same logic, if no such $a \neq b$ exists, then $H = \sum_{b} \hat{F}_{b} \, H \, \hat{F}_{b} $.
Then for each $ b $, $\hat{F}_{b} \, H \, \hat{F}_{b} \in \mathcal{F}_{b}$ from \cref{thm:in_one_comspace} and so $ H \in \mathcal{F}_{C} $.
\end{proof}

The commutation space for multiple constraints is then defined as the direct sum of the commutation space of joint feasible subspaces.

\begin{defi}[Commutation Space of Joint Feasible Subspaces]
\begin{align}
\mathcal{F}_{b_{1},\ldots,b_{m}} = \text{span}\left( \left\{ \ketbra{x}{y} \, :  x, y \in \bigcap_{j=1}^{m} F_{b_j}^{j} \right\} \right)
\end{align}
\end{defi}

\begin{defi}[Commutation Space of Multiple Constraints]
\begin{align}
\mathcal{F}_{C_1,\ldots,C_m} = \bigoplus_{(b_1,\ldots,b_m)} \mathcal{F}_{b_1,\ldots,b_m} .
\end{align}
\end{defi}

Then a matrix belongs to the commutation space if it commutes with each constraint.

\begin{thm}[Invariance under Multiple Constraints]\label{thm:invar_multi_con}
$ \left[H, \hat{C}_{j}\right] = 0 $ for each $ j \in [1,m] $ iff $ H = \sum_{(b_1,\ldots,b_m)} \hat{F}_{(b_1,\ldots,b_m)} \, H \, \hat{F}_{(b_1,\ldots,b_m)} $ iff $ H \in \mathcal{F}_{C_1,\ldots,C_m} $.
\end{thm}

\begin{proof}
Suppose there exists $\vec{a} = (a_1, \ldots, a_m), \vec{b} = (b_1, \ldots, b_m)$ such that $ a_{j} \neq b_{j} $ for some index $j$ and $ \prod_{k=1}^{m} \hat{F}_{a_{k}}^{k} \, H \, \prod_{l=1}^{m} \hat{F}_{b_{l}}^{l} \neq 0 $. From \cref{thm:spectrum_embedded_constraint}, $ \hat{C}_{j} = \sum_{c_{jr}} c_{jr} \, \hat{F}_{c_{jr}}^{j} $, and so:
\begin{align}
\left[\prod_{k=1}^{m} \hat{F}_{a_{k}}^{k} \, H \, \prod_{l=1}^{m} \hat{F}_{b_{l}}^{l}, \hat{C}_{j}\right] &= \sum_{c_{jr}} \left[\prod_{k=1}^{m} \hat{F}_{a_{k}}^{k} \, H \, \prod_{l=1}^{m} \hat{F}_{b_{l}}^{l}, c_{jr} \, \hat{F}_{c_{jr}}^{j} \right] \nonumber \\
&= \sum_{c_{jr}} c_{jr} \left( \prod_{k=1}^{m} \hat{F}_{a_{k}}^{k} \, H \, \prod_{l=1}^{m} \hat{F}_{b_{l}}^{l} \hat{F}_{c_{jr}}^{j} - \hat{F}_{c_{jr}}^{j} \prod_{k=1}^{m} \hat{F}_{a_{k}}^{k} \, H \, \prod_{l=1}^{m} \hat{F}_{b_{l}}^{l}
\right) \nonumber \\
&= c_{jr} (b_j - a_j) \, \prod_{k=1}^{m} \hat{F}_{a_{k}}^{k} \, H \, \prod_{l=1}^{m} \hat{F}_{b_{l}}^{l} \nonumber \\
&\neq 0.
\end{align}
Note that we can change the order of the projectors since they commute because they are all diagonal in the computational basis. From the same logic if no such $\vec{a}, \vec{b}$ exist, $ H $ commutes with each $ \hat{C}_{j} $. Then each $ \hat{F}_{(b_1,\ldots, b_m)} \, H  \, \hat{F}_{(b_1,\ldots, b_m)} \in \mathcal{F}_{(b_1,\ldots, b_m)} $ and so $H \in \mathcal{F}_{C_1,\ldots,C_m} $ and vice versa.
\end{proof}

\quad This includes all diagonal matrices, i.e. $ \left[H, \hat{C}\right] = 0 $ if $ H $ has no nonzero off-diagonal element over the computational basis and so we say $ H $ commutes nontrivially with $ \hat{C} $ if $ H $ has at least one off-diagonal term.

\begin{defi}[Nontrivial Commutation]
We say a Hamiltonian $ H $ commutes nontrivially with a collection of embedded constraint operators $ \hat{C}_{i} $ if $ [H,\hat{C}_{i}] = 0 $ \textit{and} $ H $ has at least one off-diagonal term.
\end{defi}

See \cref{fig:constraint_decomp} for two visual examples of constraint embeddings and their associated commutation spaces.

\subsection{An Algebraic Condition for Commuting with an Embedded Constraint Operator}

Any Hermitian matrix $ H $ over $ \mathbbm{C}^{2^{n}} $ can be written over the overcomplete basis $ \{ \mathbbm{1}, \sigma^{0}, \sigma^{1}, \sigma^{+}, \sigma^{-} \}^{\bigotimes n} $ as:
\begin{align}\label{eq:defgenham}
H &= \sum_{ ( \vec{x_{j}}, \vec{y_{j}}, \vec{v_j}, \vec{w_j} ) \in \Delta( \mathscr{X}, \; \mathscr{Y}, \; \mathscr{V}, \; \mathscr{W} ) } \alpha_{j} \; T(\vec{x_j}, \vec{y_{j}}, \vec{v_{j}}, \vec{w_{j}}) + \alpha_{j}^{\dg} \; T(\vec{x_j}, \vec{y_{j}}, \vec{w_{j}}, \vec{v_{j}}),
\end{align}

following a similar convention as Ref.~\cite{leipold2021constructing}, where $ \Delta(\mathscr{X}, \mathscr{Y}, \mathscr{V}, \mathscr{W}) $ represents the set of index-wise confederated tuples. On a matrix $ H $, $ H^{\dg} $ is the complex conjugate transpose of $ H $ and so in \cref{eq:defgenham} $, \alpha_{j}^{\dg} $ is the complex conjugate of a scalar $ \alpha \in \mathbbm{C} $. By utilizing \cref{eq:conjsymcond}, notice that $ H $ is indeed Hermitian:
\begin{align}
H^{\dg} &= \sum_{ ( \vec{x_{j}}, \vec{y_{j}}, \vec{v_j}, \vec{w_j} ) \in \Delta( \mathscr{X}, \; \mathscr{Y}, \; \mathscr{V}, \; \mathscr{W} ) } \alpha_{j}^{\dg} \; T^{\dg}(\vec{x_j}, \vec{y_{j}}, \vec{v_{j}}, \vec{w_{j}}) + \alpha_{j} \; T^{\dg}(\vec{x_j}, \vec{y_{j}}, \vec{w_{j}}, \vec{v_{j}}) \nonumber \\
&= \sum_{ ( \vec{x_{j}}, \vec{y_{j}}, \vec{v_j}, \vec{w_j} ) \in \Delta( \mathscr{X}, \; \mathscr{Y}, \; \mathscr{V}, \; \mathscr{W} ) } \alpha_{j}^{\dg} \; T(\vec{x_j}, \vec{y_{j}}, \vec{w_{j}}, \vec{v_{j}}) + \alpha_{j} \; T(\vec{x_j}, \vec{y_{j}}, \vec{v_{j}}, \vec{w_{j}}) \nonumber \\
&= H.
\end{align}

As an example of this formalism, consider the Hermitian matrix:
\begin{align}
\sigma_{1}^{0} \sigma_{2}^{1} \left( \sigma_{3}^{+} \sigma_{4}^{-} + \sigma_{3}^{-} \sigma_{4}^{+} \right)
&= T((1,0,0,0),(0,1,0,0),(0,0,1,0),(0,0,0,1)) \nonumber \\
&\phantom{= } + T((1,0,0,0),(0,1,0,0),(0,0,0,1),(0,0,1,0)),
\end{align}
with ordered sets $ \mathscr{X} = \{ (1,0,0,0) \} $, $ \mathscr{Y} = \{ (0,1,0,0) \} $, $ \mathscr{V} = \{ (0,0,1,0) \} $, $ \mathscr{W} = \{ (0,0,0,1) \} $. See \cref{fig:basis_depict} for how embedded constraint operators and their commutators are represented by the overcomplete basis $\{ \mathbbm{1}_{2}, \sigma^{0}, \sigma^1, \sigma^{+}, \sigma^{-} \} $. Further simple examples that are helpful for understanding this formalism are found at the end of \cref{sec:jordan} and the beginning of \cref{sec:alg_driver_com}.

\quad Although somewhat tedious, this formalism representing the choice function on each qubit for each term over the chosen (overcomplete) basis is of practical usefulness for deriving algebraic forms to reason about imposing symmetries. In the form we will consider throughout the manuscript, redundant forms are removed or reduced before consideration (a similar discussion appears in Ref.~\cite{leipold2021constructing}) by removing any terms with $\alpha_{j} = \alpha_{j}^{\dg} = 0 $ and reducing $ \mathscr{X}, \mathscr{Y} $ as much as possible if the identity term can be utilized.

\begin{figure}[t]
\centering
\includegraphics[width=1.0\textwidth]{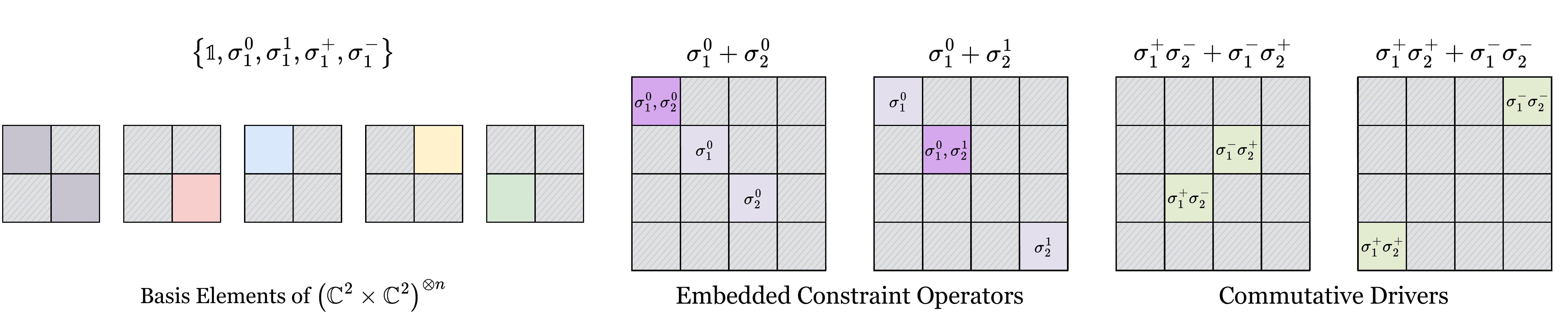}
\caption{\textbf{Embedded Constraint Operators and their Commutators.} Utilizing the frame $\{\mathbbm{1}, \sigma^{0}, \sigma^{1}, \sigma^{+}, \sigma^{-}\}^{\otimes n}$, depicted on the left, we can embed any polynomial constraint as an embedded constraint operator, depicted in the middle, and then valid driving or mixing generators commute with these constraints, depicted on the right. On the right, we depicted the unique commutator up to coefficient, with $ \left[ \alpha \, \sigma_{1}^{+} \, \sigma_{2}^{-} + \alpha^{\dg} \sigma_{1}^{-} \, \sigma_{2}^{+}, \sigma_{1}^{0} + \sigma_{2}^{0} \right] = 0 $ and $ \left[ \alpha \, \sigma_{1}^{+} \, \sigma_{2}^{+} + \alpha^{\dg} \, \sigma_{1}^{-} \, \sigma_{2}^{-}, \sigma_{1}^{0} + \sigma_{2}^{1} \right] = 0 $ for any $\alpha \in \C$.}
\label{fig:basis_depict}
\end{figure}

\quad Using this basis, we state a theorem for recognizing if a Hermitian matrix commutes with a polynomial constraint.

\begin{minipage}{\textwidth}
\begin{thm}[Algebraic Condition for Commutativity]\label{thm:gencom}
A general Hermitian matrix $ H(\mathscr{X}, \mathscr{Y}, \mathscr{V}, \mathscr{W}) $ commutes with a general constraint $ \hat{C}(\mathcal{J}) $ iff, for all $ \vec{x}, \vec{y}, \vec{v}, \vec{w} \in \Delta(\mathscr{X}, \mathscr{Y}, \mathscr{V}, \mathscr{W}) $:
\begin{adjustwidth}{-2.5em}{-2.5em}
\begin{align}
&\sum_{(\vec{a_{J}}, \vec{b_{J}}) \in \mathcal{J}} \beta_{J} \bigg( \left( 1 - \Theta((\vec{y_{j}}+\vec{w_{j}})\cdot \vec{a_{J}} + (\vec{x_{j}}+\vec{v_{j}})\cdot \vec{b_{J}}) \right) \nonumber \\
&\phantom{\sum} - \left( 1 - \Theta((\vec{y_{j}}+\vec{v_{j}})\cdot \vec{a_{J}} + (\vec{x_{j}}+\vec{w_{j}})\cdot \vec{b_{J}}) \right) \bigg) T(\vec{x_{j}} + (\mathbf{1}-\vec{v_{j}}-\vec{x_j}) \circ \vec{a_{J}}, \vec{y_{j}} + (\mathbf{1}-\vec{w_{j}}-\vec{y_{j}}) \circ \vec{b_{J}}, \vec{v_{j}}, \vec{w_{j}}) \nonumber \\
&+ \sum_{(\vec{a_{J}}, \vec{b_{J}}) \in \mathcal{J}} \beta_{J} \bigg( \left( 1 - \Theta((\vec{y_{j}}+\vec{v_{j}})\cdot \vec{a_{J}} + (\vec{x_{j}}+\vec{w_{j}})\cdot \vec{b_{J}}) \right) \nonumber \\
&\phantom{=} - \left( 1 - \Theta((\vec{y_{j}}+\vec{v_{j}})\cdot \vec{a_{J}} + (\vec{x_{j}}+\vec{v_{j}})\cdot \vec{b_{J}}) \right) \bigg) T(\vec{x_{j}} + (\mathbf{1}-\vec{w_{j}}-\vec{x_{j}}) \circ \vec{a_{J}}, \vec{y_{j}} + (\mathbf{1}-\vec{v_{j}}-\vec{y_{j}}) \circ \vec{b_{J}}, \vec{w_{j}}, \vec{v_{j}}) = 0, \label{eq:thmcom}
\end{align}
\end{adjustwidth}
where $ \Theta $ is the Heaviside step function, $ \mathbf{1} $ is the all ones vector $ (1,1, \ldots, 1) $ and $ \circ $ is the element-wise (Hadamard) product between vectors: $ \vec{a} \circ \vec{b} = (a_{1} b_{1}, \ldots, a_{n} b_{n} ) $.
\end{thm}
\end{minipage}

\begin{proof}
\quad \cref{thm:gencom} is built on a commutation relationship of general interest:
\begin{adjustwidth}{-4em}{-3em}
\centering
\begin{align}\label{eq:hamconcom}
    \left[ H, \hat{C} \right] &= \left[ \sum_{ ( \vec{x_{j}}, \vec{y_{j}}, \vec{v_j}, \vec{w_j} ) \in \Delta( \mathscr{X}, \; \mathscr{Y}, \; \mathscr{V}, \; \mathscr{W} ) } \alpha_{j} \; T(\vec{x_j}, \vec{y_{j}}, \vec{v_{j}}, \vec{w_{j}}) + \alpha_{j}^{\dg} \; T(\vec{x_j}, \vec{y_{j}}, \vec{w_{j}}, \vec{v_{j}}), \sum_{ (a_{J}, b_{J}) \in \mathcal{J} } \beta_{J} \, T(\vec{a_{J}}, \vec{b_{J}}, 0, 0) \right] \nonumber \\
    &= \sum_{} \sum_{} \alpha_{j} \, \beta_{J} \,  \bigg( \left( 1 - \Theta((\vec{y_{j}}+\vec{w_{j}})\cdot \vec{a_{J}} + (\vec{x_{j}}+\vec{v_{j}})\cdot \vec{b_{J}}) \right) \nonumber \\
    &\phantom{= \sum} - \left( 1 - \Theta((\vec{y_{j}}+\vec{v_{j}})\cdot \vec{a_{J}} + (\vec{x_{j}}+\vec{w_{j}})\cdot \vec{b_{J}}) \right) \bigg) T(\vec{x_{j}} + (\mathbf{1}-\vec{v_{j}} - \vec{x_{j}}) \circ \vec{a_{J}}, \vec{y_{j}} + (\mathbf{1} - \vec{w_{j}} - \vec{y_{j}}) \circ \vec{b_{J}}, \vec{v_{j}}, \vec{w_{j}}) \nonumber \\
    &\phantom{= } + \sum_{} \sum_{} \alpha_{j}^{\dg} \, \beta_{J} \, \bigg( \left( 1 - \Theta((\vec{y_{j}}+\vec{v_{j}})\cdot \vec{a_{J}} + (\vec{x_{j}}+\vec{w_{j}})\cdot \vec{b_{J}}) \right) \nonumber \\
    &\phantom{= \sum} - \left( 1 - \Theta((\vec{y_{j}}+\vec{w_{j}})\cdot \vec{a_{J}} + (\vec{x_{j}}+\vec{v_{j}})\cdot \vec{b_{J}}) \right) \bigg) T(\vec{x_{j}} + (\mathbf{1} - \vec{w_{j}} - \vec{x_{j}}) \circ \vec{a_{J}}, \vec{y_{j}} + (\mathbf{1} - \vec{v_{j}} - \vec{y_{j}}) \circ \vec{b_{J}}, \vec{w_{j}}, \vec{v_{j}}).
\end{align}
\end{adjustwidth}

\quad By recognizing the orthogonality of the basis (if $ H $ has nonorthogonal terms, they can be reduced to be orthogonal) and setting \cref{eq:hamconcom} $ = 0 $, the sufficient condition follows immediately. We show necessity first and then derive \cref{eq:hamconcom}.

\quad Right (left) commutation by a given matrix, $ \hat{C} $, is a linear transformation over matrices, since $ [ M_{1} + M_{2}, \hat{C} ] = [ M_{1}, \hat{C} ] + [ M_{2}, \hat{C} ] $ for any matrices $ M_{1}, M_{2} \in \mathbbm{C}^{2^{n} \times 2^{n}} $. For necessity of the condition, we show that the formalism can express the span of the kernel. Using the feasibility decomposition $ \hat{C} = \sum b \hat{F}_{b} $ from \cref{eq:defembedfeasible}, the kernel of right commutation with $ \hat{C} $ is spanned by operators that map feasible states to other feasible states. Let $ x, y \in F_{b} $, such that $ \bra{x} \hat{F}_{b} \ket{x} = \bra{y} \hat{F}_{b} \ket{y} = 1 $ for a specific $ b $, then:
\begin{align}
\ketbra{x}{y} &= \bigotimes_{i=1}^{n}  \big( \ketbra{x_{i}}{y_{i}} \big)  \nonumber \\
&= \prod_{i=1}^{n} \left( \sigma_{i}^{0} \right)^{(1-x_{i}) (1-y_{i})} \left( \sigma_{i}^{1} \right)^{x_{i} y_{i}} \left( \sigma_{i}^{+} \right)^{x_{i}(1-y_{i})} \left( \sigma_{i}^{-} \right)^{(1-x_{i}) y_{i}}.
\end{align}

\quad Then, $ \ketbra{x}{y} $ is a term in our formalism and the entire kernel (i.e. commutation space) can be spanned by such terms. In particular, there are $ | F_{b} |^2 $ such terms for a specific $ b $, by restricting to nontrivial commutators $ | F_{b} |^{2} - | F_{b} | $. This shows necessity.

\quad By using the overcomplete operator basis, we allow for lower locality weight terms by recognizing commonalities between different mappings (including between different constraint spaces). For example, suppose we consider the constraint $ C \equiv x_{1} = 1 $ over 2 bits $ \{0,1\}^{2} $ so that $ \hat{C} = \sigma_{1}^{1} $ over $ 2 $ qubits, such that $ F_{0} = \{ (0,0), (0,1) \} $ and $ F_{1} = \{ (1,0), (1,1) \} $ and therefore $ \hat{F}_{0} = \ketbra{00}{00} + \ketbra{01}{01} $ and $ \hat{F}_{1} = \ketbra{10}{10} + \ketbra{11}{11} $. Then $ \{ \sigma_{1}^{0} \sigma_{2}^{+}, \sigma_{1}^{0} \sigma_{2}^{-} \} $ are a basis for the nontrivial commutators restricted to $ F_{0} $ while $ \{ \sigma_{1}^{1} \sigma_{2}^{+}, \sigma_{1}^{1} \sigma_{2}^{-} \} $ are a basis for the nontrivial commutators restricted to $ F_{1} $. However, with the overcomplete basis that includes $ \mathbbm{1} $, we can utilize the more local forms $ \{ \sigma_{2}^{+}, \sigma_{2}^{-} \} $. As such, both $ \alpha \, \sigma_{2}^{+} + \alpha^{\dg} \, \sigma_{2}^{-} $ and $ \alpha \, \sigma_{1}^{1} \sigma_{2}^{+} + \alpha^{\dg} \, \sigma_{1}^{1} \sigma_{2}^{-} $ end up having the same relevant action for the feasibility subspace of interest (for any fixed $ \alpha \in \mathbbm{C} $).

\quad For sufficency, we show how to derive \cref{eq:hamconcom}. Beginning with the definitions \cref{eq:defconop} and \cref{eq:defgenham}:
\begin{adjustwidth}{-3em}{-3em}
\centering
\begin{align}
\left[ H, \hat{C} \right] &= \sum_{ ( \vec{x_{j}}, \vec{y_{j}}, \vec{v_j}, \vec{w_j} ) \in \Delta( \mathscr{X}, \; \mathscr{Y}, \; \mathscr{V}, \; \mathscr{W} ) } \left[ \alpha_{j} \, \prod_{i = 1}^{n} \left( \sigma_{i}^{0} \right)^{x_{ji}} \left( \sigma_{i}^{1} \right)^{y_{ji}} \left( \sigma_{i}^{+} \right)^{v_{ji}} \left( \sigma_{i}^{-} \right)^{w_{ji}}, \sum_{ (a_{J}, b_{J}) \in \mathcal{J} } \beta_{J} \, \prod_{k=1}^{n} \left(\sigma_{k}^{0}\right)^{a_{Jk}} \left(\sigma_{k}^{1}\right)^{b_{Jk}} \right] \nonumber \\
&\phantom{= } + \sum_{ ( \vec{x_{j}}, \vec{y_{j}}, \vec{v_j}, \vec{w_j} ) \in \Delta( \mathscr{X}, \;\mathscr{Y}, \; \mathscr{V}, \; \mathscr{W} ) }  \left[ \alpha_{j}^{\dg} \, \prod_{i = 1}^{n} \left( \sigma_{i}^{0} \right)^{x_{ji}} \left( \sigma_{i}^{1} \right)^{y_{ji}} \left( \sigma_{i}^{+} \right)^{w_{ji}} \left( \sigma_{i}^{-} \right)^{v_{ji}}, \sum_{ ( a_{J}, b_{J}) \in \mathcal{J} } \beta_{J} \, \prod_{k=1}^{n} \left(\sigma_{k}^{0}\right)^{a_{Jk}} \left(\sigma_{k}^{1}\right)^{b_{Jk}} \right] \nonumber \\
&= \sum_{ ( \vec{x_{j}}, \vec{y_{j}}, \vec{v_j}, \vec{w_j} ) \in \Delta( \mathscr{X}, \; \mathscr{Y}, \; \mathscr{V}, \; \mathscr{W} ) } \alpha_{j} \, \sum_{( a_{J}, b_{J}) \in \mathcal{J}} \beta_{J} \, L_{jK}(\vec{a_{J}}, \vec{b_{J}}, \vec{x_{j}},\vec{y_{j}},\vec{v_{j}},\vec{w_{j}}) \nonumber \\
&\phantom{= } + \sum_{ ( \vec{x_{j}}, \vec{y_{j}}, \vec{v_j}, \vec{w_j} ) \in \Delta( \mathscr{X}, \; \mathscr{Y}, \; \mathscr{V}, \; \mathscr{W} ) } \alpha_{j}^{\dg} \, \sum_{( a_{J}, b_{J}) \in \mathcal{J}} \beta_{J} \, R_{jK}(\vec{a_{J}}, \vec{b_{J}}, \vec{x_{j}},\vec{y_{j}},\vec{v_{j}},\vec{w_{j}})
\end{align}
\end{adjustwidth}

where:
\begin{adjustwidth}{-4em}{-4em}
\centering
\begin{align}
L_{jK}(\vec{a}, \vec{b}, \vec{x}, \vec{y}, \vec{v}, \vec{w})
&=  \prod_{k=1}^{n} \left( \sigma_{k}^{0} \right)^{x_{jk}} \left( \sigma_{k}^{1} \right)^{y_{jk}} \left( \sigma_{k}^{+} \right)^{v_{jk}} \left( \sigma_{k}^{-} \right)^{w_{jk}} \left( \sigma_{k}^{0} \right)^{a_{jk}} \left( \sigma_{k}^{1} \right)^{b_{Jk}}  \nonumber \\
&\phantom{= } - \prod_{k=1}^{n} \left( \sigma_{k}^{0} \right)^{a_{jk}} \left( \sigma_{k}^{1} \right)^{b_{Jk}} \left( \sigma_{k}^{0} \right)^{x_{jk}} \left( \sigma_{k}^{1} \right)^{y_{jk}} \left( \sigma_{k}^{+} \right)^{v_{jk}} \left( \sigma_{k}^{-} \right)^{w_{jk}} \nonumber \\
&=  \left( 1 - \Theta((\vec{y}+\vec{v})\cdot \vec{a} + (\vec{x}+\vec{w})\cdot \vec{b}) \right) T(\vec{x} + (\mathbf{1} - \vec{w} - \vec{x}) \circ \vec{a}, \vec{y} + (\mathbf{1} - \vec{v} - \vec{y}) \circ \vec{b}, \vec{v}, \vec{w}) \nonumber \\
&\phantom{= } - \left( 1 - \Theta((\vec{y}+\vec{w})\cdot \vec{a} + (\vec{x}+\vec{v})\cdot \vec{b}) \right) T(\vec{x} + (\mathbf{1} - \vec{v}-\vec{x}) \circ \vec{a}, \vec{y} + (\mathbf{1} - \vec{w} - \vec{y}) \circ \vec{b}, \vec{v}, \vec{w})
\end{align}
\end{adjustwidth}

and:
\begin{adjustwidth}{-4em}{-4em}
\centering
\begin{align}
    R_{jK}(\vec{a}, \vec{b}, \vec{x}, \vec{y}, \vec{v}, \vec{w}) &=  \prod_{k=1}^{n} \left( \sigma_{k}^{0} \right)^{x_{jk}} \left( \sigma_{k}^{1} \right)^{y_{jk}} \left( \sigma_{k}^{+} \right)^{w_{jk}} \left( \sigma_{k}^{-} \right)^{v_{jk}} \left( \sigma_{k}^{0} \right)^{a_{jk}} \left( \sigma_{k}^{1} \right)^{b_{Jk}}  \nonumber \\
    &\phantom{= } - \prod_{k=1}^{n} \left( \sigma_{k}^{0} \right)^{a_{jk}} \left( \sigma_{k}^{1} \right)^{b_{Jk}} \left( \sigma_{k}^{0} \right)^{x_{jk}} \left( \sigma_{k}^{1} \right)^{y_{jk}} \left( \sigma_{k}^{+} \right)^{w_{jk}} \left( \sigma_{k}^{-} \right)^{v_{jk}} \nonumber \\
    &=  \left( 1 - \Theta((\vec{y}+\vec{w})\cdot \vec{a} + (\vec{x}+\vec{v}) \cdot \vec{b}) \right)
    T(\vec{x} + (\mathbf{1} - \vec{v} - \vec{x}) \circ \vec{a}, \vec{y} + (\mathbf{1} - \vec{w} - \vec{y}) \circ \vec{b}, \vec{w}, \vec{v}) \nonumber \\
    &\phantom{= } - \left(1 - \Theta((\vec{y}+\vec{v}) \cdot \vec{a} + (\vec{x}+\vec{w}) \cdot \vec{b}) \right)
    T(\vec{x} + (\mathbf{1} - \vec{w} - \vec{x}) \circ \vec{a}, \vec{y} + (\mathbf{1} - \vec{v} - \vec{y}) \circ \vec{b}, \vec{w}, \vec{v}),
\end{align}
\end{adjustwidth}

with the $ \Theta(x) $ as the Heaviside step function, $\circ$ as the elementwise product, and $\mathbf{1} = (1,\ldots,1)$.

\end{proof}

\quad As a result, we immediately establish $ H $ is in the commutation space $\mathcal{F}_{C} $ of $ C $ from \cref{thm:invar_single_con} if and only if \cref{thm:gencom} holds.

\begin{cor}[Membership of Commutation Space (Single Constraint)]
$H \in \mathcal{F}_{C} $ iff \cref{thm:gencom} holds.
\end{cor}

For a collection of constraints $ \mathcal{C} $, if and only if \cref{thm:gencom} holds for each, $ H $ is in the commutation space of $ \mathcal{C} $ from \cref{thm:invar_multi_con}.

\begin{cor}[Membership of Commutation Space (Multiple Constraints)]
$ H \in \mathcal{F}_{C_1,\ldots,C_m} $ iff \cref{thm:gencom} holds for every $ C_{i} \in \mathcal{C} $.
\end{cor}

The locality hierarchy of commutators with a constraint collection can then be stated.

\begin{defi}[k-local Commutators of Multiple Constraints]\label{def:local_comterms}
Let
\begin{align}
\mathcal{T}_{k} = \{ T(\vec{x},\vec{y},\vec{v},\vec{w}) : [T, C_{j} ] = 0 \text{ for all } j \text{ and } |\vec{x}| + |\vec{y}| + |\vec{v}| + |\vec{w}| \leq k \}
\end{align}
be all the commutators up to locality $k$ in the basis \cref{eq:defi_basis_term} then the space of $k$-local commutators is $ \text{span}(\mathcal{T}_{k}) $.
\end{defi}

Most importantly, if we exhaustively find all such commutators up to locality $n$, the span saturates the commutation space - a direct corollary from \cref{thm:gencom}.

\begin{cor}[Local Commutator Hierarchy]
$ \mathcal{F}_{C_1,\ldots,C_m} = \text{span}\left( \mathcal{T}_{n} \right) $ and $ \text{span}\left( \mathcal{T}_{k} \right) \subseteq \text{span}\left( \mathcal{T}_{k+1} \right) $.
\end{cor}

For several constraints of interest that are higher order than linear, a simple strategy is to match parts of $\vec{y}$ and parts of $\vec{x}$ to overlap with the given $\vec{a_{J}}$ and $\vec{b_{J}}$. For example, for the maximum independent set (MIS) problem, the popular ansatz given in Ref.~\cite{hadfield2019quantum} has this form.

\subsection{Example: Maximum Independent Set}

The maximum independent set problem can be defined as a maximization problem with linear inequality constraints associated with a graph $ G = (V, E) $ over $ X = \{ x_{1}, \ldots, x_{n} \} $ and $ x_{i} \in \{ 0, 1 \} $ with $ x_{i} $ associated with vertex $ v_{i} \in V $ and $ |V| = n $:
\begin{align}
\text{maximize} &\; \sum_{i} x_{i} \nonumber \\
\text{subject to} &\; x_{i} + x_{j} \leq 1 \; \text{for} \; (v_{i}, v_{j}) \in E
\end{align}

Alternatively, it can be cast as a polynomial equality constraint problem and a minimization problem.
\begin{defi}[Maximum Independent Set]
\begin{align}
\text{minimize} &\; \sum_{i}^{n} (1 - x_{i}) \nonumber \\
\text{subject to} &\; \sum_{(v_{i}, v_{j}) \in E} x_{i} x_{j} = 0
\end{align}
\end{defi}

Then the embedded constraint operator can be written:
\begin{align}
\hat{C}_{\text{MIS}} = \sum_{(v_{i}, v_{j}) \in E} \sigma_{i}^{1} \sigma_{j}^{1} = 0,
\end{align}

while the cost Hamiltonian can be:
\begin{align}
H_{f,\text{MIS}} = \sum_{i=1}^{n} \sigma_{i}^{0}.
\end{align}

Define $ N(v_{i}) = \{ (v_{i}, v_{j}) \, : (v_{i}, v_{j}) \in E \} $ as the neighborhood edges of $ v_{i} $. Then a term
\begin{align}
T_{i, \text{MIS}} &= \sigma_{i}^{+} \prod_{(v_{i},v_{j}) \in N(v_{i})} \sigma_{j}^{0} + \sigma_{i}^{-} \prod_{(v_{i},v_{j}) \in N(v_{i})} \sigma_{j}^{0} \nonumber \\
&= \sigma_{i}^{x} \prod_{(v_{i},v_{j}) \in N(v_{i})} \sigma_{j}^{0},
\end{align}
commutes with the embedded equality constraint operators:
\begin{align}
\left[T_{i, \text{MIS}}, \hat{C}_{\text{MIS}}\right] &= \sum_{(v_{k}, v_{l}) \in E} \left[T_{i, \text{MIS}}, \sigma_{k}^{1} \sigma_{l}^{1}\right] \nonumber \\
&= \sum_{(v_{k}, v_{l}) \in E} \left[\sigma_{i}^{+} \prod_{(v_{i},v_{j}) \in N(v_{i})} \sigma_{j}^{0}, \sigma_{k}^{1} \sigma_{l}^{1}\right] + \left[\sigma_{i}^{-} \prod_{(v_{i},v_{j}) \in N(v_{i})} \sigma_{j}^{0}, \sigma_{k}^{1} \sigma_{l}^{1}\right] \nonumber \\
&= \sum_{(v_{k}, v_{l}) \in N(v_{i})} \left[\sigma_{i}^{+} \prod_{(v_{i},v_{j}) \in N(v_{i})} \sigma_{j}^{0}, \sigma_{k}^{1} \sigma_{l}^{1}\right] + \left[\sigma_{i}^{-} \prod_{(v_{i},v_{j}) \in N(v_{i})} \sigma_{j}^{0}, \sigma_{k}^{1} \sigma_{l}^{1}\right] \nonumber \\
&= 0.
\end{align}

Placing $ \sigma_{j}^{0} $ for each neighbor $ v_{j} $ of a node $ v_{i} $ for the term $ T_{i, \text{MIS}} $ means that $ \left[T_{i, \text{MIS}}, \sigma_{i}^{1} \sigma_{j}^{1}\right] = 0 $ because $ \sigma_{j}^{0} \sigma_{j}^{1} = \sigma_{j}^{1} \sigma_{j}^{0} =  0 $. In fact, $ \sum_{i=1}^{n} T_{i, \text{MIS}} $ is a standard driver for exploiting this symmetry of MIS. Moreover, this driver and its associated mixer provide irreducible commutation (see \cref{sec:complexity}). This example illustrates a simple means to satisfying higher order constraints: placing a single $ \sigma^{0} $ ($ \sigma^{1} $) term on a qubit where there is a term $ \sigma^{1} $ ($ \sigma^{0} $) in the constraint (provided at least one $ \sigma^{+} $ or $ \sigma^{-} $ has been placed to ensure off-diagonality). This motivates the formulation of a backtracking approach like Alg.~\ref{alg:find_com_terms} in \cref{sec:theory_to_prac}.

\quad Cast as a minimization problem, the cost Hamiltonian $ H_{f,\text{min}} $ leads to the phase-separation operator $ e^{-i \, \alpha \, \left( \sum_{i} \sigma^{0} \right) } = \prod_{i=1}^{n} e^{-i \, \alpha \, \sigma^{0} } $. As an initial state inside the constraint space, we can use the embedding of the empty independent set $ \ket{\psi(0)} = \ket{0} \ldots \ket{0} $ or an embedding of any valid independent set. \cref{fig:indset_exam} visualizes this for a simple triangle graph $ E = \{ (1,2), (2,3), (1,3) \} $.

\begin{figure}
\includegraphics[width=1.0\textwidth]{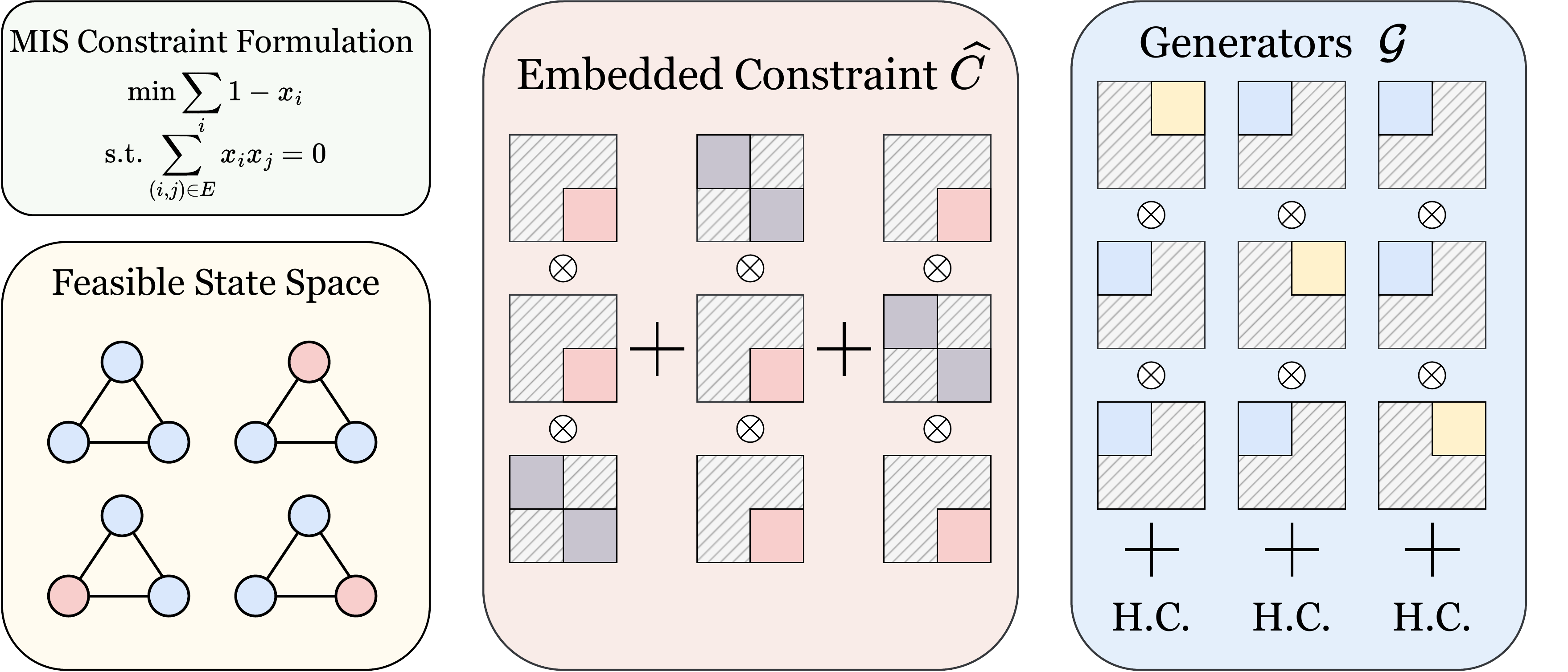}
\caption{\textbf{Imposing Set Independence.} By constructing drivers that include (or revert inclusion) of a node in the set based on no neighbors being in the set, we can explore the entire feasible state space of valid independent sets of a graph. Such drivers satisfy the algebraic condition in \cref{thm:gencom} and are the lowest locality drivers for the problem.}
\label{fig:indset_exam}
\end{figure}

\subsection{Wrap Up: Simplification for Linear Constraints and Sufficient Conditions for Ising Spin-z Constraints}

Note that the identity $ [AB, C] = A [B, C] + [A, C] B $ implies that commutation is preserved under multiplication and the identity $ [A + B, C] = [A, C] + [B, C] $ implies commutation is preserved under addition. Then finding a collection of terms $ \mathcal{T} $, each commutative with the collection of constraints of interest, can be recognized as the Jordan algebra of constraint preserving drivers (developed in \cref{sec:jordan}). In particular, we develop algorithms to find such a set $ \mathcal{T} $ that satisfies \cref{thm:gencom} for ansatz generation in \cref{sec:theory_to_prac}.

\quad In comparison to the expression derived in Appendix~\ref{app:suffquad}, the inclusion of $ \sigma^{0} $ and $ \sigma^{1} $ leads to a more inclusive way to find terms that commute to constraints, while the expressions of Appendix~\ref{app:suffquad} give a straight forward sufficient condition that generalizes the expression found in Ref.~\cite{leipold2021constructing} and are more useful in situations where the basis $ \{ \mathbbm{1}, \sigma^{z} \}^{\bigotimes n} $ is used, for example to express quantum Hamiltonians in the study of quantum materials or chemistry, such as mixers for VQE to handle Jordan-Wigner or Bravyi-Kitaev transformed operators to constrain the number of particles and symmetry in spin observables~\cite{ryabinkin2018constrained}.

\quad For completeness, we show how to recognize the main algebraic result in Ref.~\cite{leipold2021constructing} within this formalism. Consider a linear constraint $ C(x) = \sum_{j=1}^{n} c_{j} \, x_{j} $ that was embedded as $ \hat{C'} = \sum_{j=1}^{n} c_{j} \, \sigma_{j}^{z} = \sum_{j=1}^{n} c_{j} \, (2 \sigma_{j}^{0} - \mathbbm{1}) = \left( 2 \,  \sum_{j=1}^{n} \sigma_{j}^{0} \right) -\left( \sum_{j=1}^{n} c_{j} \right) \mathbbm{1} $. Since the identity operator is in the kernel of any commutation, we would use the embedded operator $ \hat{C} = \sum_{j=1}^{n} c_{j} \, \sigma_{j}^{0} $, and commutation with $ \hat{C} $ implies commutation with $ \hat{C'} $ and vice versa. Same follows for embeddings utilizing $ \sigma^{1} $ rather than $ \sigma^{0} $ (as either choice is valid) as $ \hat{C}'' = \sum_{j=1}^{n} c_{j} \, \sigma_{j}^{1} = \sum_{j=1}^{n} c_{j} \, \left( \mathbbm{1} - \sigma_{j}^{0} \right) = \left( \sum_{j=1}^{n} c_{j} \right) \mathbbm{1} - \left( \sum_{j=1}^{n} c_{j} \, \sigma_{j}^{0} \right) $. Then \cref{eq:thmcom} can be simplified as:
\begin{align}\label{eq:thmlincom}
\sum_{j=1}^{n} c_{j} \left( \Theta(v_{j}) - \Theta(w_{j}) \right) = \vec{c} \cdot \left( \vec{v} - \vec{w} \right) = 0,
\end{align}

since placing $ x_{j} \neq 0 $ or $ y_{j} \neq 0 $ is irrelevant for the commutation for any location and recognizing $ \mathcal{J} = \{ (\mathbf{e}_{1}, \z), \ldots, (\mathbf{e}_{n}, \z) \} $, which are one hot vectors for each index (assuming each $ c_{j} $ is non-zero, although this assumption can be dropped in the expression's final form). This allows us to state the algebraic condition for linear constraints found in Ref.~\cite{leipold2021constructing}:
\begin{thm}\label{thm:linconcom}
A Hermitian matrix $ H(\mathscr{X}, \mathscr{Y}, \mathscr{V}, \mathscr{W}) $ commutes with an embedded constraint operator $ \hat{C} = \sum_{j=1}^{n} c_{j} \, \sigma_{j}^{0} = \vec{c} \cdot \vec{\sigma^{0}} $ iff $ \vec{c} \cdot \left( \vec{v} - \vec{w} \right) = 0 $ for all $ (\vec{v}, \vec{w}) \in \Delta(\mathscr{V}, \mathscr{W}) $.
\end{thm}

\quad Striking in the case of linear constraints is that only $ \vec{v}, \vec{w} $ play any role in determining if a Hermitian matrix will commute with the associated embedded constraint operator. That is, $ \vec{x}, \vec{y} $ may be nonzero vectors for each entry of $ \Delta(\mathscr{X}, \mathscr{Y}, \mathscr{V}, \mathscr{W}) $, but have no impact on the resulting commutation. As in the general case, it is also striking that the implication holds for any $ \alpha \in \mathbbm{C} $.

\section{The Complexity of Imposing Constraints on Quantum Operators}\label{sec:complexity}

\begin{figure}[t]
\includegraphics[width=0.6\textwidth]{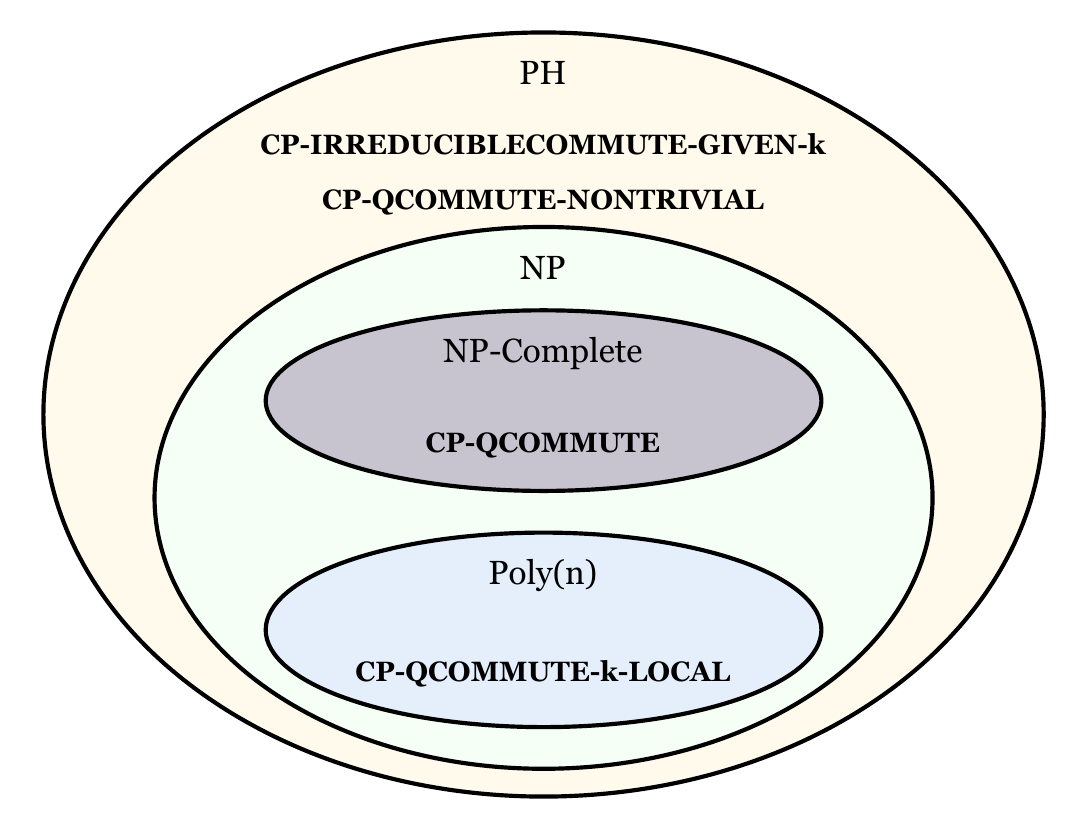}
\caption{\textbf{The Complexity of Imposing Constraints on Quantum Operators.} A Venn diagram showing the classification of important problems associated with imposing constraints, under the assumption that $NP\neq P$ and \texttt{CP-IRREDUCIBLECOMMUTE-GIVEN-k} is strictly higher in the polynomial hierarchy PH (which does not collapse).}
\label{fig:complex}
\end{figure}

In this section we turn our attention to the related computational tasks associated with finding a Hamiltonian that commutes with an arbitrary set of constraints. The definitions follow those discussed in Ref.~\cite{leipold2021constructing} and use the notation introduced in \cref{sec:algecond}.

\begin{defi}[CP-QCOMMUTE]
	Given a set $ \mathcal{C} = \{ C_{1}, \ldots, C_{m} \} $ of constraints such that $ \hat{C_{i}} = \sum_{(\vec{a_{J}}, \vec{b_{J}}) \in \mathcal{J}} \beta_{J} \prod_{k=1}^{n} \left( \sigma_{k}^{0} \right)^{a_{Jk}} \left( \sigma_{k}^{1} \right)^{b_{Jk}} $, over a space $ \mathbb{C}^{2^{n}} $ with $ \beta_{J} \in \Z $ and $ |\mathcal{J}| = \mathcal{O}\left(\text{poly}(n)\right) $, is there a Hermitian Matrix $ H $, with $ \mathcal{O} \left( \text{poly}( n) \right) $ nonzero coefficients over a basis $ \{ \chi_1, \chi_2, \chi_3, \chi_4 \}^{\bigotimes n} $, such that $ \left[ H, \hat{C_{i}} \right] = 0 $ for all $ \hat{C_{i}} $ and $ H $ has at least one off-diagonal term in the spin-z basis?
\end{defi}

\quad In the case where all $ \vec{a}_{J}, \vec{b}_{J} $ are one-hot vectors (such that they only have one nonzero entry), the constraints are linear and the corresponding version of \texttt{CP-QCOMMUTE} is called \texttt{ILP-QCOMMUTE}~\cite{leipold2021constructing}. Notice that $ \chi_{i} $ need not form a Hermitian basis with coefficients restricted to the reals, but $ H $ should be Hermitian over this general linear operator basis, which can be computed in polynomial time $ \texttt{P} $ given that $ H $ has $ \mathcal{O}(\text{poly}(n)) $ terms to conjugate transpose.

\begin{thm}\label{thm:qcommute_is_nphard}
    \texttt{CP-QCOMMUTE} is NP-Hard.
\end{thm}

\begin{proof}
    Follows immediately from the hardness of \texttt{ILP-QCOMMUTE}.
\end{proof}

\begin{thm}\label{thm:qcommute_is_npcomplete}
    \texttt{CP-QCOMMUTE} is NP-Complete.
\end{thm}

\begin{proof}
    To prove this, we rely on the algebraic conditions derived in \cref{sec:algecond}, specifically \cref{eq:thmcom} and App.~\ref{app:suffquad}. Checking commutation can be done in polynomial time through these formulas. Note it is possible to verify if a solution $ H $ has at least one off-diagonal basis term in polynomial time. Thus a solution can be verified in polynomial time.
\end{proof}

\begin{defi}[CP-QCOMMUTE-k-LOCAL]
	Given a set $ \mathcal{C} = \{ C_{1}, \ldots, C_{m} \} $ of constraints such that $ \hat{C_{i}} = \sum_{(\vec{a_{J}}, \vec{b_{J}}) \in \mathcal{J}} \beta_{J} \prod_{k=1}^{n} \left( \sigma_{k}^{0} \right)^{a_{Jk}} \left( \sigma_{k}^{1} \right)^{b_{Jk}} $, over a space $ \mathbb{C}^{2^{n}} $ with $ \beta_{J} \in \Z $, is there a Hermitian Matrix $ H $, with $ \mathscr{O} \left( \text{poly}( n) \right) $ nonzero coefficients over a basis $ \{ \chi_1, \chi_2, \chi_3, \chi_4 \}^{\bigotimes n} $, such that no term acts on more than $ k $ qubits, that $ \left[ H, \hat{C_{i}} \right] = 0 $ for all $ \hat{C_{i}} $, and $ H $ has at least one off-diagonal term in the spin-z basis?
\end{defi}

\begin{thm}\label{thm:klocal_is_poly}
    \texttt{CP-QCOMMUTE-k-LOCAL} is in polynomial time, \texttt{P}, for bounded locality weight $k$ in $ \mathcal{O}(1) $.
\end{thm}

\begin{proof}
    We present an algorithm utilizing \cref{thm:gencom}. At least one term must be offdiagonal and we have two choices, $ \sigma^{+} $ or $ \sigma^{-} $. Then we have $ 2^{k} \sum_{j=1}^{k-1} \binom{n}{j} 4^{j} = \mathcal{O}\left( 4^{k} n^{k} \right) $ possible terms to consider. In \cref{sec:theory_to_prac}, we consider an exact backtracking algorithm that can achieve an effective runtime significantly lower than this upper bound, depending on the instance.
\end{proof}

\quad Notice that, as in Ref.~\cite{leipold2021constructing}, the relationship between unitary operators and Hamiltonians means that these results translate directly to the task of designing mixing operators for Quantum Alternating Operator Ansatz algorithms.

\quad However, finding a Hamiltonian $ H $ by solving \texttt{CP-QCOMMUTE} does not mean the action of $ H $ is off-diagonal in a specific commutation space $ \mathcal{F}_{(b_1,\ldots,b_m)} $. In that particular subspace, $ H $ may act trivially or only introduce phasing. A more stringent condition requires finding $ H $ such that $ H $ is off-diagonal in a specified commutation space $ \mathcal{F}_{b_{1}, \ldots, b_{m}} $.

\begin{defi}[CP-QCOMMUTE-NONTRIVIAL]
    Given a set $ \mathcal{C} = \{ C_{1}, \ldots, C_{m} \} $ of polynomial constraints and constraint values $ b = \{ b_{1}, \ldots, b_{m} \} $ such that $ \hat{C}_{i} = \sum_{(\vec{a_{J}}, \vec{b_{J}}) \in \mathcal{J}} \beta_{J} \prod_{k=1}^{n} \left( \sigma_{k}^{0} \right)^{a_{Jk}} \left( \sigma_{k}^{1} \right)^{b_{Jk}} $, over a space $ \mathbb{C}^{2^{n}} $ with $ \beta_{J} \in \Z $, is there a Hermitian Matrix $ H $, with $ \mathscr{O} \left( \text{poly}( n) \right) $ nonzero coefficients over a basis $ \{ \chi_1, \chi_2, \chi_3, \chi_4 \}^{\bigotimes n} $, such that $ \left[ H, \hat{C}_{i} \right] = 0 $ for all $ \hat{C}_{i} $ and $ \hat{F}_{ b_{1} }^{1} \, \cdots \, \hat{F}_{b_{m}}^{m} H \hat{F}_{b_{m}}^{m} \, \cdots \, \hat{F}_{b_{1}}^{1} $ has at least one off-diagonal term in the spin-z basis?
\end{defi}

\quad The hardness of this problem follows directly from Ref.~\cite{leipold2021constructing}.

\begin{thm}\label{thm:nontrivial_is_nphard}
    \texttt{CP-QCOMMUTE-NONTRIVIAL} is NP-Hard.
\end{thm}

\begin{proof}
    Follows directly from the hardness of \texttt{ILP-QCOMMUTE-NONTRIVIAL}.
\end{proof}

\quad A related question is whether a collection of terms $ \mathcal{G} $ connects the space of feasible states, such that $ \bra{x} \left( \sum_{g \in \mathcal{G}} g \right)^{r} \ket{y} \neq 0 $ for some $ r $ given that $ x, y \in F_{(b_{1},\ldots,b_{m})} $ (check \cref{eq:defembedfeasible}). \texttt{CP-QCOMMUTE} is enough to guarantee a term is not harmful, i.e. for two states $ \ket{x}, \ket{y} $ such that $ x \in F_{(b_{1},\ldots, b_{m})} $ and $ y \in F_{(a_{1},\ldots, a_{m})} $ where $ \vec{b} \neq \vec{a} $, the terms found satisfy $ \bra{x} \left( \sum_{g \in \mathcal{G}} g \right)^{r} \ket{y} = 0 $ for all $ r $. Satisfying both conditions defines irreducible commutation \footnote{The resulting eigendecomposition of any Hamiltonian over the basis $ \mathcal{G} $ does not include any span of embedded subspaces of the feasibility space.}.
\begin{defi}[CP-QIRREDUCIBLECOMMUTE-GIVEN-k]
Given a set $ \mathcal{C} = \{ C_{1}, \ldots, C_{m} \} $ of linear constraints and constraint values $ b = \{ b_{1}, \ldots, b_{m} \} $ such that $ \hat{C}_{i} = \sum_{(\vec{a_{J}}, \vec{b_{J}}) \in \mathcal{J}} \beta_{J} \prod_{k=1}^{n} \left( \sigma_{k}^{0} \right)^{a_{Jk}} \left( \sigma_{k}^{1} \right)^{b_{Jk}} $, over a space $ \mathbb{C}^{2^{n}} $ with $ \beta_{J} \in \Z $, and a set of basis terms $ \mathcal{G} = \{ \hat{G}_{1}, \ldots, \hat{G}_{k} \} $ such that $ \hat{G}_{i} \in \{ \chi_1, \chi_2, \chi_3, \chi_4 \}^{\bigotimes n} $ and $\hat{G}_{i} $ commutes with each embedded constraint, does $ \mathcal{G} $ connect the entire nonzero eigenspace of the operator $ \hat{F}_{b_{1}}^{1} \cdots \hat{F}_{b_{m}}^{m} $?
\end{defi}

\begin{thm}\label{thm:givenk_is_nphard}
    \texttt{CP-QIRREDUCIBLECOMMUTE-GIVEN-k} is NP-Hard.
\end{thm}

\begin{proof}
    Follows directly from the restricted case over linear constraints, \texttt{ILP-QIRREDUCIBLECOMMUTE-GIVEN-k}~\cite{leipold2021constructing}.
\end{proof}

\quad In practice, we are interested in driver Hamiltonians that have bounded locality weight, so that they can be written as a sum of terms with each term acting on a constant or less number of qubits. In that case, solving the corresponding problems of finding the most terms possible that commute and act on a specific subspace is in \texttt{P} just as \texttt{CP-QCOMMUTE-k-LOCAL}. Assuming that local operators are sufficient to achieve irreducible commutation, as has been found for many practical symmetries of interest, the corresponding local version of \texttt{CP-QIRREDUCIBLECOMMUTE-GIVEN-k} is resolved since we know that all terms up to a fixed locality weight $ l $ can be found in \texttt{P}. \cref{fig:complex} depicts the complexity classifications discussed in this section.

\begin{figure}[t]
\includegraphics[width=0.8\textwidth]{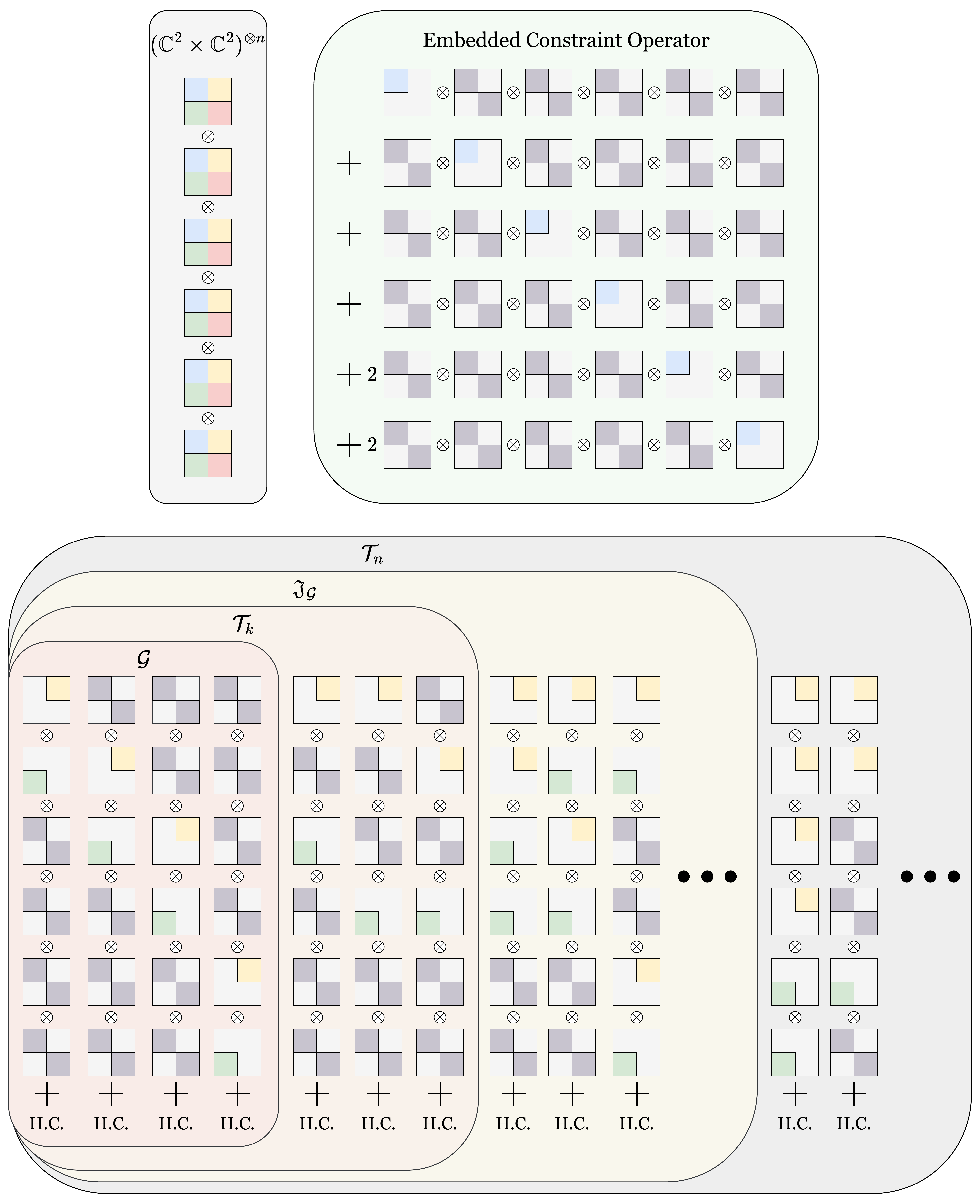}
\caption{\textbf{Hierarchy of Constraint Preserving Operators.} A visualization of an example embedded constraint operator $ \hat{C} = \sigma_{1}^{0} + \sigma_{2}^{0} + \sigma_{3}^{0} + \sigma_{4}^{0} + 2 \, \sigma_{5}^{0} + 2 \, \sigma_{6}^{0} $ and the associated collections of important commutator sets. Here $\mathcal{T}_2$ are all nontrivial commutators with $\hat{C}$ up to locality $k=2$. The generator set $\mathcal{G}$ is a subset of $\mathcal{T}_2$ that is sufficient to generate the Jordan algebra $\MG$ which has $\mathcal{T}_2$ as a subset. In \cref{sec:theory_to_prac}, we develop an efficient algorithm Alg.~\ref{alg:find_com_terms} to find a set $\mathcal{T}_k$, established to be polynomial time by \cref{thm:klocal_is_poly}, and then an algorithm Alg.~\ref{alg:gen_com_unis} to find an approximately minimal set $\mathcal{G}$ that generates the Jordan algebra $\MG$ such that $T_k$ is a subset. As established in \cref{thm:qcommute_is_nphard}, it is intractable to establish that $\mathcal{T}_n$ is a subset of $\MG$.}
\label{fig:commute_hierarchy}
\end{figure}

\section{The Jordan Algebra of Constraint Preserving Drivers and Their Generators}\label{sec:jordan}

A natural question, given a collection of generators, is whether this collection is enough to span the corresponding commutation space $ \mathcal{F}_{b_1, \ldots, b_m} $ for a collection of embedded constraint operators. Irreducible commutation (see \cref{sec:algecond}, \cref{sec:complexity} for context) stipulates that $ \mathcal{G} $ can connect the whole feasible subspace if for every $ x, y \in F_{b_1,\ldots, b_m} $, we have $ \bra{x} \left( \sum_{g \in \mathcal{G}} g \right)^{r} \ket{y} \neq 0 $. Satisfying the irreducible commutation condition means that the associated Hamiltonian driver has full reachability within the feasible subspace $ F_{b_1,\ldots,b_m} $ of interest, or equivalently that $ e^{i \left( \sum_{j} g_{j} \right) } = \sum_{k=0}^{\infty} \left( i \sum_{j} g_{j} \right)^{k} / k! $ has nonzero support on every basis term in the commutation subspace $ \mathcal{F}_{b_1,\ldots,b_m} $.

\quad Given a final Hamiltonian $ H_{f} $ that specifies a solution subspace within the feasible subspace, a Hamiltonian driver $ H_{d} $ is sufficient (for any given $ H_{f} $) to explore the feasible space to find a solution if and only if $ H_{d} $ irreducibly commutes over the feasible subspace. In particular, given an initial state $ \ket{\psi(0)} $ such that $ \bra{\psi(0)} \hat{F}_{b} \ket{\psi(0)} = 1 $ and an initial Hamiltonian $ H_{in} $ such that $ \ket{\psi(0)} $ is the unique ground state such that $ [ H_{in}, H_{d} ] = 0 $, $ H_{d} $ is universal for exploring the feasible space such that a solution state to $ H_{f} $ can be found in the long time limit, for example through an annealing schedule such as $ (1-s) \, H_{in} + (1-s) \, s \, H_{d} + s \, H_{f} $ with $ s = t / T $, where $ T $ is sufficiently large. The corresponding mixer $ U_{d}(\beta) $ is then universal for solving such problems within QAOA given $ U_{in}(\gamma) = e^{ i \, \gamma H_{in} } $ and $ U_{f} = e^{ i \, \alpha H_{f} } $ through Trotterization in the high $ p $ limit. In practice, given a feasible state $ x $, $ H_{in} $ can be made of one-local spin-z terms for the state $ \ket{\psi(0)} = \ket{x} $ (for example, see Ref.~\cite{leipold2022quantum} for a corresponding discussion in the case of circuit fault diagnostics).

\quad Indeed, as shown in \cref{sec:complexity}, it is NP-Hard to know if a collection of terms is enough to irreducibly commute. In particular, given a set of embedded constraint operators $ \mathcal{C} $ and a set of commutative terms $ \mathcal{T} = \{ T_{1}, \ldots, T_{|\mathcal{T}|} \} $, deciding whether there is a term $ g $ such that $ g $ is not in the set $ \mathcal{T} $ but impacts the reachability of the driver Hamiltonian for $ \mathcal{C} $ is NP-Hard. This limitation must either be overcome through other information, as known ans{\"a}tze for many studied constrained optimization problems have, or simply ignored. For example, we find a set of commutative terms without the explicit guarantee of irreducible commutation and solve the corresponding optimization problem within the effective constraint space. An important question for practitioners is then: given a set of constraints, is it plausible that a collection of local commutator terms is sufficient to irreducibly commute over the corresponding feasible subspace? If so, either designed mixers through careful analysis or automated search as will be discussed in \cref{sec:theory_to_prac} can be used to construct mixers for the constraints considered.

\subsection{The Jordan Algebra of Constraint Preserving Drivers}

Any term can be a driving generator and thereby the driver Hamiltonian $ H = \sum_{j} g_{j} $ has a corresponding unitary
\begin{align}\label{eq:expexpansion}
e^{i \, H} = \sum_{k=0}^{\infty} \frac{1}{k!} \left( i \sum_{j} g_{j} \right)^{k} ,
\end{align}

which has terms associated with both multiplication and addition of the individual generators. Therefore each term in the associative algebra generated by this set of generators is obtained under addition and multiplication. However, $ g_{j} g_{k} $ is not Hermitian unless $ [ g_{j},  g_{k} ] = 0 $, meaning that many terms are not Hermitian in the resulting associative algebra. As such, Hermiticity is a \textit{desired} feature for generators.

\quad Instead, utilizing addition and anticommutation allows for expressing the same exponential while preserving Hermiticity. The resulting structure, a set of Hermitian operators closed under addition and the anticommutator $ \{ a, b \} = a \, b + b \, a $, is a \textbf{Jordan algebra}.\footnote{This is a \textit{special} Jordan algebra, since it sits inside an associative algebra (complex matrices over $ n $ qubits, $ \mathbb{C}^{2^{n} \times 2^{n}} $, under multiplication and addition) as a subspace closed under the anticommutator. Every Jordan algebra arising in this work is special in this sense, so we drop the qualifier for brevity.} In our setting we are interested in the Jordan algebra generated by Hermitian operators $ \mathcal{G} $ that each commute with every embedded constraint $ \hat{C} $. The resulting algebra consists of Hermitian operators each of which also commutes with $ \hat{C} $ for every $ C \in \mathcal{C} $ (formalized in Lemma~\ref{lem:anticom_commutes} below). Anticommutation is not associative, i.e. $ \{ a, \{ b, c \} \} \neq \{ \{ a, b \}, c \} $, but distributes (left/right) over addition, i.e. $ \{ a, b + c \} = a \, ( b + c ) + ( b + c ) \, a = a \, b + b \, a + a \, c + c \, a = \{ a, b \} + \{ a, c \} $ and $ \{ a + b, c \} = \{ c, a + b \} $.
\begin{lem}[Anticommutator of Commutators Commutes]\label{lem:anticom_commutes}
If $ [ a, \hat{C} ] = 0 $ and $ [b, \hat{C} ] = 0 $ then $[ \{ a, b \}, \hat{C} ] = 0 $.
\end{lem}
\begin{proof}
\begin{align}
[ \{ a, b \} , \hat{C} ] &= (a \, b + b \, a) \, \hat{C} - \hat{C} \, (a \, b + b \, a)  \nonumber \\
&= a \, b \, \hat{C} + b \, a \, \hat{C} - \hat{C} \, a \, b - \hat{C} \, b \, a \nonumber \\
&= (a \, b \, \hat{C} - \hat{C} \, a \, b) + (b \, a \, \hat{C} - \hat{C} \, b \, a) \nonumber \\
&= 0.
\end{align}
The last step follows from $ [ a, \hat{C} ] = [ b, \hat{C} ] = 0 $ since it implies $ \hat{C} \, a \, b = a \, \hat{C} \, b =  a \, b \hat{C} $ and $ b \, a \, \hat{C} = b \, \hat{C} \, a = \hat{C} \, b \, a $.
\end{proof}

\quad We show that terms of the exponent expansion \cref{eq:expexpansion} are sums of nested anticommutators, and therefore addition and anticommutation are enough to express the resulting commutation matrix space over which $ e^{i H} $ is nonzero. Define
\begin{align}
E_{k} = \left( \sum_{j} g_{j} \right)^{k}.
\end{align}
Note that since $ \sum_{j} g_{j} $ is Hermitian, so is $ E_{k} $ for every $ k $. Then:
\begin{align}
E_{2} &= \left( \sum_{j} g_{j} \right) \left( \sum_{k} g_{k} \right) = \sum_{j,k} g_{j} g_{k} = \sum_{j,k} \frac{1}{2} \{ g_{j}, g_{k} \}
\end{align}

Recognize that $ E_{2} $ is a sum of anticommutators of the generators and also Hermitian. Furthermore:
\begin{align}
E_{k} &= \left( \sum_{j} g_{j} \right)^{k} = \frac{ 1 }{ 2 } \left( E_{k-1} \left( \sum_{j} g_{j} \right) + \left( \sum_{j} g_{j} \right) E_{k-1} \right) \nonumber \\
&= \sum_{j} \frac{1}{2} \left\{ g_{j}, E_{k-1} \right\} .
\end{align}

\quad As such, $ e^{ i \left( \sum_{j} g_{j} \right) } = \sum_{k=0}^{\infty} i^{k} E_{k} / k! $ is made of Hermitian terms $ E_{k} $ generated through nested anticommutation and addition.

\begin{defi}[Jordan Algebra of Constraint Preserving Drivers]
Given a set of \textit{generators} $ \mathcal{G} = \{ g_{1}, \ldots, g_{k} \} $ such that $ [ g_{i}, \hat{C} ] = 0 $, the Jordan algebra of the generators is generated through summation of nested anticommutators
\begin{align}
\MG = \text{Alg}\left(\{ \R, \mathcal{G}, +, \{ \cdot, \cdot \} \}\right) = \text{span}\left( \bigcup_{k=1}^{\infty} \mathcal{G}_{k}^{\{\cdot,\cdot\}} \right),
\end{align}
with $ \{ \cdot, \cdot \} $ as the bilinear anticommutator and $ X \times_{\{ \cdot, \cdot \}} Y = (\{X_{i}, Y_{j}\})_{(i-1)|Y|+j} $ is the list (or vector) of anticommutators between each element in the two sets. For example, let $X = ( X_1, X_2 ), Y = ( Y_1, Y_2 ) $, then:
\begin{align}
X \times_{\{ \cdot, \cdot \}} Y &= (\{ X_1, Y_1 \}, \{ X_1, Y_2 \}, \{ X_2, Y_1 \}, \{ X_2, Y_2 \}).
\end{align}
Then $\mathcal{G}_{k}^{\{\cdot,\cdot\}}$ is defined by the hierarchy
\begin{align}
\mathcal{G}_{1}^{\{\cdot,\cdot\}} &= \mathcal{G}, \label{eq:g_com_base}\\
\mathcal{G}_{k}^{\{\cdot,\cdot\}} &= \mathcal{G} \times_{\{\cdot,\cdot\}} \mathcal{G}_{k-1}^{\{\cdot,\cdot\}} .
\end{align}
\end{defi}

Up unto this point, we have relied only on the Kronecker (tensor) product for matrices, but hereforth it is useful to use the Kronecker product for vectors.

\begin{defi}[Kronecker Product for Vectors]
Given $ \beta = (\beta_1, \ldots, \beta_{p}), \gamma = (\gamma_1, \ldots, \gamma_{q}) $, the Kronecker product for the two vectors is:
\begin{align}
(\beta \, \otimes \, \gamma)_{p (i-1) + j} = \beta_i \, \gamma_j \in \R^{p \times q}, \; 1 \leq i \leq p, 1 \leq j \leq q
\end{align}
\end{defi}

For example, given $\beta = (\beta_1, \beta_2) $ and $\gamma = (\gamma_1, \gamma_2)$ the Kronecker product is $\beta \otimes \gamma = (\beta_1 \gamma_1, \beta_1 \gamma_2, \beta_2 \gamma_1, \beta_2 \gamma_2) $. While we have previously used the element-wise product $ A \circ B $ for binary vectors, here we extend it to operate between vectors of scalars and vectors of operators.

\begin{defi}[Element-wise product of vectors (of scalars and operators)]
Given a vector of scalars $(a_1,\ldots,a_n)$ and a vector of operators $(X_1,\ldots,X_n)$, the elementwise product is $ a \circ X = \sum_{j=1}^{n} a_{j} X_{j} $.
\end{defi}

In fact, anticommutation of two element-wise products of scalars and operators follows a simple distributivity law involving the Kronecker product of vectors and the direct product of anticommutation for the operators.

\begin{lem}[Coefficient–Anticommutator Distributivity]\label{lem:trans_anticom_coef}
$ \{ \beta \circ X, \gamma \circ Y \} = (\beta \otimes \gamma) \circ X \times_{\{\cdot,\cdot\}} Y $
\end{lem}
\begin{proof}
The proof follows from algebraic manipulation:
\begin{align}
\{ \beta \circ X, \gamma \circ Y \} &= \left\{ \sum_{j} \beta_{j} \, X_{j}, \sum_{k} \gamma_{k} \, Y_{k} \right\} \nonumber \\
&= \sum_{jk} \beta_{j} \, \gamma_{k} \, \{ X_{j}, Y_{k} \} \nonumber \\
&= \sum_{jk} (\beta \otimes \gamma)_{(j-1)|X| + k} \, (X \times_{\{\cdot,\cdot\}} Y)_{(j-1)|X| + k} \nonumber \\
&= (\beta \otimes \gamma) \circ X \times_{\{\cdot,\cdot\}} Y .
\end{align}
\end{proof}

\begin{thm}[Anticommutation Expansion within the Jordan Algebra of Constraint Preserving Drivers]
Given a set of \textit{generators} $ \mathcal{G} = \{ g_{1}, \ldots, g_{k} \} $ such that $ [ g_{i}, \hat{C} ] = 0 $, then
\begin{align}
U_{d}(\beta) = e^{i \, \sum_{j} \beta_{j} \, g_{j} } = \mathbbm{1} + \frac{1}{2} \sum_{k=1}^{\infty} i^{k} \, \beta^{\otimes k} \circ \mathcal{G}_{k}^{\{\cdot,\cdot\}} / (2^{k} k!),
\end{align}
and so $ U_{d}(\beta) \in \MG $.
\end{thm}

\begin{proof}
The proof follows from generalizing the anticommutation expansion to include coefficients.

We show the $k$ term expansion matches the summation of element-wise products:
\begin{align}
\beta^{\otimes k} \circ \mathcal{G}_{k}^{\{\cdot,\cdot\}} / 2^{k-1} = \left(\sum_j \beta_j g_{j} \right)^{k}.
\end{align}

For $k=1$, $\beta \circ \mathcal{G}_{1}^{\{\cdot,\cdot\}} = \sum_{j} \beta_{j} g_{j} $ directly from \cref{eq:g_com_base}. Assume it is true for $E_{k-1}$ such that $ \beta^{\otimes k-1} \circ \mathcal{G}_{k-1}^{\{\cdot,\cdot\}} / 2^{k-2} = E_{k-1} $, then we show the result holds for $E_{k}$:
\begin{align}
E_{k} &= \left( \sum_{j} \beta_{j} g_{j} \right)^{k} \nonumber \\
&= \frac{ 1 }{ 2 } \left( \left( \beta^{\otimes k-1} \circ \mathcal{G}_{k-1}^{\{\cdot,\cdot\}} / 2^{k-2} \right) \left( \sum_{j} \beta_{j} g_{j} \right) + \left( \sum_{j} \beta_{j} g_{j} \right) \left( \beta^{\otimes k-1} \circ \mathcal{G}_{k-1}^{\{\cdot,\cdot\}} / 2^{k-2} \right) \right) \nonumber \\
&= \frac{1}{2^{k-1}} \left\{ \beta \circ \mathcal{G}, \beta^{\otimes k-1} \circ \mathcal{G}_{k-1}^{\{\cdot,\cdot\}} \right\} \nonumber \\
&= \frac{1}{2} (\beta \otimes \beta^{\otimes k-1}) \circ \left( \mathcal{G} \times_{\{\cdot,\cdot\}} \mathcal{G}_{k-1}^{\{\cdot,\cdot\}} \right) / 2^{k-2} \;\; \text{(from Lemma~\ref{lem:trans_anticom_coef})} \nonumber\\
&= \beta^{\otimes k} \circ \mathcal{G}_{k}^{\{\cdot,\cdot\}} / 2^{k-1} .
\end{align}
Then we have the expansion for exponentiation in terms of the anticommutators of the Jordan algebra:
\begin{align}
e^{i \sum_{j} \beta_{j} g_{j}} &= \mathbbm{1} + \sum_{k=1}^{\infty} i^{k} \frac{1}{k!} \left( \sum_{j} \beta_{j} g_{j} \right)^{k} \\
&= \mathbbm{1} + \sum_{k=1}^{\infty} i^{k} \frac{\beta^{\otimes k} \circ \mathcal{G}_{k}^{\{\cdot,\cdot\}}}{2^{k-1} k!},
\end{align}
and the Hermiticity of terms yields:
\begin{align}
\cos\left( \sum_{j} \beta_{j} g_{j} \right) &= \mathbbm{1} + \sum_{r=1}^{\infty} (-1)^{r} \frac{\beta^{\otimes 2r} \circ \mathcal{G}_{2r}^{\{\cdot,\cdot\}}}{2^{2r-1} (2r)!}, \\
\sin\left( \sum_{j} \beta_{j} g_{j} \right) &= \sum_{r=0}^{\infty} (-1)^{r} \frac{\beta^{\otimes 2r+1} \circ \mathcal{G}_{2r+1}^{\{\cdot,\cdot\}}}{2^{2r} (2r+1)!}.
\end{align}
\end{proof}

\begin{thm}[Separability under Generators]\label{thm:separability_gens}
A set of generators $ \mathcal{G} = \{ g_{j} \} $ imposes each $ C \in \mathcal{C} $ such that $ [ g_{j}, \hat{C} ] = 0 $ iff $ \MG \subseteq \mathcal{F}_{C_{1},\ldots,C_{m}} $.
\end{thm}

\begin{proof}
The backward direction holds immediately by recognizing that $ g_{j} \in \MG $ and so, if $ g_{j} \in \mathcal{F}_{C_{1},\ldots,C_{m}} $ by assumption then $[ g_{j}, \hat{C} ] = 0$ by \cref{thm:invar_multi_con}.

In the forward direction, we use the closure over addition and anticommutation. Suppose $h \in \MG $ and $[ h, \hat{C}_{j} ] \neq 0 $. Then $ h = \sum_{kl} \alpha_{kl} h_{k,l} $ where each $ h_{k,l} \in \bigcup_{k} \mathcal{G}_{k}^{\{\cdot,\cdot\}} $. By induction on Lemma~\ref{lem:anticom_commutes}, each Hermitian operator $ q \in \mathcal{G}_{k}^{\{\cdot,\cdot\}} $ commutes with $ \hat{C}_{j} $. Then from (bi)linearity $ [h, \hat{C}_{j} ] = \sum_{kl} \alpha_{kl} [ h_{k,l}, \hat{C}_{j} ] = 0 $.
\end{proof}

\begin{thm}[Reachability under Generators]\label{thm:reachability_gens}
A set of generators $ \mathcal{G} $ \textit{irreducibly commutes} with $\hat{F}_{(b_1,\ldots,b_m)}$ iff $ \mathcal{F}_{(b_1,\ldots,b_m)} \subseteq \MG $ and $\MG \subseteq \mathcal{F}_{C_{1},\ldots,C_{m}}$.
\end{thm}

\begin{proof}
From \cref{thm:separability_gens}, we know that $\mathcal{G}$ does not mix states in separate feasibility subspaces. We must only show the forward and backward implications for $\mathcal{F}_{(b_{1}, \ldots,b_{m})} \subseteq \MG$.

Suppose the forward direction fails. Then there exists $ \ketbra{x}{y} $ such that $\ketbra{x}{y} \notin \MG $ with $ x\neq y$. Then with $\mathbf{1} = (1,\ldots,1) \in \R^{|\mathcal{G}|}$:
\begin{align}
\bra{x} \, e^{i H} \ket{y} &= \text{Tr}\left( \ketbra{y}{x} \, e^{i \sum_{j} g_{j}} \right) \nonumber \\
&= \sum_{k=1}^{\infty} i^{k} \, \text{Tr}\left( \ketbra{y}{x}\frac{\mathbf{1}^{\otimes k} \circ \mathcal{G}_{k}^{\{\cdot,\cdot\}}}{2^{k-1} k!} \right) \; \; \nonumber \\
&= \sum_{k=1}^{\infty} i^{k} \, \frac{1}{2^{k-1} k!} \sum_{g \in \mathcal{G}_{k}^{\{\cdot,\cdot\}}} \text{Tr}\left( \ketbra{y}{x} \, g \right) \nonumber \\
&= 0.
\end{align}

Suppose backward direction fails. Then an identical argument would show $ \bra{x} e^{i H} \ket{y} \neq 0 $. Note that this derivation for driving generators extends to mixing generators, see Def.~\ref{defi:genetator_mixing_driving}, through the Lie-Trotter formula for sufficiently high depth circuits:
\begin{align}
e^{i \sum_{j} g_{j}} = \lim_{n \rightarrow \infty} \left( \prod_{j} e^{i g_{j}/n} \right)^{n} .
\end{align}
\end{proof}

From \cref{sec:algecond}, if we find all terms for $\mathcal{T} $ with no locality constraint, we construct a span for $\mathcal{F}_{b} $. In practice, however, a set of local generators $ \mathcal{G} \subseteq \mathcal{T}_{k} $ for some locality $k$ is sufficient such that $ \mathcal{F}_{b} \subseteq \MG $.

\subsection{Illustrative Example}\label{subsec:jordan_exam}

\cref{fig:commute_hierarchy} shows an example for $k=2$, where $\mathcal{G} \subset \mathcal{T}_{k}$ and $\text{span}(\mathcal{T}_{k}) \subset \MG \subset \text{span}(\mathcal{T}_{n}) $. The example in the figure considers a single constraint
\begin{align}
C = (1-x_1) + (1-x_2) + (1-x_3) + (1-x_4) + 2 \, (1-x_5) + 2\, (1-x_6)
\end{align}
that has the embedding
\begin{align}
\hat{C} = \sigma_{1}^{0} + \sigma_{2}^{0} + \sigma_{3}^{0} + \sigma_{4}^{0} + 2 \, \sigma_{5}^{0} + 2 \, \sigma_{6}^{0} .
\end{align}
Then suppose we have the generator set
\begin{align}
\mathcal{G} &= \left\{ \sigma_{1}^{+} \sigma_{2}^{-} + \sigma_{1}^{-} \sigma_{2}^{+}, \sigma_{2}^{+} \sigma_{3}^{-} + \sigma_{2}^{-} \sigma_{3}^{+}, \sigma_{3}^{+} \sigma_{4}^{-} + \sigma_{3}^{-} \sigma_{4}^{+}, \sigma_{5}^{+} \sigma_{6}^{-} + \sigma_{5}^{-} \sigma_{6}^{+} \right. \nonumber \\
&\phantom{===} i(\sigma_{1}^{+} \sigma_{2}^{-} - \sigma_{1}^{-} \sigma_{2}^{+}), i(\sigma_{2}^{+} \sigma_{3}^{-} - \sigma_{2}^{-} \sigma_{3}^{+}), i(\sigma_{3}^{+} \sigma_{4}^{-} - \sigma_{3}^{-} \sigma_{4}^{+}), i(\sigma_{5}^{+} \sigma_{6}^{-} - \sigma_{5}^{-} \sigma_{6}^{+}) \left. \right\} .
\end{align}

The set is enough to generate the missing elements of $\mathcal{T}_{2}$; recall $\sigma^{+} \sigma^{-} = \sigma^{1}, \sigma^{-} \sigma^{+} = \sigma^{0}, \sigma^{0} + \sigma^{1} = \mathbbm{1}_{2} $ such that
\begin{align}
\{ \sigma_{1}^{+} \sigma_{2}^{-} + \sigma_{1}^{-} \sigma_{2}^{+} , \sigma_{2}^{+} \sigma_{3}^{-} + \sigma_{2}^{-} \sigma_{3}^{+} \} &= (\sigma_{1}^{+} \sigma_{2}^{-} + \sigma_{1}^{-} \sigma_{2}^{+})(\sigma_{2}^{+} \sigma_{3}^{-} + \sigma_{2}^{-} \sigma_{3}^{+}) + (\sigma_{2}^{+} \sigma_{3}^{-} + \sigma_{2}^{-} \sigma_{3}^{+}) (\sigma_{1}^{+} \sigma_{2}^{-} + \sigma_{1}^{-} \sigma_{2}^{+}) \nonumber \\
&= (\sigma_{1}^{+} \sigma_{2}^{0} \sigma_{3}^{-} + 0 + \sigma_{1}^{-} \sigma_{2}^{1} \sigma_{3}^{+} + 0) + (\sigma_{1}^{+} \sigma_{2}^{1} \sigma_{3}^{-} + 0 + 0 + \sigma_{1}^{-} \sigma_{2}^{0} \sigma_{3}^{+}) \nonumber \\
&= \sigma_{1}^{+} \sigma_{3}^{-} + \sigma_{1}^{-} \sigma_{3}^{+} \in \MG, \\
\{ i(\sigma_{1}^{+} \sigma_{2}^{-} - \sigma_{1}^{-} \sigma_{2}^{+}) , \sigma_{2}^{+} \sigma_{3}^{-} + \sigma_{2}^{-} \sigma_{3}^{+} \} &= i(\sigma_{1}^{+} \sigma_{2}^{-} - \sigma_{1}^{-} \sigma_{2}^{+})(\sigma_{2}^{+} \sigma_{3}^{-} + \sigma_{2}^{-} \sigma_{3}^{+}) + i (\sigma_{2}^{+} \sigma_{3}^{-} + \sigma_{2}^{-} \sigma_{3}^{+}) (\sigma_{1}^{+} \sigma_{2}^{-} - \sigma_{1}^{-} \sigma_{2}^{+}) \nonumber \\
&= i (\sigma_{1}^{+} \sigma_{2}^{0} \sigma_{3}^{-} + 0 - \sigma_{1}^{-} \sigma_{2}^{1} \sigma_{3}^{+} + 0) + i (\sigma_{1}^{+} \sigma_{2}^{1} \sigma_{3}^{-} + 0 - \sigma_{1}^{-} \sigma_{2}^{0} \sigma_{3}^{+} + 0) \nonumber \\
&= i (\sigma_{1}^{+} \sigma_{3}^{-} - \sigma_{1}^{-} \sigma_{3}^{+}) \in \MG,
\end{align}
and then similarly:
\begin{align}
\{ \sigma_{1}^{+} \sigma_{3}^{-} + \sigma_{1}^{-} \sigma_{3}^{+} , \sigma_{3}^{+} \sigma_{4}^{-} + \sigma_{3}^{-} \sigma_{4}^{+} \} &= \sigma_{1}^{+} \sigma_{4}^{-} + \sigma_{1}^{-} \sigma_{4}^{+} \in \MG , &
\{ i(\sigma_{1}^{+} \sigma_{3}^{-} - \sigma_{1}^{-} \sigma_{3}^{+}), \sigma_{3}^{+} \sigma_{4}^{-} + \sigma_{3}^{-} \sigma_{4}^{+} \} &= i(\sigma_{1}^{+} \sigma_{4}^{-} - \sigma_{1}^{-} \sigma_{4}^{+}) \in \MG , \nonumber \\
\{ \sigma_{2}^{+} \sigma_{3}^{-} + \sigma_{2}^{-} \sigma_{3}^{+} , \sigma_{3}^{+} \sigma_{4}^{-} + \sigma_{3}^{-} \sigma_{4}^{+} \} &= \sigma_{2}^{+} \sigma_{4}^{-} + \sigma_{2}^{-} \sigma_{4}^{+} \in \MG , &
\{ i(\sigma_{2}^{+} \sigma_{3}^{-} - \sigma_{2}^{-} \sigma_{3}^{+}) , \sigma_{3}^{+} \sigma_{4}^{-} + \sigma_{3}^{-} \sigma_{4}^{+} \} &= i(\sigma_{2}^{+} \sigma_{4}^{-} - \sigma_{2}^{-} \sigma_{4}^{+}) \in \MG . \nonumber
\end{align}

There are also many more higher order terms, as an example
\begin{align}
\MG &\ni (\sigma_{1}^{+} \sigma_{3}^{-} + \sigma_{1}^{-} \sigma_{3}^{+})(\sigma_{2}^{+} \sigma_{4}^{-} + \sigma_{2}^{-} \sigma_{4}^{+}) - i^2 (\sigma_{1}^{+} \sigma_{3}^{-} - \sigma_{1}^{-} \sigma_{3}^{+})(\sigma_{2}^{+} \sigma_{4}^{-} - \sigma_{2}^{-} \sigma_{4}^{+}) \nonumber \\
&\phantom{\ni==} = (\sigma_{1}^{+} \sigma_{3}^{-} + \sigma_{1}^{-} \sigma_{3}^{+})(\sigma_{2}^{+} \sigma_{4}^{-} + \sigma_{2}^{-} \sigma_{4}^{+}) + (\sigma_{1}^{+} \sigma_{3}^{-} - \sigma_{1}^{-} \sigma_{3}^{+})(\sigma_{2}^{+} \sigma_{4}^{-} - \sigma_{2}^{-} \sigma_{4}^{+}) \nonumber \\
&\phantom{\ni==} = \sigma_{1}^{+} \sigma_{2}^{+} \sigma_{3}^{-} \sigma_{4}^{-} + \sigma_{1}^{+} \sigma_{2}^{-} \sigma_{3}^{-} \sigma_{4}^{+} + \sigma_{1}^{-} \sigma_{2}^{+} \sigma_{3}^{+} \sigma_{4}^{-} + \sigma_{1}^{-} \sigma_{2}^{-} \sigma_{3}^{+} \sigma_{4}^{+} \nonumber \\
&\phantom{\ni===} + \sigma_{1}^{+} \sigma_{2}^{+} \sigma_{3}^{-} \sigma_{4}^{-} - \sigma_{1}^{+} \sigma_{2}^{-} \sigma_{3}^{-} \sigma_{4}^{+} - \sigma_{1}^{-} \sigma_{2}^{+} \sigma_{3}^{+} \sigma_{4}^{-} + \sigma_{1}^{-} \sigma_{2}^{-} \sigma_{3}^{+} \sigma_{4}^{+} \nonumber \\
&\phantom{\ni==} = \sigma_{1}^{+} \sigma_{2}^{+} \sigma_{3}^{-} \sigma_{4}^{-} + \sigma_{1}^{-} \sigma_{2}^{-} \sigma_{3}^{+} \sigma_{4}^{+} .
\end{align}

But some terms in $\mathcal{F}_{C} = \text{span}\left( \mathcal{T}_{n} \right) $ are not in $\MG \cong \mathcal{T}_{2}$, some illustrative examples
\begin{align}
\sigma_{1}^{+} \sigma_{2}^{+} \sigma_{5}^{-} + \sigma_{1}^{-} \sigma_{2}^{-} \sigma_{5}^{+} &\in \mathcal{T}_{3} , \\
\sigma_{1}^{+} \sigma_{4}^{+} \sigma_{6}^{-} + \sigma_{1}^{-} \sigma_{4}^{-} \sigma_{6}^{+} &\in \mathcal{T}_{3} , \\
\sigma_{1}^{+} \sigma_{2}^{+} \sigma_{3}^{+} \sigma_{4}^{+} \sigma_{5}^{-} \sigma_{6}^{-} + \sigma_{1}^{-} \sigma_{2}^{-} \sigma_{3}^{-} \sigma_{4}^{-} \sigma_{5}^{+} \sigma_{6}^{+} &\in \mathcal{T}_{6} .
\end{align}

We derive multiplication, commutation, and anti-commutation identities for linear constraint terms such that these can be used in \cref{sec:theory_to_prac} to find generators for the set $ \mathcal{T}_{k} $ in \cref{sec:alg_driver_com}.

\section{From Theory to Practice: Tractably Finding Commutators, Generator Selection, and Sequencing into Mixers}\label{sec:theory_to_prac}

\newcommand{\alglincom}{\texttt{bt\_find\_linear\_commutators}}
\newcommand{\algcom}{\texttt{bt\_find\_commutators}}
\newcommand{\algcheckgen}{\texttt{check\_linear\_in\_jordan}}
\newcommand{\algcheckcon}{\texttt{check\_constraint\_status}}
\newcommand{\algconcond}{\texttt{constraint\_condition}}

\pagestyle{plain}
\begin{figure}[H]
\centering
\includegraphics[trim=0.7cm 0 0.0cm 0, clip, width=1.0\textwidth]{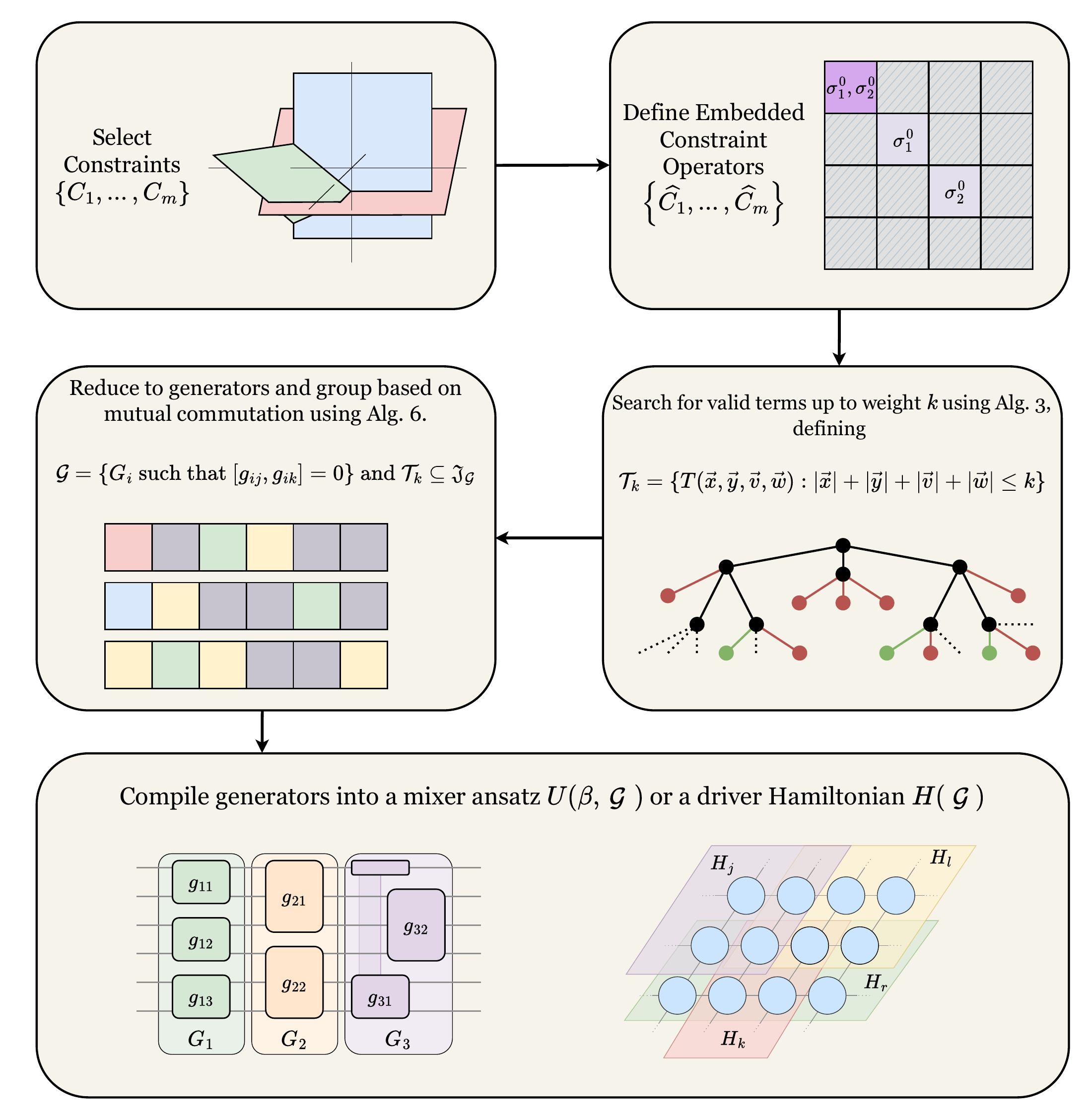}
\caption{\textbf{Algorithimic Flowchart for Imposing Classical Symmetries.} A flowchart of how commutators for general constraints can be found and then sequenced into an ansatz. We consider a collection of constraints, find commutators based on algebraic expressions in Sec.~\ref{sec:algecond} or App.~\ref{app:suffquad}, find a collection of generator sets of low locality and low cardinality for these commutators in Sec.~\ref{subsec:redgen_sequencing}, and then define corresponding unitaries based on their representation in Sec.~\ref{subsec:gen_to_proj}.}
\label{fig:flowchart}
\end{figure}

Following the theoretical results in Sections~\ref{sec:algecond}-\ref{sec:jordan}, this section addresses the practical implementation of constraint imposition. Symmetry preserving ans{\"a}tze are widely studied and typically selected through detailed analysis of a specific problem's symmetry. This section gives practical prescriptions for the general ansatz problem: given a collection of constraints $ \mathcal{C} $, identify a small set of local mixers that generate the Jordan algebra of all possible Hermitian transformations that are off-diagonal and commute with each embedding $ \hat{C}_{i} $ for $ C_{i} \in \mathcal{C} $.

\quad Theoretical results from the previous section set important boundaries for this question: it is, in general, NP-Complete to find a nonlocal unitary that commutes with the embedded constraints and NP-Hard to know if a set of driver terms provides enough expressibility at sufficient depth for the collection of all possible unitaries that commute with the embedded constraints. In practice, for many optimization problems of interest, mixers with sufficient expressibility are indeed local and so solving the algebraic forms of \cref{sec:algecond} can in practice be utilized for such problems given a sufficiently large locality bound $ \ell $.

\quad As discussed in \cref{sec:complexity}, finding commutators with bounded locality weight is tractable, but in practice, for maximal locality weight $ \ell $, the runtime scales as $ \mathcal{O}(4^{\ell} n^{\ell}) $, which can be too expensive for large $ n $ even for modest $ \ell $. However, it is not clear how \textit{tight} the bound on the worst-case is nor is it clear how common the worst-case is. The result of Alg.~\ref{alg:find_lincom_terms} or Alg.~\ref{alg:find_com_terms} given a maximum locality weight $\ell$ is the set:
\begin{align}
\mathcal{T}_{\ell} = \{ T(\vec{x}, \vec{y}, \vec{v}, \vec{w}) \, | \, \vec{x}, \vec{y}, \vec{v}, \vec{w} \in \{0,1\}^{n} \text{ s.t. $x_{i} + y_{i} + v_{i} + w_{i} \leq 1$ for $i \in [n]$ and } |\vec{x}| + |\vec{y}| + |\vec{v}| + |\vec{w}| \leq \ell \},
\end{align}
which form generators of the additive space of commutators that commute with the set of constraints up to a locality weight $ \ell $ (see \cref{def:localityweight} and \cref{def:local_comterms}). Other constraints, such as geometric locality, can also be enforced with minimal changes to the approach. Tracking the constraint statuses in Alg.~\ref{alg:find_lincom_terms} and Alg.~\ref{alg:find_com_terms} helps quickly discard unusable terms, but we omit this detail here for simplicity.

\quad Once we have found a set of terms $ \mathcal{T}_{\ell} $ such that each term commutes with the embedded constraints, we aim to identify a subset of terms to generate the Jordan algebra and decide how to construct the mixers as unitary operators. \cref{fig:flowchart} is a flowchart showing, starting from a set of constraints, how a final mixer ansatz or driver Hamiltonian is constructed step by step through algorithms in this section.

\quad \cref{subsec:find_lincoms,subsec:find_coms} discusses a backtracking approach for finding commutators, which is exact and achieves a runtime superior to brute-force enumeration; this algorithm is used in the constructions considered in \cref{sec:qaoa_1in3}. \cref{subsec:redgen_sequencing} then introduces a method for partitioning commutative terms into a collection of generator sets with low locality weight within each set and low overall cardinality. \cref{subsec:gen_to_proj} then gives an approach for defining unitaries that maximizes the parallelization of these terms through projection operators.

\subsection{A Sound and Complete Algorithm for Generating Linear Constraint Commutators}\label{subsec:find_lincoms}

\newcommand{\issat}{\texttt{is\_satisfied}}
\newcommand{\isvia}{\texttt{is\_viable}}
\begin{figure}[!t]
\centering
\setlength{\intextsep}{0.5em}
\begin{minipage}{1.0\linewidth}
\begin{algorithm}[H]
\caption{\algcheckcon$(u,\mathcal{C})$}\label{alg:check_con}
\begin{algorithmic}[1]
\Statex \textbf{Inputs: } term $u$, constraint list $ \mathcal{C} $
\State \issat{} = True
\State Let $ i $ be the max index of non-identity terms in $u$
\For{constraint $C$ in $ \mathcal{C} $}
    \If{\algconcond$(u,C) \neq 0$} \Comment{See Alg.~\ref{alg:concond} and Alg.~\ref{alg:concondlin}}
        \State \issat{} = False
        \State Let $j$ be the max index of nonzero entry of $C$
        \If{$j < i$}
            \State \textbf{return} nonviable.
        \EndIf
    \EndIf
\EndFor
\If{\issat{} == True}
    \State \textbf{return} satisfied.
\EndIf
\State \textbf{return} viable.
\end{algorithmic}
\end{algorithm}
\end{minipage}
\caption*{Alg.~\ref{alg:check_con}: \textbf{Checking the Constraint Status.} This algorithm checks the constraint status of a term $u$ against a collection of constraints $\mathcal{C}$. For each constraint, we check if the constraint condition is currently satisfied and if all are, we return satisfied. If a constraint is not satisfied, we check if it is still viable, that is, could be satisfied at higher recursion depth.}
\begin{minipage}{1.0\linewidth}
\begin{algorithm}[H]
\caption{\alglincom$(u, i, d, \mathcal{C}, \ell, n, \mathcal{T}_{\ell})$}\label{alg:find_lincom_terms}
\begin{algorithmic}
\State \textbf{Inputs:} current term $u$, current index $i$, current length $d$, constraint list $\mathcal{C}$, maximum length $\ell$, number of variables $n$, commutator collection $ \mathcal{T}_{\ell} $
\If{\algcheckcon$(u, \mathcal{C})$ is satisfied}
    \If{$u$ not in $ \mathcal{T}_{\ell} $ (under addition)}
        \State add $u$ to commutative term list $ \mathcal{T}_{\ell} $. \Comment{Store valid commutator term.}
    \EndIf
\EndIf
\If{$ i > n $ or $ d = \ell $} \Comment{We've hit end of indices or hit maximum locality}
    \State \textbf{return. } \Comment{End search.}
\EndIf
\For{$v \in \{ \sigma^{+}, \sigma^{-} \} $}
    \State append $ v $ to $ u $.
    \If{\algcheckcon($u$,$\mathcal{C}$) is viable or satisfied} \Comment{See Alg.~\ref{alg:check_con}}
        \State \alglincom$(u, i+1, d+1, \mathcal{C},\ell, n, \mathcal{T}_{\ell})$ \Comment{Recurse with $\sigma^{+}$ or $\sigma^{-}$.}
    \EndIf
    \State pop $ v $ from $ u $. \Comment{Backtracking.}
\EndFor
\State \alglincom$(u, i+1, d, \mathcal{C}, \ell, n, \mathcal{T}_{\ell})$ \Comment{Recurse with $\mathbbm{1}_{2}$.}
\State \textbf{return. }
\end{algorithmic}
\end{algorithm}
\end{minipage}
\caption*{Alg.~\ref{alg:find_lincom_terms}: \textbf{Searching for Linear Constraint Commutators}. A backtracking approach to finding all commutative terms up to locality weight $\ell$ for embedded linear constraint operators. A collection $ \mathcal{T}_{\ell} $ is built by the algorithm. The term $u$ at each recursion is a viable term for $\mathcal{T}_{\ell}$ since it has not violated any algebraic condition associated with the constraints, as tracked by Alg.~\ref{alg:check_con}. The algorithm considers appending $\sigma^+$ or $\sigma^-$ to every valid $u$ at each index $i$ starting at $1$ and ending at $n$.}
\setlength{\intextsep}{\ointextsep}
\end{figure}

Alg.~\ref{alg:find_lincom_terms} is a backtracking algorithm for finding commutators for a collection of linear constraints given a locality weight bound $\ell$. Backtracking involves recursive depth-first search, each recursive call can be marked by its recursion depth. In the case of Alg.~\ref{alg:find_lincom_terms}, every recursion has an incremented index $i$, so we can associate the depth with the current qubit on which we are considering including a new basis term.

\quad In words, we recursively construct terms by finding appropriate basis elements at each index starting with the first index. We place either of $\{\sigma^{+}, \sigma^{-}\}$ or place no term meaning $\mathbbm{1}_{2}$ remains there. We end the recursion on any branch that becomes nonviable according to the conditions of \cref{eq:thmcom} or if the length would be above the maximum $\ell$. We check if the constraint conditions of \cref{eq:thmcom} are satisfied for each term constructed in each branch and add it to the list $\mathcal{T}_{\ell}$ if so.

\quad To generate $\mathcal{T}_{\ell}$ we would pass it as an empty list $\mathcal{T}_{\ell} \leftarrow \emptyset$ to Alg.~\ref{alg:find_lincom_terms}, which would then populate $\mathcal{T}_{\ell}$ before the top level returns; the call would be \alglincom$((),1,0,\mathcal{C}, \ell, n, \mathcal{T}_{\ell})$.

First we prove that Alg.~\ref{alg:find_lincom_terms} does indeed terminate.

\begin{thm}[Termination of Alg.~\ref{alg:find_lincom_terms}]
Alg.~\ref{alg:find_lincom_terms} eventually terminates and does so with the total recursive calls bounded by $\mathcal{O}\left( 2^{\ell} n^{\ell+1} \right)$.
\end{thm}

\begin{proof}[Proof Sketch]
Since each recursive call has an incremented index $i$, eventually a branch will terminate when $i > n$ if not earlier. In fact, in the case that $\ell>n$ (no length limit) and there are no constraints, the backtracking algorithm would enumerate all $3^{n}$ possible terms in the set $ \{ \mathbbm{1}_{2}, \sigma^{-}, \sigma^{+} \}^{\otimes n} $.

\quad To bound the number of recursive calls, recognize that if a term reaches length $\ell$, the branch is ended. As such, the algorithm constructs at most
\begin{align}
\sum_{k \leq \ell} 2^{k} \binom{n}{k} &\leq \sum_{k \leq \ell} 2^{\ell} n^{\ell} \nonumber \\
&= (\ell+1) 2^\ell n^\ell \nonumber \\
&= \mathcal{O}\left( 2^{\ell} n^{\ell} \right)
\end{align}
terms since we have $n$ slots into which we can place \textit{up to} $\ell$ objects and in each slot we may place one of two operators: $\{\sigma^{+}, \sigma^{-}\}$. Each possible term requires no more than $n$ recursions to be constructed, leading to the bound $\mathcal{O}\left( 2^{\ell} n^{\ell+1} \right)$.
\end{proof}

Next, we prove that Alg.~\ref{alg:find_lincom_terms} would consider any valid commutator term in its search.

\begin{thm}[Completeness of Alg.~\ref{alg:find_lincom_terms}]
If a generator term is in the $\ell$-local commutator set, $g \in \mathcal{T}_{\ell} $, then $g$ will be included in Alg.~\ref{alg:find_lincom_terms}.
\end{thm}

\begin{proof}[Proof Sketch]
A recursive branch with current term $u$ is pruned from the search if and only if \algcheckcon($u$, $\mathcal{C}$) returned nonviable or the length $d$ has reached the maximum $\ell$.

\quad In the case of the latter, it is not possible to place terms $v \in \{ \sigma^{-}, \sigma^{+} \} $ on indices beyond the current $i$ without going beyond length $\ell$. Since the algorithm never deletes placed terms, pruning these calls does not impact completeness.

\quad In the case of the former, it is sufficient to show that a nonviable term $u$ can never become viable through appending $v \in \{\sigma^{-},\sigma^{+}\}$ on some index beyond the current index $i$. Since \algcheckcon($u$,$C$) returned nonviable, there must be a constraint vector $\vec{c} \in C$ such that $\vec{c} \, \cdot \, \vec{u} \neq 0$ and the maximal index of non-zero entries of $\vec{c}$ is equal to or less than the current index $i$. Then placing $v$ in an index beyond $i$ has no impact on $\vec{c} \, \cdot \, \vec{u} \neq 0$ and so all such branches are nonviable as well.
\end{proof}

Finally, we prove that Alg.~\ref{alg:find_lincom_terms} will not include any terms that are not in the valid commutator set.

\begin{thm}[Soundness of Alg.~\ref{alg:find_lincom_terms}]
If a generator term $g$ is included in Alg.~\ref{alg:find_lincom_terms}, then $g \in \mathcal{T}_{\ell}$
\end{thm}

\begin{proof}[Proof Sketch]
The proof follows from similar logic as theorem~\ref{thm:gencom}. A commutator term $u$ is included in the list only if \algcheckcon($u$,$\mathcal{C}$) is valid and the term would be generated by the recursive call structure. The former is true due to theorem~\ref{thm:gencom}. The latter is true since the term $u$ has length less than or equal to $\ell$ and each previous branch included a term that must have been viable.
\end{proof}

\quad \cref{fig:tree_backtrack} depicts how such a commutator set is constructed for a simple example with the constraint $ \sigma_{1}^{1} + 2 \, \sigma_{2}^{1} + \sigma_{3}^{1} $. As shown in \cref{fig:flowchart}, once we have found $ \mathcal{T}_{\ell} $, we can consider reducing it to a smaller set of generators for the Jordan algebra.

\begin{figure}
    \centering

    \includegraphics[width=0.95\linewidth]{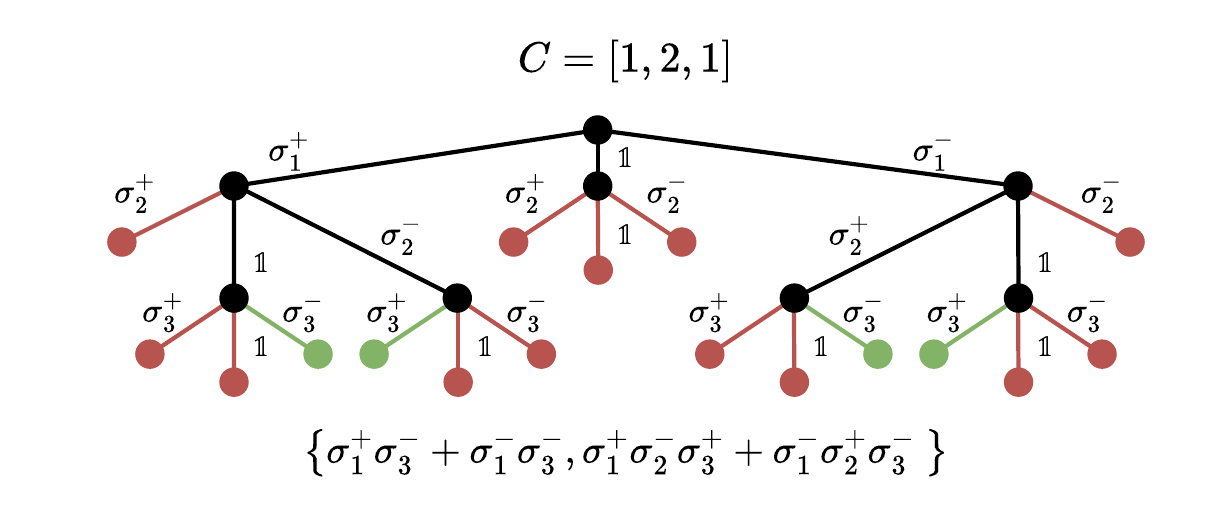} \\
    (A) \\
    \includegraphics[width=0.95\linewidth]{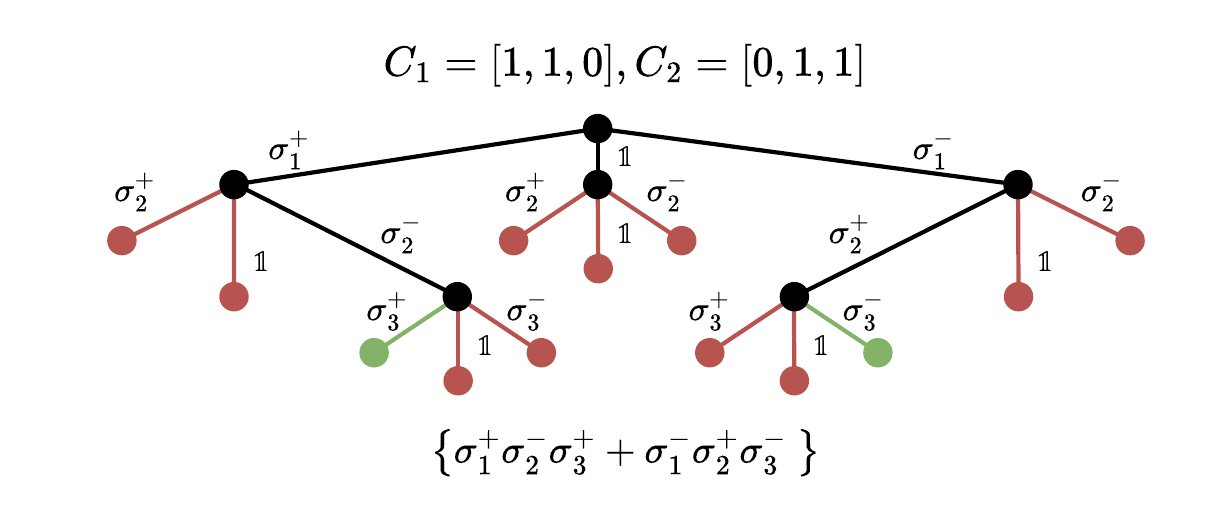} \\
    (B)
    \caption{\textbf{Backtracking to Find Commutators.} Two simple examples of how backtracking edits a term recursively to construct all commutators up to given conditions. For (A) note that after building $ \mathbbm{1} \otimes \sigma_{2}^{+} $ (check the middle stem) ($\vec{v} = (0,1,0)$, $\vec{w} = (0,0,0)$), it becomes impossible to place any term on qubit $3$ to satisfy the algebraic constraint condition $\vec{c} \cdot (\vec{v} - \vec{w}) = 0$. Note (B) is for two small constraints and backtracking quickly discards nonviable branches. This allows us to dramatically prune the search space.}
    \label{fig:tree_backtrack}
\end{figure}

\subsection{Soundness and Completeness for Generating Polynomial Constraint Commutators}\label{subsec:find_coms}

\begin{figure}[!t]
\centering
\setlength{\intextsep}{0.5em}
\begin{minipage}{1.0\linewidth}
\begin{algorithm}[H]
\caption{\algcom$(u, i, d, j, \mathcal{C}, \ell, n, \mathcal{T}_{\ell})$}\label{alg:find_com_terms}
\begin{algorithmic}
\State \textbf{Inputs:} current term $u$, current index $i$, current length $d$, first slot index $j$, constraint list $\mathcal{C}$, maximum length $\ell$, number of variables $n$, commutator collection $ \mathcal{T}_{\ell} $
\If{\texttt{check\_con($u$,$\mathcal{C}$)} is satisfied}
    \If{$u$ not in $ \mathcal{T}_{\ell} $ (under addition)}
        \State add $u$ to commutative term list $ \mathcal{T}_{\ell} $. \Comment{Store valid commutator term.}
    \EndIf
\EndIf
\If{$ i > n $ or $ d = \ell $} \Comment{We've hit end of indices or hit maximum locality weight}
    \State \textbf{return. } \Comment{End search.}
\EndIf
\If{$d > 1$}
\For{$v \in \{ \sigma^{+}, \sigma^{-}, \sigma^{0}, \sigma^{1} \} $}
    \State append $ v $ to $ u $.
    \If{\algcheckcon$(u, \mathcal{C})$ is viable or satisfied} \Comment{See Alg.~\ref{alg:check_con}}
        \State \algcom$(u, i+1, d+1, j, \mathcal{C}, \ell, n, \mathcal{T}_{\ell})$ \Comment{Recurse with one of $\{ \sigma^{+}, \sigma^{-}, \sigma^{0}, \sigma^{1} \}$.}
    \EndIf
    \State pop $ v $ from $ u $. \Comment{Backtracking.}
\EndFor
\State \algcom$(u, i+1, d, j, \mathcal{C}, \ell, n, \mathcal{T}_{\ell})$ \Comment{Recurse with $\mathbbm{1}_{2}$.}
\Else
\For{$v \in \{ \sigma^{+}, \sigma^{-} \} $}
    \State append $v$ to $u$
    \If{\algcheckcon$(u, \mathcal{C})$ is viable or satisfied} \Comment{Placed first $\sigma^{+}$ or $\sigma^{-}$.}
        \State \algcom$(u, 0, d+1, i, \mathcal{C}, \ell, n, \mathcal{T}_{\ell})$ \Comment{Subsequent recursion with $j \leftarrow i$, $d>1$}
    \EndIf
    \State pop $v$ from $u$
    \State \algcom$(u, i+1, d, -1, \mathcal{C}, \ell, n, \mathcal{T}_{\ell})$ \Comment{Recurse with $\mathbbm{1}_{2}$.}
\EndFor
\EndIf
\State \textbf{return. }
\end{algorithmic}
\end{algorithm}
\end{minipage}
\caption*{Alg.~\ref{alg:find_com_terms}: \textbf{Searching for Constraint Commutators}. A backtracking approach to finding all commutative terms up to locality weight $\ell$ for embedded constraint operators. A collection $ \mathcal{T}_{\ell} $ is built by the algorithm. The term $u$ at each recursion is a viable term for $\mathcal{T}_{\ell}$ since it has not violated any algebraic condition associated with the constraints $ C $. The algorithm considers appending $\sigma^+$ or $\sigma^-$ to begin a term, and then appending one of $\{ \sigma^{+}, \sigma^{-}, \sigma^{0}, \sigma^{1} \}$ in subsequent recursion.}
\setlength{\intextsep}{\ointextsep}
\end{figure}

The case for polynomial constraints is a generalization of the case for linear constraints. The search is expanded to include $\sigma^{0}, \sigma^{1}$ and since we are interested in nontrivial commutators, we wish to restrict the search to include at least one $\sigma^{+}, \sigma^{-}$. In words, we recursively construct terms by finding appropriate basis elements at each index starting with the first index. At first, we place a single $\sigma^{+}, \sigma^{-}$ and then in subsequent calls place any of $\{\sigma^{+}, \sigma^{-}, \sigma^{0}, \sigma^{1}\}$. We end the recursion on any branch that becomes nonviable according to the conditions of \cref{eq:thmcom} or if the length would be above the maximum $\ell$. We check if the constraint conditions of \cref{eq:thmcom} are satisfied for each term constructed in each branch and add it to the list $\mathcal{T}_{\ell}$ if so.

As such, the termination, soundness, and completeness of this algorithm follow similarly to the linear constraint setting.

\begin{thm}[Termination of Alg.~\ref{alg:find_com_terms}]
Alg.~\ref{alg:find_com_terms} eventually terminates and does so with the total recursive calls bounded by $\mathcal{O}\left( 4^{\ell} n^{\ell+1} \right)$.
\end{thm}

\begin{proof}[Proof Sketch]
Since each recursive call has an incremented index $i$, eventually a branch will terminate when $i > n$ if not earlier. In fact, in the case that $\ell \geq n$ (no length limit) and there are no constraints, the backtracking algorithm would enumerate all $5^{n}$ possible terms in the set $ \{ \mathbbm{1}_{2}, \sigma^{0}, \sigma^{1}, \sigma^{-}, \sigma^{+} \}^{\otimes n} $.

As in the previous section, recognize that if a term reaches length $\ell$, the branch is ended. For the first location, we have $\binom{n}{1} = n$ slots and may place $\sigma^{+}$ or $\sigma^{-}$ in the slot\footnote{In fact, we could also place only $\sigma^{+}$, since we know from Hermiticity another term $\sigma^{-}$ in the slot should exist.}. As such, the algorithm constructs at most
\begin{align}
2^{\ell} \binom{n}{1} \sum_{k=1}^{\ell-1} 4^{k} \binom{n}{k} &\leq \sum_{k \leq \ell} 4^{\ell} n^{\ell} \nonumber \\
&= (\ell+1) 4^{\ell} n^{\ell} \nonumber \\
&= \mathcal{O}\left( 4^{\ell} n^{\ell} \right)
\end{align}
terms since we have $n-1$ remaining slots that we can place $\ell$ objects into and in each slot we may place one of four operators: $\{ \sigma^{0}, \sigma^{1}, \sigma^{+}, \sigma^{-}\}$. Each possible term requires no more than $n$ recursions to be constructed, so the total number of recursive calls is bound by $ \mathcal{O}\left( 4^{\ell} n^{\ell+1} \right) $.
\end{proof}

Next, we prove that Alg.~\ref{alg:find_com_terms} would consider any valid commutator term in its search.

\begin{thm}[Completeness of Alg.~\ref{alg:find_com_terms}]
If a generator term is in the $\ell$-local commutator set, $g \in \mathcal{T}_{\ell} $, then $g$ will be included in Alg.~\ref{alg:find_com_terms}.
\end{thm}

\begin{proof}[Proof Sketch]
The proof follows from the same logic as in the previous section.  A recursive branch with current term $u$ is pruned from the search if and only if \algcheckcon($u$,$\mathcal{C}$) returned nonviable or the length $d$ has reached the maximum $\ell$. Both cases follow the same logic as in the previous section. The only difference lies in the former, where \algcheckcon($u$,$\mathcal{C}$) utilizes the condition \cref{eq:thmcom} rather than the simpler condition \cref{eq:thmlincom}. This difference is expressed in Alg.~\ref{alg:concond} rather than Alg.~\ref{alg:concondlin} for the call \algconcond$(u,C)$ and that the maximal index for a nonzero entry would be any variable in the higher-order polynomial rather than just a linear polynomial.
\end{proof}

Finally, we prove that Alg.~\ref{alg:find_com_terms} will not include any terms that are not in the valid commutator set.

\begin{thm}[Soundness of Alg.~\ref{alg:find_com_terms}]
If a generator term $g$ is included in Alg.~\ref{alg:find_com_terms}, then $g \in \mathcal{T}_{\ell}$.
\end{thm}

\begin{proof}[Proof Sketch]
The proof again is similar to the linear constraint case. A commutator term $u$ is included in the list only if \algcheckcon($u$,$\mathcal{C}$) is valid and the term would be generated by the recursive call structure. The former is true due to \cref{thm:gencom}. The latter is true since the term $u$ has length less than or equal to $\ell$ and each previous branch included a term that must have been viable.
\end{proof}

\begin{figure}[!t]
\begin{minipage}{1.0\linewidth}
\begin{algorithm}[H]
\caption{\algconcond($u$,$C$)}\label{alg:concondlin}
\begin{algorithmic}
\State \textbf{Inputs:} term $u$, constraint $C$
\State Let $ \vec{c} $ be such that $ C = \sum_{j=1}^{n} c_{j} \, x_{j} $
\State \textbf{return $\vec{c} \cdot \vec{u}$} \Comment{see Thm.~\ref{thm:linconcom}}
\end{algorithmic}
\end{algorithm}
\end{minipage}
\caption*{Alg.~\ref{alg:concondlin}: \textbf{Checking the Linear Constraint Condition.} In Alg.~\ref{alg:check_con} when called from Alg.~\ref{alg:find_lincom_terms}, we utilize this linear condition to determine whether or not a constraint condition is satisfied.}
\begin{minipage}{1.0\linewidth}
\begin{algorithm}[H]
\caption{\algconcond($u$,$C$)}\label{alg:concond}
\begin{algorithmic}
\State \textbf{Inputs:} term $u$, constraint $C$
\State \texttt{comterm} = \{:\} \Comment{empty dictionary}
\State Define $\vec{x}, \vec{y}, \vec{v}, \vec{w}$ from $u$.
\State Let $ \texttt{constraint\_set} = \{ c_{J}, \vec{a}_{J}, \vec{b}_{J} \}_{J \in \mathcal{J}} $ be the collection such that $ C = \sum_{J \in \mathcal{J}} c_{J} \prod_{k=1}^{n} \left( 1 - x_{k} \right)^{a_{Jk}} \left( x_{k} \right)^{b_{Jk}} $.
\For{$ (c_{J}, \vec{a}_{J}, \vec{b}_{J}) \in \texttt{constraint\_set} $}
    \State Compute coefficient $\texttt{co}_{u}$ and binary vectors $\vec{x}',\vec{y}',\vec{v}',\vec{w}'$ using $c_{J}, \vec{a}_{J}, \vec{b}_{J}$ and $\vec{x}, \vec{y}, \vec{v}, \vec{w}$ with Eq.~\ref{eq:thmcom}.
    \State $\texttt{comterm}[\vec{x}',\vec{y}',\vec{v}',\vec{w}'] += \texttt{co}_{u}$
\EndFor
\For{\texttt{key} in \texttt{comterm.keys()}}
    \If{$\texttt{comterm}[\texttt{key}] \neq 0$}
        \State \textbf{return 1}
    \EndIf
\EndFor
\State \textbf{return 0} \Comment{see Thm.~\ref{thm:gencom}}
\end{algorithmic}
\end{algorithm}
\end{minipage}
\caption*{Alg.~\ref{alg:concond}: \textbf{Checking the Polynomial Constraint Condition.} In Alg.~\ref{alg:check_con} when called from Alg.~\ref{alg:find_com_terms}, we utilize this condition to determine whether or not a constraint condition is satisfied. \cref{thm:gencom} gives a checkable condition for any term to see if it commutes with the embedded constraint operator. In particular, the resulting commutator must be zero. By using a dictionary, we can check this condition symbolically with memory resources bounded by the number of constraint terms in the sum.}
\setlength{\intextsep}{\ointextsep}
\end{figure}

\subsection{Generators of the Jordan Algebra of Constraint Preserving Drivers and Splitting into Commutative Sets}\label{subsec:redgen_sequencing}

\begin{figure}[!t]
\centering
\setlength{\intextsep}{0.5em}
\begin{minipage}{1.0\linewidth}
\begin{algorithm}[H]
\caption{Selecting Commutating Generators based on Linear Constraints}\label{alg:gen_com_unis}
\begin{algorithmic}[1]
\Statex \textbf{Inputs: } commutator list $ K $ from Alg.~\ref{alg:find_com_terms} ordered by locality weight \Comment{Note the ordering.}
\State initialize $ \mathcal{G} $ as empty \Comment{$ \mathcal{G} $ holds each block.}
\While{$ K $ is non-empty}
    \For{$ g $ in $ K $}
        \State mark $ g $ for deletion in $ K $
        \If{Alg.~\ref{alg:check_gens} shows $ g $ is not generatable by $ \mathcal{G} $} \Comment{If $g$ is generated by $ \mathcal{G} $, it is not added.}
            \If{$ g $ commutes with all terms in $\mathcal{G}[i]$ for some $ i $}
                \State place $ g $ in $ \mathcal{G}[i] $ and remove from $ K $ \Comment{$g$ is safe to add to block $\mathcal{G}[i]$}
            \Else
                \State place $ [ g ] $ at the end of $ G $ and remove from $ K $ \Comment{$g$ cannot be added to any block $\mathcal{G}[i]$ so extend $\mathcal{G}$.}
            \EndIf
            \State \textbf{break.} \Comment{By breaking, we may favor placing low order terms.}
        \EndIf
    \EndFor
    \State remove all terms marked for deletion in $ K $
\EndWhile
\State \textbf{return} $ \mathcal{G} $.
\end{algorithmic}
\end{algorithm}
\vspace{-1.0em}
\end{minipage}
\caption*{Alg.~\ref{alg:gen_com_unis}: \textbf{Selecting and grouping generators.}
This algorithm reduces a commutator list $K$ to a set of mutually commuting generator groups $G$. Note that we break after placing $g$ to favor utilization of low order terms over higher order terms.}
\setlength{\intextsep}{\ointextsep}
\end{figure}

Given a collection of commutative terms $ \mathcal{T}_{\ell}$, for example found by Alg.~\ref{alg:find_com_terms}, we wish to find a small set of generators $\mathcal{G}$ to express the entire set, $ \mathcal{T}_{\ell} \subset \MG $. Moreover, the generator set should be partitioned into commutative sets of generators:
\begin{align}
    \mathcal{G} = \{ G_{i} , \; \text{such that } [ g_{ij}, g_{ik} ] = 0, \; g_{ij}, g_{ik} \in G_{i} \}.
\end{align}

\quad Each $ G_{i} $ consists of mutually commutative terms that can be jointly applied as unitaries, while terms belonging to different generator sets in general may not commute. Each $ G_{i} $ therefore generates an associative subalgebra of mutually commuting terms. The different generator sets need not commute with each other, so the full structure spanned by addition and anticommutation across all $ G_{i} $ is the Jordan algebra $ \MG $. As such, the algorithmic task is to find an appropriate $ \mathcal{G} $ such that $ | \mathcal{G} | $ is small to allow for maximal parallel application of terms.

\quad Alg.~\ref{alg:gen_com_unis} gives a sound approach whereby low locality weight terms are placed into a growing collection $ \mathcal{G} $. When a term $ t $ is considered for inclusion and no set $ G_{i} $ in the collection $ \mathcal{G} $ commutes with the term, we first check, as discussed in \cref{sec:check_gens}, whether $ t $ can be generated from the terms currently in $ \mathcal{G} $ (building on formalisms discussed in \cref{sec:alg_driver_com}). Alg.~\ref{alg:check_gens} in \cref{sec:check_gens} can be modified to utilize the reduced form \cref{eq:redcomterm}, yielding an approximate pruning variant that may discard more terms, but is no longer guaranteed to preserve soundness. In the form used here, however, Alg.~\ref{alg:check_gens} is applied in its exact form, so the reduction step remains sound. For example, given a term $ g_{3} = \sigma_{1}^{+} \sigma_{3}^{-} + \sigma_{1}^{-} \sigma_{3}^{+} $ and a collection $ g_{1} = \sigma_{1}^{+} \sigma_{2}^{-} + \sigma_{1}^{-} \sigma_{2}^{+}, g_{2} = \sigma_{2}^{+} \sigma_{3}^{-} + \sigma_{2}^{-} \sigma_{3}^{+} $, we can see that $ \{ g_{1}, g_{2} \} = g_{3} $ and so $ g_{3} $ would be expressed in the Jordan algebra.

\quad Note that even though $ g_{3} $ has no support on qubit $ 2 $, we had to consider anticommutation relations associated with the location. The result of Alg.~\ref{alg:gen_com_unis} is the collection $ \mathcal{G} $, which is used to define a driver Hamiltonian or mixer operator in \cref{subsec:gen_to_proj}.

\subsection{From Generators to Projectors to Unitaries}\label{subsec:gen_to_proj}

Given a set of commutative generators, $ G = \{ g_{1}(\vec{x}_{1}, \vec{y}_{1}, \vec{v}_{1}, \vec{w}_{1}), \ldots, g_{|G|}(\vec{x}_{|G|}, \vec{y}_{|G|}, \vec{v}_{|G|}, \vec{w}_{|G|}) \} $, for example found through Alg.~\ref{alg:find_com_terms} followed by Alg.~\ref{alg:gen_com_unis}, we map each term into intermediate rank-1 projection operators (density operators) that are then mapped to parameterized unitaries. The resulting mixers act as a collection of rotations, each around a specific density operator related to the symmetry found.

\quad Given a generator $ g(\vec{x}, \vec{y}, \vec{v}, \vec{w}) $, we define

\begin{align}
Q(\theta, \vec{v},\vec{w}) &= \frac{1}{2} \prod_{j\in \text{support}(\vec{v},\vec{w})} \left( \sigma_{j}^{0} \right)^{v_{j}} \left( \sigma_{j}^{1} \right)^{w_{j}}
+ \frac{e^{i \theta}}{2} \prod_{j\in \text{support}(\vec{v},\vec{w})} \left( \sigma_{j}^{+} \right)^{v_{j}} \left( \sigma_{j}^{-} \right)^{w_{j}} \nonumber \\
&\phantom{= } + \frac{e^{-i \theta}}{2} \prod_{j\in \text{support}(\vec{v},\vec{w})} \left( \sigma_{j}^{-} \right)^{v_{j}} \left( \sigma_{j}^{+} \right)^{w_{j}}
+ \frac{1}{2} \prod_{j\in \text{support}(\vec{v},\vec{w})} \left( \sigma_{j}^{1} \right)^{v_{j}} \left( \sigma_{j}^{0} \right)^{w_{j}},
\end{align}

up to a choice in phase $ \theta $ and that can be shown to be a density operator over the qubits in the support of $ \vec{v}, \vec{w} $. We define $\text{support}(\vec{v},\vec{w}) = \text{support}(\vec{v}) \cup \text{support}(\vec{w}) $, where $\text{support}(\vec{v})$ denotes the set of indices with nonzero entries in $\vec{v}$. Note that $\text{support}(\vec{v})$ and $\text{support}(\vec{w})$ are disjoint. We often consider the simplified form $ Q(\vec{v}, \vec{w}) = Q(\z, \vec{v}, \vec{w}) $, although expressions can be generalized to include further parameterizations with a $ \theta $ for each $ Q $.

\quad Then we define the projector:
\begin{align}
P(\vec{x}, \vec{y}, \vec{v}, \vec{w}) &=   T(\vec{x}, \vec{y}, \z, \z) Q(\vec{v}, \vec{w}),
\end{align}

which can be shown to also be a density operator over the support of $ \vec{x}, \vec{y}, \vec{v}, \vec{w} $.

\quad Then we define the single parameter unitary:
\begin{align}\label{eq:diffusor_mixer}
U(\beta, \vec{x}, \vec{y}, \vec{v}, \vec{w}) = \left( \mathbbm{1} + (e^{-i \beta }-1) P(\vec{x}, \vec{y}, \vec{v}, \vec{w}) \right),
\end{align}

which acts like an identity operator over the entire Hilbert space, except for a difference in phase over the projection operator $ P $. We call mixing operators of this form as diffusor mixers. Because $ P $ is a projection operator, the unitary is equivalent expressible as $ e^{-i \beta P(\vec{x},\vec{y},\vec{v},\vec{w})} $.

\quad Then the overall resulting mixing operator associated with a collection $ \mathcal{G} $ of generator sets is:
\begin{align}\label{eq:compiled_uni}
U(\beta, \mathcal{G}) = \prod_{G_{i} \in \mathcal{G}} \prod_{g(\vec{x}, \vec{y}, \vec{v}, \vec{w}) \in G_{i}} U(\beta, \vec{x}, \vec{y}, \vec{v}, \vec{w}).
\end{align}

\quad Then given a collection of valid commutative generators, found for example as discussed in \cref{subsec:redgen_sequencing}, we can see that the resulting projectors and thereby unitaries are built on mutually commutative terms. Diffusor mixers such as these were discussed in Ref.~\cite{leipold2022tailored} for the circuit fault diagnosis problem.

\quad For a Hamiltonian driver we can use:
\begin{align}\label{eq:compiled_ham}
    H(\mathcal{G}) = \sum_{G_{i} \in \mathcal{G}} \sum_{g_{ij} \in G_{i}} \alpha_{ij} g_{ij} + \alpha_{ji}^{\dg} g_{ij}^{\dg}.
\end{align}

\m

\quad A more general formulation would allow individualized parameters for each term, leading to a richer and more expressive learning space. Notice that such individual terms are not forced to be stoquastic.

\quad Equivalent action of the X-mixer is achieved through the rotation around $ \ket{+}_{j} \bra{+}_{j} $ (or $ \ket{-}_{j} \bra{-}_{j} $) on location $ j $ up to a global phase:
\begin{align}\label{eq:}
U_{X_{j}}(\beta) &= \mathbbm{1} + (e^{-i \beta} - 1) \ket{+}_{j} \bra{+}_{j} \nonumber \\
&= \mathbbm{1} + \frac{1}{2} (e^{-i \beta} - 1) (\sigma_{j}^{0} + \sigma_{j}^{+} + \sigma_{j}^{-} + \sigma_{j}^{1}).
\end{align}

Note that $ \ket{+}_{j} \bra{+}_{j} $ is a density operator over qubit $ j $ and a projection over $n$ qubits, with $U_{X_{j}}$ phasing this subspace. Then the X-mixer is given by phasing each such disjoint subspace:
\begin{align}
U_{X}(\beta) = \prod_{j=1}^{n} U_{X_{j}}(\beta).
\end{align}

The associated driver Hamiltonian would be the standard transverse field:
\begin{align}
H_{X}(\beta) = - \sum_{j=1}^{n} \beta \, \left( \sigma_{i}^{+} + \sigma_{i}^{-} \right) = -\sum_{j=1}^{n} \beta \, \sigma_{i}^{x}.
\end{align}

We can also recognize that $ \sigma_{j}^{x} = \sigma_{j}^{+} + \sigma_{j}^{-} $ is associated with $ \vec{x} = \z, \vec{y} = \z, \vec{v} = \mathbf{e}_{j}, \vec{w} = \z $. It would result from Alg.~\ref{alg:find_com_terms} or Alg.~\ref{alg:find_lincom_terms} by considering $1$-local operators with no constraint satisfaction requirement.

\quad Another relevant example is the single global embedded constraint studied in the context of many important problems, including quantum chemistry, portfolio optimization, and graph partitioning~\cite{hen2016quantum,hadfield2019quantum}. The embedded constraint operator is $ \sum_{i} \sigma_{i}^{0} $ (or equivalent forms).

\quad A simple mixer that irreducibly commutes (see \cref{sec:complexity}) is the ring XY mixer associated with partial swaps or hopping terms:
\begin{align}
U_{j}(\beta) &= \mathbbm{1} + (e^{-i \beta} - 1) P_{j}^{XY} \nonumber\\
&= \mathbbm{1} + \frac{1}{2} (e^{-i \beta} - 1) \left( \ket{0}_{j} \ket{1}_{j+1} + \ket{1}_{j} \ket{0}_{j+1} \right) \left( \bra{0}_{j} \bra{1}_{j+1} + \bra{1}_{j} \bra{0}_{j+1} \right) \nonumber\\
&= \mathbbm{1} + \frac{1}{2} (e^{-i \beta} - 1) \left( \sigma_{j}^{0} \sigma_{j+1}^{1} + \sigma_{j}^{+} \sigma_{j+1}^{-} + \sigma_{j}^{-} \sigma_{j+1}^{+} + \sigma_{j}^{1} \sigma_{j+1}^{0} \right),
\end{align}

with $ n+1 $ defined to be equivalent to $ 1 $. We can associate $ P_{j}^{XY} $ with $ \vec{x} = \z, \vec{y} = \z, \vec{v} = \mathbf{e}_{j}, \vec{w} = \mathbf{e}_{j+1} $ and check that it satisfies \cref{thm:algconcom} as well as \cref{thm:linconcom}. A collection of commutative generators to accomplish irreducible commutation can be applying the ring XY mixer on odd locations followed by even locations:
\begin{align}
U_{\text{ring-XY}}(\beta) = \left( \prod_{j=1}^{n/2} U_{2\,j}(\beta) \right) \left( \prod_{j=1}^{n/2} U_{2\,j-1}(\beta) \right)
\end{align}

This mixer corresponds to the XY ring driver Hamiltonian~\cite{hen2016quantum,leipold2021constructing,kordonowy2026lie}:
\begin{align}
H_{XY} &= -\sum_{j=1}^{n} \left( \sigma_{j}^{+} \sigma_{j+1}^{-} + \sigma_{j}^{-} \sigma_{j+1}^{+} \right) \nonumber \\
&= -\sum_{j=1}^{n} \left( \sigma_{j}^{x} \sigma_{j+1}^{x} + \sigma_{j}^{y} \sigma_{j+1}^{y} \right) .
\end{align}

\begin{figure}
\includegraphics[width=0.8\textwidth]{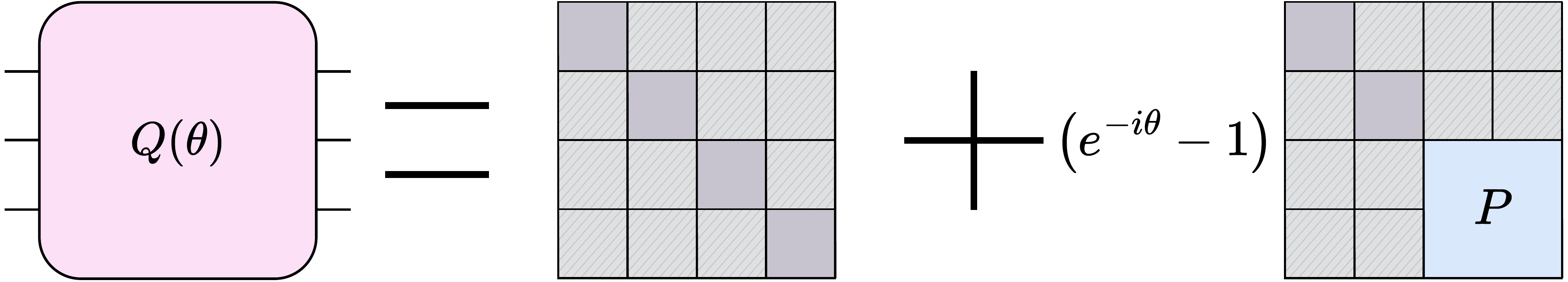}
\caption{\textbf{Diffusor Mixer Decomposition.} A unitary $Q(\theta)$ is a diffusor mixer if it has the decomposition from \cref{eq:diffusor_mixer}, a form that most mixers used in practice have. Generators found in \cref{subsec:redgen_sequencing} can easily be sequenced into such unitaries.}
\label{fig:diffusor_mixer}
\end{figure}

\section{Three Ansatz Constructions for QAOA Targeting Random 1-in-3 SAT}\label{sec:qaoa_1in3}

In the previous Sections~\ref{sec:algecond}-\ref{sec:theory_to_prac}, we have presented the mathematical foundations of enforcing constraints on quantum evolution by imposing the constraints on the quantum operators applied to the system, illustrated the complexity of finding such terms, specified how such terms form a Jordan algebra, and outlined how mixing operators can be constructed from such terms through algorithms detailed in \cref{sec:theory_to_prac} (see \cref{fig:flowchart}). In this Section, we apply QAOA to a well-studied NP-Complete feasibility problem: 1-in-3 SAT.

\quad We consider three mixers to help solve this problem: the X-mixer, a novel mixer built around imposing the maximum disjoint subset (MDS) of constraints, and a mixer built on imposing constraints on shared variables for each constraint in the MDS. Results from benchmarking each approach in \cref{subsec:bench1in3} demonstrate the practical benefits of constraint-enforcing mixers for solving feasibility problems. Similar constructions could be used for a wide range of optimization problems.

\subsection{Random 1-in-3 SAT Instances around the Satisfiability Threshold}\label{subsec:rand_1in3}

\begin{figure}
\centering
\begin{tabular}{c}
\includegraphics[height=0.17\textheight]{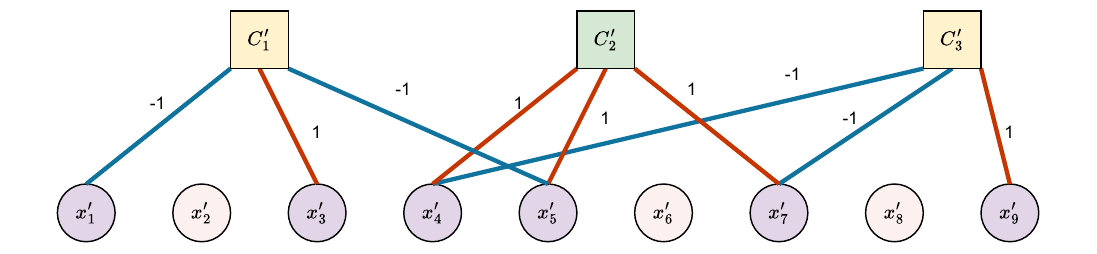} \\
\\ (A) \\
\\
\includegraphics[height=0.17\textheight]{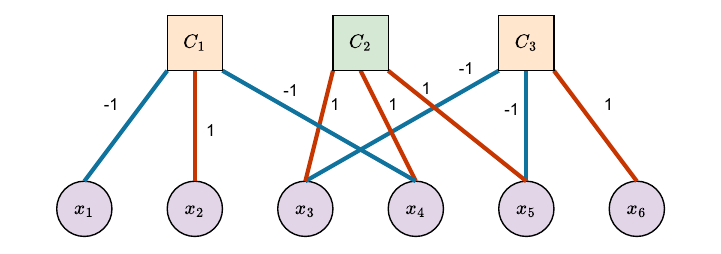} \\
\\ (B)
\end{tabular}
\caption{\textbf{Instance of 1-in-3 SAT.} (A) shows a graphical depiction of an instance of 1-in-3 SAT with 9 variables, $ \{ x_{1}^{'}, \ldots, x_{9}^{'} \} $, and 3 clauses. Variables $ x_{2}^{'}, x_{6}^{'}, x_{8}^{'} $ are not used in any clauses and clauses $ C_{1}' $ and $ C_{3}' $ form the maximum disjoint set of clauses (MDS) for this problem. (B) shows a graphical depiction of the same instance over the reduced variable space. This specific instance is discussed in \cref{subsec:rand_1in3} as an example of random 1-in-3 SAT instances, in \cref{subsec:qaoa_mdsmixer} as an instance for which explicit constructions of the initial wavefunction and the mixing operators associated with the MDS are given, and in \cref{sec:algecond} as an example of the automated approach to developing the symmetric cover mixer (SymCov).}
\label{fig:con_bipart}
\end{figure}

Given a collection of $ m $ clauses, each with $ 3 $ literals, we wish to find an assignment to all $ n $ variables such that exactly one of the literals in each clause is satisfied. This means exactly $ 1 $ literal is satisfied and $ 2 $ are unsatisfied for each clause in the collection. 1-in-3 SAT can be formulated as the following binary linear program (BLP) with no optimization function since it is a feasibility problem:
\begin{align}
C^{M} \cdot \vec{x} = \vec{b}, \text{ s.t. } x \in \{0,1\}^n, \; C^{M} \in \{-1,0,1\}^{m \times n}, \; b \in \{-2,-1,0,1\}^{m} \text { and } C^{M} \cdot (1,\ldots,1)^{T} = \abs{b}.
\end{align}

Given a collection of constraints, $ \mathcal{C} = \{ C_{1}, \ldots, C_{m} \} $ with $ m = \lceil \frac{n}{3} \rceil $ and $ C_{i}(x) = \sum_{j=1}^{n} c_{ij} \, x_{j} $ set to constraint value $ b_{i} = 1 - \sum_{j=1}^{n} c_{ij}^{2} \, (1 - c_{ij})/2 $ and $ \sum_{j=1}^{n} \abs{c_{ij}} = 3 $. Then the linear constraint has the form $ \vec{c_{i}} \cdot \vec{x} = b_{i} $ for $ b_{i} = 1 - \sum_{j=1}^{n} c_{ij}^{2} \,  (1 - c_{ij})/2 $. Then let $ C^{M} $ be the row concatenation of $ \vec{c_{1}}, \ldots, \vec{c_{m}} $ and $ \vec{b} = (b_{1}, \ldots, b_{m}) $. This means we can interpret each $C_{i}$ with a valid 1-in-3 constraint and $C^{M}$ as enforcing each, as we will show.

\quad Each constraint $ C_{i} $ can be associated with a clause, $ E_{i} = \{ l_{i1}, l_{i2}, l_{i3} \} $, with literals $ l_{ij} = (c_{ij}, v_{j}) $ such that the polarity $ c_{ij} \in \{-1, 1\} $ associates $ 1 $ (-1) with a positive literal (negative literal) of the atom and $ v_{j} $ is the index of the associated atom (propositional variable) in $ X $. Since $C_{i}(x)$ has the constraint $\sum_{j=1}^{n} |c_{ij}| = 3$, we can interpret the $3$ non-zero entries as the polarities of the clause $E_{i}$.

\quad Satisfying a literal is associated with $(1-x_{j})$ being $1$ for $c_{ij} = -1$ or $x_{j}$ being $1$ for $c_{ij} = 1$. A linear function $l(x) $ for literal $l$ is:
\begin{align}
(1-c_{ij}) (1-x_{j}) / 2 + (1+c_{ij}) (x_{j}) / 2 &= 1/2 (1-c_{ij} - x_{j} + c_{ij} \, x_{j} + x_{j} + c_{ij} \, x_{j}) \nonumber \\
&= 1/2 (1-c_{ij}) + c_{ij} \, x_{j} .
\end{align}

To set the function to $ 0 $ if $ c_{ij} = 0 $ and unaltered otherwise:
\begin{align}
c_{ij}^2 ( 1 / 2 (1 - c_{ij}) + c_{ij} \, x_{j} ) &= c_{ij}^{2} (1 - c_{ij}) / 2 + c_{ij}^{3} \, x_{j} \nonumber \\
&= c_{ij}^{2} (1 - c_{ij}) / 2 + c_{ij} \, x_{j}.
\end{align}

Then 1-in-3 SAT stipulates only one of the literals may be set to $1$:
\begin{align}
\sum_{j=1}^{n} l_{ij}(x) &= 1 \nonumber \\
\sum_{j=1}^{n} c_{ij}^{2} \, (1-c_{ij}) / 2 + c_{ij} \, x_{j} &= 1 \nonumber \\
\sum_{j=1}^{n} c_{ij} \, x_{j} &= 1 - \sum_{j=1}^{n} c_{ij}^{2} \, (1-c_{ij}) / 2 .
\end{align}

Then the condition $|c_{ij}| = 3$ ensures each clause has exactly $3$ literals, while the constraint value $ 1 - \sum_{j} c_{ij}^{2} (1-c_{ij}) / 2 $ enforces that only one literal is satisfied.

\quad The transition from likely to be satisfiable to likely to be unsatisfiable for random 1-in-k SAT problems occurs with $ n/\binom{k}{2} $ clauses~\cite{achlioptas2001phase}. For random  1-in-3 SAT, this occurs with n/3 clauses. To sample from this random distribution, we construct $ m = \lceil \frac{n}{3} \rceil $ random clauses such that $ c_{ij} \in \{ -1, 1 \} $ uniformly at random over 3 random locations indexed by $ j $.

\quad Given that not every variable may participate in the $\lceil \frac{n}{3} \rceil$ clauses, we begin with clauses $ C_{i}^{'} $ over variables $ X^{'} = \{ x_{1}^{'}, \ldots, x_{n}^{'} \} $ and $ C_{i} $ is the constraint over the reduced variable space $ X = \{ x_{1}, \ldots, x_{|X|} \} $ by filtering out variables not used (utilizing the unapostrophed definitions for the reduced case keeps the notation cleaner in later sections).

\quad An instance of the 1-in-3 SAT problem is shown in \cref{fig:con_bipart}.A with 9 variables, $ \{ x_{1}^{'}, \ldots, x_{9}^{'} \} $, and 3 clauses: $ C_{1}^{'} = (-1,0,1,0,-1,0,0,0,0) $, $ C_{2}^{'} = (0,0,0,1,1,0,1,0,0) $, and $ C_{3}^{'} = (0,0,0,-1,0,0,-1,0,1) $. Variables $ x_{2}^{'}, x_{6}^{'}, x_{8}^{'} $ are not used in any clauses. For random instances of 1-in-3 SAT with $ \lceil n/3 \rceil $ clauses, around the satisfiability threshold, it is typical that some variables will not be used.

\quad The reduced instance of the problem is represented in \cref{fig:con_bipart}.B with $ C_{1} = (-1,1,0,-1,0,0), b_1 = -1 $, $ C_{2} = (0,0,1,1,1,0), b_2 = 1 $, $ C_{3} = (0,0,-1,0,-1,1), b_3 = -1 $ over the reduced variables $ X = \{ x_{1}, \ldots, x_{6} \} $. We use this example for explicit constructions of the cost operator Hamiltonian in \cref{subsec:qaoa_xmixer}, the maximum disjoint mixer in \cref{subsec:qaoa_mdsmixer}, and the symmetric cover mixer in \cref{subsec:qaoa_mdsmixer}. This instance has two solutions over the reduced variables: $ (1, 0, 0, 0, 1, 0), ( 1, 0, 1, 0, 0, 0) $.

\subsection{Quantum Approximate Optimization Algorithm for 1-in-3 SAT}\label{subsec:qaoa_xmixer}

Here we describe QAOA with the X-mixer at a high level for 1-in-3 SAT. Let $\sigma^{0} = (1/2)(\mathbbm{1} + \sigma^{z})$, $\sigma^{1} = (1/2)(\mathbbm{1} - \sigma^{z})$, and $\alpha, \beta \in [ 0, 2  \pi ]^{p}$. Then:
\begin{align}
    U_{p}(\alpha, \beta) = U_{\text{mixer}}(\beta_{p}) U_{\text{cost}}(\alpha_{p}) \,  \ldots \, U_{\text{mixer}}(\beta_1) \, U_{\text{cost}}(\alpha_1) .
\end{align}

\quad The cost of a clause, $ E_{i} $, is minimized by satisfying one literal and unsatisfying the other literals for 1-in-k SAT:
\begin{align}
    H_{\text{clause}}^{i} =  \sum_{ j } \left( \mathbbm{1} - \sum_{l_{jk} \in E_{i}} \sigma_{v_{k}}^{(1 + c_{jk})/2} \prod_{j \neq i} \sigma_{v_{j}}^{(1 - c_{jk})/2} \right),
\end{align}

where $ \sigma_{v_{j}}^{(1 + c_{ij})/2} $ for a positive literal $(1,v_j)$ would be $\sigma_{v_{j}}^{1} $ signifying a reward for satisfying that literal and likewise $\sigma_{v_{j}}^{(1 - c_{ij})/2} $ for a positive literal $(1,v_j)$ would be $\sigma_{v_{j}}^{0} $ to reward failing to satisfy that literal.

\begin{figure}
    \centering
    \includegraphics[width=0.4\linewidth]{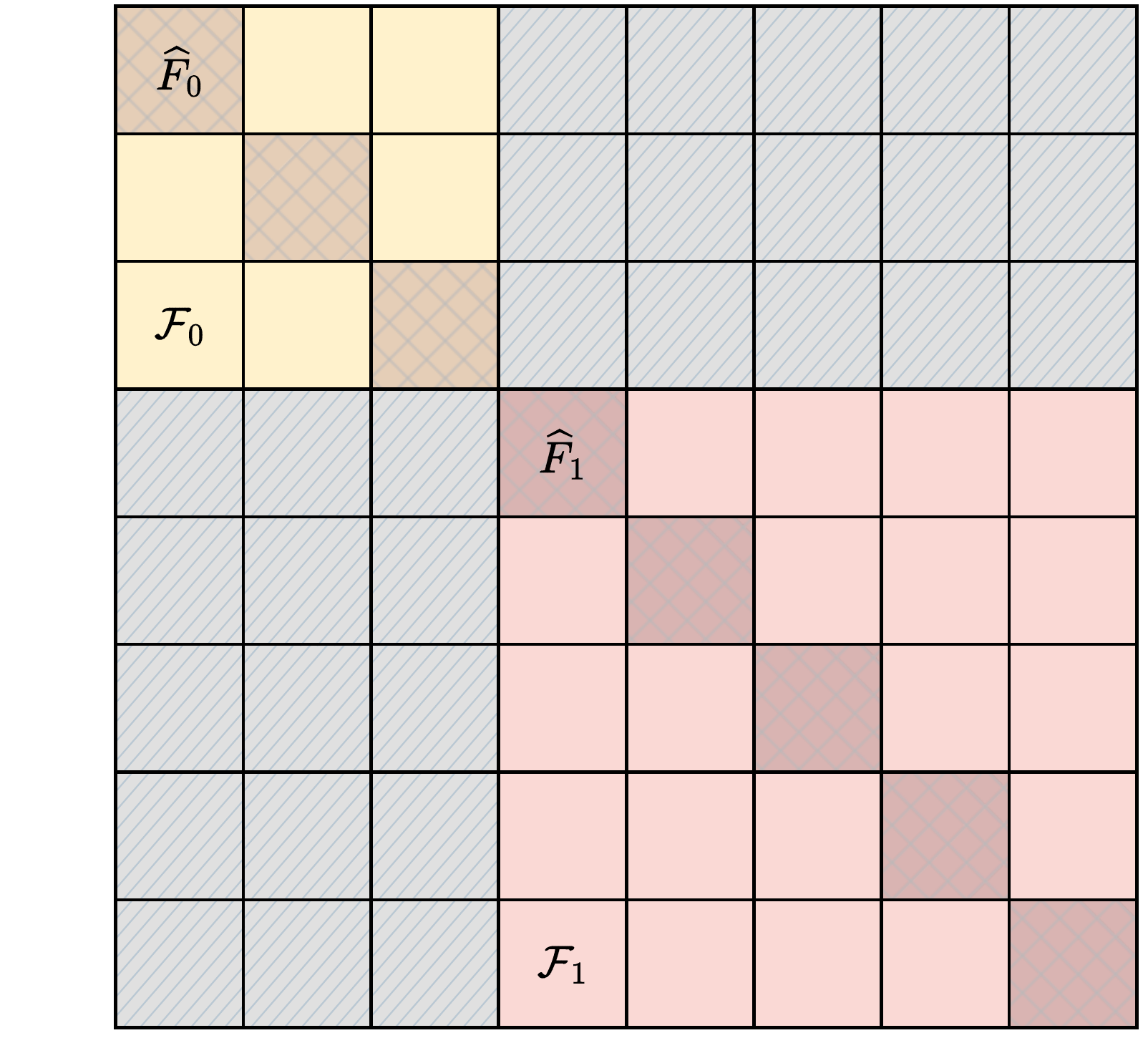}
    \caption{\textbf{Constraint Space of a Single 1-in-3 SAT Clause.} For any SAT clause with 3 variables, the 1-in-3 formulation is associated with the depicted commutation space. Any such clause is decomposed into $ 0 \, \hat{F}_{0} + 1 \, \hat{F}_{1}$ with precisely $3$ assignments in $F_{0}$ and $5$ assignments in $F_{1}$ over the $3$ variables present in the clause. }
    \label{fig:single_clause}
\end{figure}

For example, for $ E_{1} = \{ (-1, x_1), (1, x_2), (1, x_4) \} $ the associated reward for satisfying $(-1, x_1)$ and not satisfying $(1,x_2), (1, x_4)$ is $ \sigma_{1}^{0} \, \sigma_{2}^{0} \, \sigma_{3}^{0} $, for satisfy $ (1, x_2) $ and not satisfying $ (-1,x_1), (1,x_4) $ is $ \sigma_{2}^{1} \, \sigma_{1}^{1} \, \sigma_{4}^{0} $, for satisfying $(1,x_4)$ and not satisfying $(-1,x_1), (1,x_2)$ is $ \sigma_{4}^{1} \, \sigma_{1}^{1} \, \sigma_{2}^{0} $. This yields the overall cost $ \mathbbm{1} - \sigma_{1}^{0} \, \sigma_{2}^{0} \, \sigma_{3}^{0} - \sigma_{2}^{1} \, \sigma_{1}^{1} \, \sigma_{4}^{0} - \sigma_{4}^{1} \, \sigma_{1}^{1} \, \sigma_{2}^{0} $, which has two unique eigenvalues, $0$ for the assignments that satisfy exactly one literal in the clause and $1$ for the assignments that fail to satisfy exactly one literal. \cref{fig:single_clause} depicts the eigendecomposition of a single clause of the space $\mathbb{C}^{8 \times 8}$ associated with the 3 (atom) variables in that clause.

\quad Consider the instance introduced in \cref{subsec:qaoa_mdsmixer} and depicted in \cref{fig:con_bipart}. In this case, $ \text{ker}(H_{\text{clause}}^{1}) = \text{span}(\{ \ket{0,0,1,x_{4},x_{5},x_{6}}, \ket{1,0,0,x_{4},x_{5},x_{6}}, \ket{1,1,1,x_{4},x_{5},x_{6}} \}) $, where $ x_{4}, x_{5}, x_{6} \in \{ 0 , 1 \} $, representing the solution subspace for constraint $ C_{1} $.

\quad Then the phase-separating operator for a parameter $ \alpha_{l} $ is:
\begin{align}
    U_{\text{cost}}(\alpha_{l}) = \prod_{\text{clause $E_{j}$ in clauses}} e^{i \, \alpha_{l} \, H_{\text{clause}}^{j} }
\end{align}

Let $ |+\rangle_{k} $ represent the plus state over qubit $ k $. For the mixing operator, we have the X rotations per qubit for $ \beta_{l} $ (written slightly differently than popular as a diffusor):
\begin{align}
    U_{\text{mixer}}(\beta_{l}) = \prod_{j=1}^{n} \left( \mathbbm{1} + (e^{-i \, \beta_{l}} - 1) |+\rangle_{j}\langle+|_{j} \right).
\end{align}

\quad Notice we place the negative sign for the exponent of the mixing operator and the positive sign for the exponent of the phase-separating operator. The initial wavefunction is:
\begin{align}
\ket{\phi(0)} = \ket{+} \ldots \ket{+} = \sum_{x \in \{0,1\}^{n}} \frac{1}{\sqrt{2^{n}}} \ket{x}
\end{align}

\quad And so the final wavefunction is:
\begin{align}
\ket{\phi(p)} = U_{p}(\alpha, \beta) \ket{\phi(0)}
\end{align}

\subsection{Maximum Disjoint Clauses and the Quantum Alternating Operator Ansatz}\label{subsec:qaoa_mdsmixer}

This section details a tailored ansatz approach for solving the 1-in-3 SAT problem using QAOA based on the maximum disjoint set of constraints. When constraints are disjoint, they can be individually solved. Every 1-in-3 SAT clause has precisely 3 solutions out of 8 configurations, just like every 3 SAT clause has precisely 7 solutions out of 8 configurations. The configuration space, $ X $, has $ 2^{n} $ possible configurations, but restricting the space over each constraint in a disjoint set reduces the space by $ 3/8 $. Therefore, given $ k $ constraints in the disjoint set, the resulting feasibility constraint space size is $ (3/8)^{k} 2^{n} $. This section constructs a tailored ansatz utilizing the maximum disjoint set of constraints and \cref{subsec:bench1in3} supports that this approach provides empirical evidence for improved scaling relative to the X-mixer baseline.

\quad Here we describe the first tailored ansatz for the QAOA on 1-in-3 SAT. Given $ m $ linear constraints $ \mathcal{C} = \{ C_{1}, \ldots, C_{m} \}$ over $ n $ (reduced) variables $ X = \{ x_{1}, \ldots, x_{n} \} $, we find the maximum disjoint set of constraints (MDS) over the $ n $ variables. The MDS is the largest subset of the constraints such that no two constraints share variables (i.e. they are disjoint).

\quad In general, the maximum disjoint set problem is NP-Hard, but in practice over sparse random instances of constraint problems it is tractable in practice. Approximations to the maximum disjoint set problem also exist for cases where finding the maximum disjoint set appears too expensive to compute through standard solvers for integer linear programming. Let $ \mathcal{D} = \{ C_{D_1}, \ldots, C_{D_{|\mathcal{D}|}} \}$ be this set. Let $ \mathcal{R} = \mathcal{C} - \mathcal{D} $ be the remaining clauses. Let $ \mathcal{N} $ be the set of variables not in any constraint in $ \mathcal{D} $ (since this is the reduced form of the problem, each variable in $ \mathcal{N} $ appears at least once in the remaining clauses $ \mathcal{R} $).

\quad For each clause, $ C_{D_{i}} $, we find all solutions (there are always three) to define $ S_D = \left\{ S(C_{D_1}), \ldots, S(C_{D_{|D|}}) \right\} $ . Then we associate the uniform superposition density operator for each solution set in $ S_{D} $:
\begin{align}
\hat{S}_{D} &= \left\{ \hat{S}(C_{D_1}), \ldots, \hat{S}(C_{D_{|D|}}) \right\}, \\
\text{where } \hat{S}(C_{D_j}) &= \frac{1}{3} \Big( \ket{s_{1}^{j}} + \ket{s_{2}^{j}} + \ket{s_{3}^{j}} \Big)\Big( \bra{s_{1}^{j}} + \bra{s_{2}^{j}} + \bra{s_{3}^{j}} \Big)
\end{align}

\quad Here each 1-in-3 SAT constraint, $ C_{D_{i}} $, in the disjoint set has precisely 3 solutions, $ s_{1}^{i}, s_{2}^{i}, s_{3}^{i} $, over three variables. Hence, each $ \hat{S}(C_{D_{i}}) $ is a 3-local diffusor. An example based on the 1-in-3 SAT instance visualized in \cref{fig:con_bipart} is described later in this section.

\quad Then the phase-separating operator, given a specific $ \alpha_{l} $, is restricted to the constraints in the set $ \mathcal{R} $ instead of $ \mathcal{C} $:
\begin{align}
    U_{\text{cost}}(\alpha_{l}) = \prod_{C_{j} \in R} e^{i \, \alpha_{l} \, H_{\text{clause}}^{j} }
\end{align}

\quad For the mixing operator, we use the solutions of each disjoint set in $ D $ with the shared layer parameter $ \beta_{l} $:
\begin{align}
    U_{\text{D}}(\beta_{l}, \mathcal{D}) = \prod_{\hat{S}(K) \in \hat{S}_{D}} \left( \mathbbm{1} + (e^{-i \, \beta_{l}} - 1) \hat{S}(K) \right),
\end{align}

and the ordinary single qubit diffusor ($X$ rotation) on the variables in $ \mathcal{N} $:
\begin{align}
    U_{\text{N}}(\beta_{l}, \mathcal{N}) = \prod_{j \in \mathcal{N}} \left( \mathbbm{1} + (e^{-i \, \beta_{l}} - 1) |+\rangle_{j}\langle+|_{j} \right),
\end{align}

Then the overall mixing operator is:
\begin{align}
    U_{\text{mds}}(\beta_{l}, \mathcal{C}) = U_{\text{D}}(\beta_{l}, \mathcal{D}) U_{\text{N}}(\beta_{l}, \mathcal{N}),
\end{align}
following the construction of \cref{eq:compiled_uni}.

\quad The initial state is then a uniform superposition over all disjoint clauses, which is easy to prepare.

\quad Consider the instance discussed in \cref{subsec:rand_1in3} and depicted in \cref{fig:con_bipart}.B, the MDS is $ \{ C_{1}, C_{3} \} $. The resulting relevant feasible subspace projector is given by
\begin{align}
\hat{F}_{0,0} &= \left( \sum_{s_1 \in S(C_{1}))} s_1 \right) \otimes \left( \sum_{s_2 \in S(C_{2})} s_2 \right) \nonumber \\
&= \Big( \ketbra{001}{001} + \ketbra{100}{100} + \ketbra{111}{111} \Big) \otimes \Big( \ketbra{010}{010} + \ketbra{100}{100} + \ketbra{111}{111} \Big) \nonumber \\
&= \ketbra{001010}{\text{`'}} + \ketbra{001100}{\text{`'}} + \ketbra{001111}{\text{`'}} + \ketbra{100010}{\text{`'}} + \ketbra{100100}{\text{`'}}  \nonumber \\
&\phantom{= } + \ketbra{100111}{\text{`'}} + \ketbra{111010}{\text{`'}} + \ketbra{111100}{\text{`'}} + \ketbra{111111}{\text{`'}} ,
\end{align}
where $\ketbra{x}{\text{`'}} \equiv \ketbra{x}{x}$ to improve readability. The resulting density operators are:
\begin{align}
\hat{S}(C_{1}) &= \frac{1}{3} \Big( \ket{001} + \ket{100} + \ket{111} \Big) \Big( \bra{001} + \bra{100} + \bra{111} \Big) \otimes \mathbbm{1}_{4:6}, \\
\hat{S}(C_{3}) &= \frac{1}{3} \mathbbm{1}_{1:3} \otimes \Big( \ket{010} + \ket{100} + \ket{111} \Big) \Big( \bra{010} + \bra{100} + \bra{111} \Big).
\end{align}

In this case, $ \mathcal{R} = \{ C_{2} \} $ and there are no uncovered variables. Then the initial wavefunction is:
\begin{align}
    \ket{ \phi(0) } = \frac{1}{\sqrt{9}} \Big( \ket{001} + \ket{100} + \ket{111} \Big) \Big( \ket{010} + \ket{100} + \ket{111} \Big).
\end{align}

The mixing operator is:
\begin{align}
U_{\text{mds}}(\beta, \mathcal{C}) = \left( \mathbbm{1} + (e^{-i \beta} - 1) \, \hat{S}(C_{1}) \right) \left( \mathbbm{1} + (e^{-i \beta} - 1) \, \hat{S}(C_{3}) \right).
\end{align}

The mixing operator written explicitly in the relevant search subspace as a 9$\times$9 matrix:

\m
\centerline{
\begin{minipage}{1.5\linewidth}
\begin{align}
\begin{pmatrix}
f(\beta)^{2} & f(\beta)g(\beta) & f(\beta)g(\beta) & f(\beta)g(\beta) & g(\beta)^{2} & g(\beta)^{2} & f(\beta)g(\beta) & g(\beta)^{2} & g(\beta)^{2} \\
f(\beta)g(\beta) & f(\beta)^{2} & f(\beta)g(\beta) & g(\beta)^{2} & f(\beta)g(\beta) & g(\beta)^{2} & g(\beta)^{2} & f(\beta)g(\beta) & g(\beta)^{2} \\
f(\beta)g(\beta) & f(\beta)g(\beta) & f(\beta)^{2} & g(\beta)^{2} & g(\beta)^{2} & f(\beta)g(\beta) & g(\beta)^{2} & g(\beta)^{2} & f(\beta)g(\beta) \\
f(\beta)g(\beta) & g(\beta)^{2} & g(\beta)^{2} & f(\beta)^{2} & f(\beta)g(\beta) & f(\beta)g(\beta) & f(\beta)g(\beta) & g(\beta)^{2} & g(\beta)^{2} \\
g(\beta)^{2} & f(\beta)g(\beta) & g(\beta)^{2} & f(\beta)g(\beta) & f(\beta)^{2} & f(\beta)g(\beta) & g(\beta)^{2} & f(\beta)g(\beta) & g(\beta)^{2} \\
g(\beta)^{2} & g(\beta)^{2} & f(\beta)g(\beta) & f(\beta)g(\beta) & f(\beta)g(\beta) & f(\beta)^{2} & g(\beta)^{2} & g(\beta)^{2} & f(\beta)g(\beta) \\
f(\beta)g(\beta) & g(\beta)^{2} & g(\beta)^{2} & f(\beta)g(\beta) & g(\beta)^{2} & g(\beta)^{2} & f(\beta)^{2} & f(\beta)g(\beta) & f(\beta)g(\beta) \\
g(\beta)^{2} & f(\beta)g(\beta) & g(\beta)^{2} & g(\beta)^{2} & f(\beta)g(\beta) & g(\beta)^{2} & f(\beta)g(\beta) & f(\beta)^{2} & f(\beta)g(\beta) \\
g(\beta)^{2} & g(\beta)^{2} & f(\beta)g(\beta) & g(\beta)^{2} & g(\beta)^{2} & f(\beta)g(\beta) & f(\beta)g(\beta) & f(\beta)g(\beta) & f(\beta)^{2} \\
\end{pmatrix}, \nonumber
\end{align}
\end{minipage}
}

with $ f(\beta) = (e^{-i\beta} + 2)/3 $ and $ g(\beta) = (e^{-i\beta}-1)/3 $.

\quad The cost operator is:
\begin{align}
H_{C_{2}} &= \mathbbm{1} - \sigma_{3}^{1} \sigma_{4}^{0} \sigma_{5}^{0} + \sigma_{3}^{0} \sigma_{4}^{1} \sigma_{5}^{0} + \sigma_{3}^{0} \sigma_{4}^{0} \sigma_{5}^{1} \nonumber \\
U_{\text{cost}}(\alpha, \mathcal{R}) &= e^{i \alpha H_{C_{2}}} \nonumber \\
&= \mathbbm{1} + (e^{i \alpha} - 1) \left( \sigma_{3}^{1} \sigma_{4}^{0} \sigma_{5}^{0} + \sigma_{3}^{0} \sigma_{4}^{1} \sigma_{5}^{0} + \sigma_{3}^{0} \sigma_{4}^{0} \sigma_{5}^{1} \right),
\end{align}

Inside the feasible subspace, $H_{C_{2}} $ marks two states as solutions $\ket{100010}$ and $\ket{100100}$. The cost operator can be written in the relevant feasible subspace as a $9\times9$ matrix:
\begin{align}
\begin{pmatrix}
1   &   0   &   0   &   0   &   0   &   0   &   0   &   0   &   0   \\
0   &   1   &   0   &   0   &   0   &   0   &   0   &   0   &   0   \\
0   &   0   &   1   &   0   &   0   &   0   &   0   &   0   &   0   \\
0   &   0   &   0   &   1   &   0   &   0   &   0   &   0   &   0   \\
0   &   0   &   0   &   0   &   1   &   0   &   0   &   0   &   0   \\
0   &   0   &   0   &   0   &   0   &   1   &   0   &   0   &   0   \\
0   &   0   &   0   &   0   &   0   &   0   &   h(\alpha)   &   0   &   0   \\
0   &   0   &   0   &   0   &   0   &   0   &   0   &   h(\alpha)   &   0   \\
0   &   0   &   0   &   0   &   0   &   0   &   0   &   0   &   1
\end{pmatrix}, \nonumber
\end{align}
with $ h(\alpha) = e^{i\alpha} $. \cref{fig:three_clause_space} depicts the commutation space $\mathcal{F}_{C_{1},C_{3}} = \mathcal{F}_{0,0} \oplus \mathcal{F}_{0,1} \oplus \mathcal{F}_{1,0} \oplus \mathcal{F}_{1,1} $. By selecting an initial state inside $\mathcal{F}_{1,1}$, the evolution of the wavefunction remains constraint inside this subspace.

\begin{figure}
\includegraphics[width=0.6\textwidth]{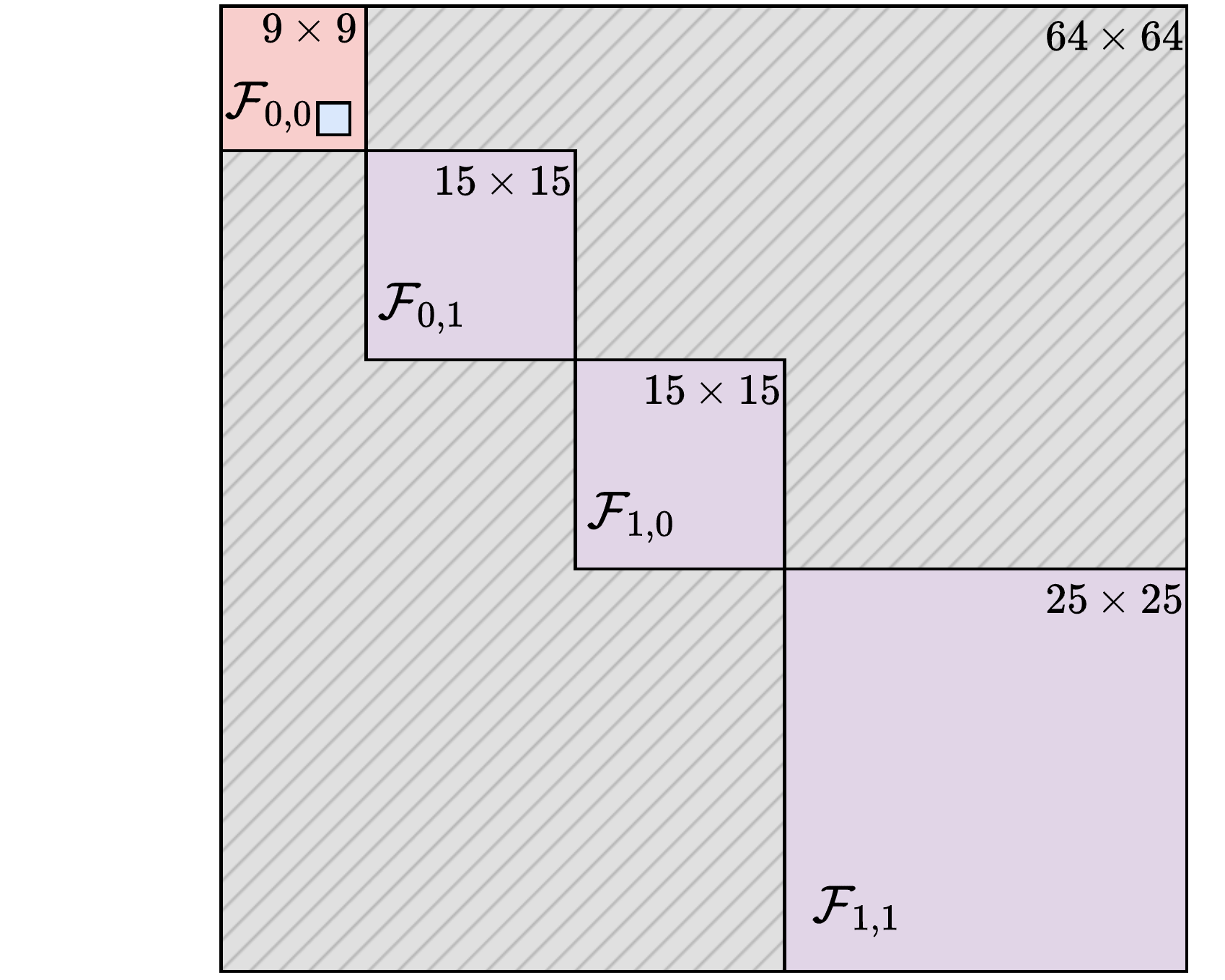}
\caption{\textbf{Commutation Space of Example MDS $\{ C_{1}, C_{3} \}$.} A depiction of the space $\mathcal{F}_{C_{1},C_{3}} = \mathcal{F}_{0,0} \oplus \mathcal{F}_{0,1} \oplus \mathcal{F}_{1,0} \oplus \mathcal{F}_{1,1} $ the example instance discussed in this section. In red is the relevant commutation space $\mathcal{F}_{0,0}$ and in blue is the solution subspace that is specified by $H_{C_{2}}$ from the remaining constraint $C_{2}$ that needs to be satisfied.}
\label{fig:three_clause_space}
\end{figure}

\subsection{Partial and Clause Neighborhood Commutation in the Constraint Graph}

Here we describe our third ansatz for QAOA on 1-in-3 SAT, which includes symmetry terms associated with partial constraints as well as between each constraint and the constraints that it shares variables with. Each 1-in-3 SAT clause has 3 (propositional) variables and satisfying conditions associated with each defines the constraints of the problem (\cref{subsec:rand_1in3}). The difficulty of satisfying a collection of 1-in-3 SAT clauses, as for all boolean satisfiability problems, arises from the conflicts between a clause and the other clauses that share variables with it since they may have different preferences on what assignments satisfy.

\quad Consider the reduced version of the example 1-in-3 SAT instance introduced in \cref{subsec:rand_1in3}, shown in \cref{fig:con_bipart}, and discussed in \cref{subsec:qaoa_xmixer} and \cref{subsec:qaoa_mdsmixer}. $ C_{1} $ and $ C_{2} $ share variable $ x_{4} $. Any valid assignment to the 1-in-3 SAT problem must satisfy $ C_{1} $ and $ C_{2} $, which means the subsets of valid assignments to $ C_{1} $ and $ C_{2} $ must be compatible, specifically with regard to $ x_{4} $. Then it is natural to consider the terms that commute with $ C_{1} $ and $ C_{2} $, since these terms mixes valid assignments of both $ C_{1} $ and $ C_{2} $. Similarly, a valid assignment to this example must also be a valid assignment for $ C_{2} $ and $ C_{3} $, with symmetry terms that mixes the valid assignments of $ C_{2} $ and $ C_{3} $. Of course, symmetry terms associated with all constraints could be particularly powerful catalyst mixers, but they are unlikely to be local and likely computationally expensive to precompute. By considering a single constraint and the symmetries associated with this particular constraint and those that have at least one variable in common with it, we define terms that mix more constricted assignments even if they do not maintain the global invariance desired. Such mixers, which we call Symmetric Cover (SymCov) mixers, could be applied to many diverse optimization problems characterized by having many local constraints that must be enforced together.

\quad Terms in the MDS could have been found using the flowchart in \cref{fig:flowchart}, but for such simple cases it is not necessary. In this case, we find that higher symmetries, such as the SymCov mixers, are accessible through the approach described in \cref{sec:theory_to_prac}.

\quad For a constraint $ C $, we define:
\begin{align}
S(C) = \{ C_{i} \; | \; C_{i} \text{ and } C \text{ share variables} \},
\end{align}
as the neighborhood of $ C $.

\quad To construct the Symmetric Cover mixer, the constraints that share variables for each constraint $ C $ in the MDS set $ \mathcal{D} $ are considered:
\begin{align}
\mathcal{S}(\mathcal{D}) = \{ S(C) \, | \, C \in \mathcal{D} \}.
\end{align}

\quad We follow the procedure delineated in \cref{fig:flowchart}, starting with finding all commutative terms $ \mathcal{T}(S(C)) $ associated with each $ S(C) $ in $ \mathcal{S}(\mathcal{D}) $. Consider a specific $ C \in \mathcal{D} $, since the other constraints in the MDS $ \mathcal{D} $ by definition have no shared variables, $ S(C) \subseteq \mathcal{R} \cup \{ C \} $ and $ S(C) $ is disjoint from $ \mathcal{D} - \{ C \} $.

\quad Suppose $ 3 $ constraints in $ \mathcal{R} $ are in the neighborhood $ S(C) $, then in principle there are more than $ 2^{55} $ ($ 12^{12} \times 2^{12} $) possibilities through a brute force method in the worst case, but the backtracking algorithm Alg.~\ref{alg:find_com_terms} from \cref{sec:theory_to_prac} in practice runs quickly, when implemented with minimal data structures in Julia.

\quad Then, by utilizing Alg.~\ref{alg:gen_com_unis} from \cref{sec:theory_to_prac} we find a collection of generators $ \mathcal{G}(S(C)) $ from $ \mathcal{T}(S(C)) $ and a corresponding unitary mixer $ U(\gamma, \mathcal{G}(S)) $ for each neighborhood $ S(C) $ in $ \mathcal{S}(\mathcal{D}) $. Then the symmetric cover mixer is:
\begin{align}
U_{\text{symcov}}(\gamma, \mathcal{S}) = \prod_{S \in \mathcal{S}} U(\gamma, \mathcal{G}(S)).
\end{align}

\quad Intuitively, the symmetric cover mixer $ U(\beta, \mathcal{S}) $ acts on the subspace of solutions to the neighborhood $ S(C) $ of each constraint $ C $ in $ \mathcal{D} $, which is strictly smaller than the space of solutions to $ \mathcal{D} $. While we considered the unweighted form here, the more general form would have individual optimized angles for each term.

\quad In principle, the subroutine Alg.~\ref{alg:check_gens} and therefore Alg.~\ref{alg:gen_com_unis} can be exponential time algorithms in the number of generators selected, but we find that in practice Alg.~\ref{alg:gen_com_unis} runs quickly for instances up to size $ 22 $.

\quad While the MDS leaves some variables uncovered, we can cover these \textit{partially} with some of the remaining clauses, although this only leads to a change if there is some constraint with both these terms. Consider a 1-in-3 constraint $ C = x_{1} + x_{2} + x_{3} = 1 $. Suppose $ x_{1} $ is covered in the MDS, then $ x_{2} + x_{3} \leq 1 $ while $ x_{1} \leq 1 $. Then $ (0,0), (1,0), (0,1) $ are possible configurations over $ \{ x_{2}, x_{3} \} $, but not $ (1,1) $. Suppose however, $ x_{2} $ is also covered in the MDS. Then the only partial constraint we have is $ x_{3} \leq 1 $, which is identical to no constraint, so we would use the X-mixer. While we will still need to apply penalty terms associated with these clauses, we can apply partial mixers rather than the X-mixer on any constraint with two variables outside the MDS. Let $ U_{R}(\beta, \mathcal{R}) $ refer to this approach and $ U_{mds}(\beta, \mathcal{C}) = U_{D}(\beta,\mathcal{D}) U_{R}(\beta, \mathcal{R}) $.

\quad Then one round of generalized QAOA is:
\begin{align}\label{eq:symcov_round}
\ket{\psi_{p}} = U_{\text{symcov}}(\gamma, \mathcal{S}(\mathcal{D})) \; U_{\text{mds}}(\beta, \mathcal{C}) \; U_{\text{cost}}(\alpha, \mathcal{R}) \ket{\psi_{p-1}}
\end{align}

\quad The generic algorithms presented in \cref{sec:theory_to_prac} can be used in many constructions, and their use in this section is only an example of their utility. While considered in the context of 1-in-3 SAT, similar constructions as the symmetric cover mixer discussed in this Section could be applied to most feasibility or optimization problems that have local constraints. Even if the constraints are not local, Alg.~\ref{alg:find_com_terms} can be utilized for finding local terms up to a desired weight.

\subsection{Benchmarks and Scaling Extrapolation}\label{subsec:bench1in3}

\begin{figure}[!t]
\centering
\includegraphics[width=1.0\textwidth]{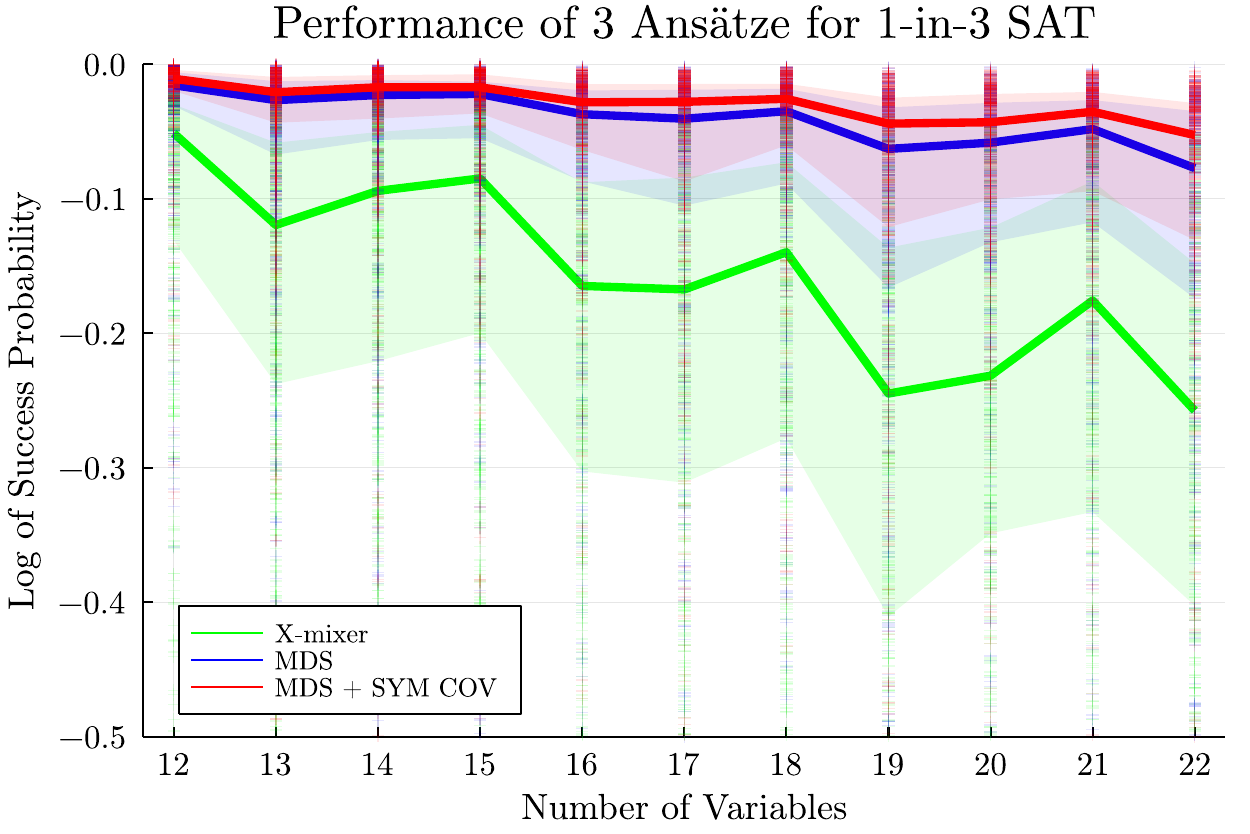}
\caption{\textbf{Ansatz Benchmarks for 1-in-3 SAT.} The success probability versus number of variables in the random instance of the 1-in-3 SAT problem with sizes between $ 12 $ and $ 22 $. The green QAOA with X-mixer (blue QAOA with MDS mixer, red QAOA with MDS+SYMCOV) dots show the success probability after $ 14 $ rounds of QAOA on each instance while the green (blue, red) line is the median surrounded by the upper and lower quartile ribbon on $ 500 $ random satisfiable instances for each size.}
\label{fig:succ_v_nbits}
\end{figure}

\begin{table}[!t]
    \centering
    \begin{tabular}{c|c|c}
                            &   cofactor    &   exponent base \\
        \hline
        X mixers            & 0.8973 &  1.0209 \\
        MDS mixers          & 0.9193 &  1.0107  \\
        MDS+SymCov mixers   & 0.9256 &  1.0092
    \end{tabular}
    \caption{\textbf{Exponential Curve Fits for each Ansatz.} We fit exponential curves $ f(n) = A B^{n} $, with cofactor $ A $ and exponent base $ B $ using a package to minimize the mean squared error for problem size $ n $ and empirical inverse success probability $1/p(n)$; \cref{fig:succ_v_nbits} shows $ \log(p(n))$. The low exponent for each approach suggests QAOA is well suited for solving 1-in-k SAT problems and that tailored ans{\"a}tze can improve QAOA's relative performance, consistent with an approximately quadratic improvement relative to the X-mixer baseline.}
    \label{tab:expcurvefits}
\end{table}

We benchmark QAOA with X-mixer, MDS mixer, and MDS mixer plus symmetric cover mixer with $ p=14 $ from size $ 12 $ through $ 22 $ over $ 500 $ random instances drawn with $ n/3 $ clauses and $ 3 $ variables per clause, with each clause generated by selecting three variables in the range $ [1, n] $ at random and the polarity of the literal (negative or positive) selected with even probability. Resulting instances with no solution are discarded and the process is repeated until $ 500 $ satisfiable instances are generated in total. We used vector simulation of the noiseless quantum system to compute exact expectation values, implemented in \texttt{Julia} with a sparse matrix library. All phase separating operators were diagonal in the computational basis, while the mixers had the diffusor form of \cref{eq:diffusor_mixer}, enabling efficient sparse matrix computation. On actual hardware, noise is expected to degrade QAOA performance~\cite{alam2019analysis, pellow2023qaoa} and suppress gradients for sufficiently deep circuits~\cite{schumann2024emergence}. Estimating gradients requires many measurement shots~\cite{schuld2019evaluating,wierichs2022general}.

\quad We choose $ p=14 $ as a representative intermediate-depth benchmarking regime rather than as a claim of optimality or universality. This choice is motivated by Ref.~\cite{boulebnane2024solving}, which reported that for random 8-SAT, QAOA with around $14$ layers matches the scaling of the strongest classical solver considered in that work, while larger depths can further improve performance. We expect increasing $p$ to improve the absolute success probabilities of all three ans{\"a}tze, although the precise gain is instance and optimizer dependent.

\quad With instances of size $ 12 $, we optimize the collections of angles $ \alpha, \beta $ of each approach in two rounds for \textit{p} set to $ 14 $ in a small angle regime, similar to the approach used in Ref.~\cite{boulebnane2024solving}. We sweep over different starting values for overall coefficient $ a $ and $ b $ (10 each) and two choices for the function associated each vector: constant $ \alpha = (a, \ldots, a, a), \beta = (b, \ldots, b) $ and linear $ \alpha = (a / p, \ldots, a(p-1)/p, a), \beta = (b, \ldots, 2 b / p, b/p) $. The linear-ramp schedule has been shown to be a good initialization strategy for QAOA across several domains~\cite{sack2021quantum,kremenetski2023quantum,montanez2025toward,he2024parameter}. Based on preliminary tuning, we sweep over $ \alpha \in [ 0.0, 0.2 ] $ and $ \beta \in  [ 0.0, 0.05 ] $ with 100 individual initializations considered in the grid search for each schedule family (constant or linear). For each sweep, we do $ 5,000 $ rounds of finite difference gradient descent, then take the best result and run $ 50,000 $ rounds of gradient descent to further optimize the angles. For the third approach, MDS+SymCov described by \cref{eq:symcov_round}, we then do $ 50,000 $ rounds of gradient descent over the collection of angles $ \alpha, \beta, \gamma $ starting with $ \gamma = \z $ (all zeros vector). Intuitively, $ U_{\text{symcov}}(\gamma, \mathcal{S}(\mathcal{D})) $ acts as a catalyst.

\quad In \cref{fig:succ_v_nbits}, results for each approach are shown with the median performance shown as a solid line and ribbons depicting the upper and lower quartiles. The results are consistent with better overall performance for the MDS mixer compared to the X-mixer and the strongest performance for the MDS mixer with the symmetric cover mixer.

\quad We are interested in the scaling for the time-to-solution (TTS) from repeatedly running the quantum circuit. Given a desired probability $P_{TS}$ to find an optimal solution and time $t_{run}$ to run the circuit for a single sample, the expected TTS for problems of size $n$ is:
\begin{align}
TTS(n) &= t_{run} \, \frac{\log\left( 1 - P_{TS} \right)}{\log\left(1-p(n) \right)} \nonumber \\
&\approx t_{run} \, \frac{\log\left( 1 - P_{TS} \right)}{p(n)} \; \; \text{ for } p(n) \ll 1 \nonumber \\
&= \mathcal{O}\left(\frac{1}{p(n)} \right) .
\end{align}

\quad We use \textit{CurveFit.jl} to optimize the mean-squared loss of fitting an exponential curve $ A \, B^{n} $ for values of $ A $ and $ B $ on the empirical $1/p(n)$; \cref{fig:succ_v_nbits} shows $ \log(p(n)) $. The fitted parameters are summarized in Tab.~\ref{tab:expcurvefits}. Relative to the X-mixer baseline, the tailored ans\"atze exhibit smaller fitted exponent bases, consistent with an approximately quadratic improvement in the fitted exponential scaling (e.g. $1.0092^{2.25} \approx 1.0209$). Our $ p=14 $ results show strong performance from all approaches and a soft decay for larger size $ n $, suggesting that QAOA may be well suited for solving 1-in-k SAT just as it has seen utility for k-SAT in Ref.~\cite{boulebnane2024solving}. However, results in this work suggest that if the Quantum Approximate Optimization Algorithm (QAOA with the X-mixer) is able to achieve a form of quantum advantage, a better suited ansatz could potentially enhance that advantage. Our results highlight ansatz-dependent effects in the idealized setting. Investigating how these effects persist under finite-shot estimation, on current quantum hardware, or under realistic hardware simulation is an important direction for future work.

\section{Conclusion and Outlook}\label{sec:conclusion}

\begin{table}[!t]
\begin{tabular}{l|l|l}
Question & Decision Problem & Complexity \\
\hline
Find a Phase-separator for cost function & $-$ & Poly(n) \\
Find a Mixer (Driver) for constraints & \texttt{CP-QCOMMUTE} & NP-Complete \\
Find a local Mixer (Driver) & \texttt{CP-QCOMMUTE-k-LOCAL} & Poly(n) \\
Can these mixers explore the entire feasible space? & \texttt{CP-QIRREDUCIBLECOMMUTE-GIVEN-k} & NP-Hard
\end{tabular}
\caption{\textbf{Important Questions for QAOA Designers and Their Complexity.} When designing QAOA for constrained optimization, we must find the appropriate mixers and phase-separators. The latter is straightforward (e.g. see Eq.~\ref{thm:eig_cor_for_con}). \cref{sec:complexity} provides the complexity of finding mixers, showing that only local mixers can be found tractably, and that they may still be insufficient to explore the full feasible subspace.}
\end{table}
\begin{table}[!t]
\begin{tabular}{l|l|l}
Question & Algorithm & Comment \\
\hline
Find local Mixers (Drivers) for Linear Constraints & Alg.~\ref{alg:find_lincom_terms} & Up to locality $ k $ \\
Find local Mixers (Drivers) for Polynomial Constraints & Alg.~\ref{alg:find_com_terms} & Up to locality $ k $ \\
Reduce number of mixers, but maintain reach & Alg.~\ref{alg:gen_com_unis} & Ensure $\mathcal{T}_{k} \subseteq \MG$ \\
Sequence mixers into commuting blocks & Alg.~\ref{alg:gen_com_unis} & Construct final mixer ansatz \\
\end{tabular}
\caption{\textbf{Important Questions for QAOA Designers and Our Prescriptions.} We summarize algorithmic procedures from \cref{sec:theory_to_prac} to generate a final mixer ansatz, (1) we find appropriate local mixers (for linear or polynomial constraints), (2) reduce the set to a collection of generators (of the associated Jordan algebra $M_{\mathcal{G}}$), and (3) sequence the mixers into commuting blocks to define the ansatz.}
\end{table}

In this manuscript, we discuss the algebraic and algorithmic foundations of imposing constraints on quantum operators, focusing on the utilization of such operators to aid new ansatz constructions and clarify the complexity of the associated tasks. Ansatz construction can aid in improving success probability~\cite{hodson2019portfolio,cook2020quantum,leipold2022tailored,golden2023numerical,boulebnane2023peptide}, reducing expressibility~\cite{sim2019expressibility, tangpanitanon2020expressibility,nakaji2021expressibility}, improving learnability~\cite{holmes2022connecting}, and enhancing noise resistance~\cite{streif2021quantum}. Our results support the conclusion that symmetry terms can act as catalysts in QAOA and can also significantly reduce the search space, thereby improving success probability. Our algebraic techniques to find commutators to improve QAOA performance were recently utilized in Ref.~\cite{xiang2025choco}.

\quad Finding commutative terms for embedded constraints is in general NP-Complete, but becomes polynomial for locally bounded constraints in the degree of the bound (\cref{sec:complexity}). \cref{sec:theory_to_prac} provides practical algorithmic procedures that can solve the problem in reasonable run-times despite the high upper bound. Along the way (see Fig.~\ref{fig:flowchart}), we clarify important points of consideration for practical ansatz construction, such as grouping commutative terms (\cref{sec:alg_driver_com} and \cref{subsec:redgen_sequencing}) and defining unitary operators (\cref{subsec:gen_to_proj}).

\m

\quad Such algorithms can be useful for constructing driver Hamiltonians or unitary mixers for a large array of optimization problems, while the reformulations, such as in App.~\ref{app:suffquad}, can be geared towards native quantum tasks.

\quad In Ref.~\cite{leipold2021constructing}, it was shown that any 2-local Hamiltonian is nonviable for certain classes of constraints. The algebraic condition found in this manuscript can also be utilized to find similar exclusionary statements for problems using more general constraints, thereby showing that certain constructions are not viable under resource limitations.

\quad This work presents many avenues for future study. While the problem \texttt{CP-QIRREDUCIBLE-COMMUTE} is NP-Hard, a more precise statement about its complexity within the polynomial hierarchy is an open problem, as well as the same problem for linear constraints. Coefficients $ \alpha $ on commutative terms $ T $ can be any complex number, thereby including nonstoquastic Hamiltonian drivers and their associated mixers. A satisfactory mathematical framework that allows us to select these coefficients or general learning techniques to find optimal patterns is of clear interest. Finding practical problems in quantum chemistry for which the expressions found in App.~\ref{app:suffquad} provide a useful secondary ansatz or catalyst driver term is another exciting avenue. These constructions could also be useful for error suppression.

\quad While Alg.~\ref{alg:find_com_terms} for finding commutative terms is sound and complete, alternative primitives may offer better runtime guarantees with good approximation and wider applicability. Alg.~\ref{alg:gen_com_unis} is a simple practical algorithm; alternative algorithmic primitives for grouping and sequencing generators into unitary mixers could improve practical circuit depth or perform better based on other criteria. Developing special purpose transpilers for mapping generators to low-body gates with limited connectivity can be a useful avenue for future consideration to apply these techniques to current NISQ hardware. Beyond this, future work can consider suitable constructions for still wider varieties of problems, constraints, and encodings, including integer problems and qudit-based mixers constructions suitable for qudits~\cite{sawaya2023encoding}.

\section*{Acknowledgements}

We are grateful for the support from the NASA Ames Research Center and from DARPA under IAA 8839, Annex 128. The research is based upon work (partially) supported by the Office of the Director of National Intelligence (ODNI), Intelligence Advanced Research Projects Activity (IARPA) and the Defense Advanced Research Projects Agency (DARPA), via the U.S. Army Research Office contract W911NF-17-C-0050. HL was also supported by the USRA Feynman Quantum Academy and funded by the NASA Academic Mission Services R\&D Student Program, Contract no. NNA16BD14C. SH is thankful for support from NASA Academic Mission Services, Contract No. NNA16BD14C. HL thanks Lutz Leipold for assistance in writing and debugging sparse matrix calculations in the programming language Julia, especially GPU implementations with CUDA. The authors thank anonymous reviewers for their insightful comments and suggestions.

\bibliographystyle{IEEEtran}
\bibliography{references}

\pagebreak

\appendix

\section{Sufficient Conditions for Quadratic Constraints}\label{app:suffquad}

Given any complete basis for a single qubit system, we can extend that basis to define a basis over $ n $ qubits. As discussed in previous sections $ \sigma^{+} = \ketbra{1}{0} $ and $ \sigma^{-} = \ketbra{0}{1} $. Thus, we can utilize the complete basis $ \{ \mathbbm{1}_2, \sigma^{z}, \sigma^{+}, \sigma^{-} \}^{\otimes n} $. Note that $ \left( \alpha_{j} \sigma_{i}^{\pm} \right)^{\dg} = \alpha_{j}^{\dg} \sigma_{i}^{\mp} $ for any $ \alpha_{j} \in \mathbb{C} $. Consequently, \textit{any} Hermitian matrix can be written in the form:
\begin{align}
	H &= \sum_{ ( \vec{y_{j}}, \vec{v_j}, \vec{w_j} ) \in \Delta( \mathscr{Y},  \; \mathscr{V}, \; \mathscr{W} ) } \alpha_{j} \prod_{i = 1}^{n} \left( \sigma_{i}^{z} \right)^{y_{ji}} \left( \sigma_{i}^{+} \right)^{v_{ji}} \left( \sigma_{i}^{-} \right)^{w_{ji}} \nonumber \\
	&\phantom{= } + \sum_{ ( \vec{y_{j}}, \vec{v_j}, \vec{w_j} ) \in \Delta( \mathscr{Y}, \; \mathscr{V}, \; \mathscr{W} ) } \alpha_{j}^{\dg} \prod_{i = 1}^{n} \left( \sigma_{i}^{z} \right)^{y_{ji}} \left( \sigma_{i}^{+} \right)^{w_{ji}} \left( \sigma_{i}^{-} \right)^{v_{ji}}, \label{eq:hexpand}
\end{align}

Recall the main algebraic result from Ref.~\cite{leipold2021constructing}:

\begin{thm}[Algebraic Condition for Commutativity] \label{thm:algconcom}
	A Hermitian matrix $ H $ commutes with an embedding of a linear constraint $ C = \sum_{j=1}^{n} h_{j} \, \sigma_{j}^{z} $ if and only if $ \vec{h} \cdot (\vec{v}_{j} - \vec{w}_{j}) = 0 $ for all $ \vec{v}_{j}, \vec{w}_{j} \in \Delta( \mathscr{V}, \mathscr{W} ) $.
\end{thm}

\quad In this paper, we wish to find similar algebraic conditions for more general constraints. For higher order term constraints, we can write the constraint in the general form:
\begin{align}
    \hat{C} = \sum_{ J_{K} \in \mathcal{J} } J_{K} \prod_{k \in K} \sigma_{k}^{z},
\end{align}

where $ \mathcal{J} $ is a collection of sets of polynomial size in the number of qubits $ n $, each with a coefficient $ J_{K} $ and a set of qubit indices $ K \subseteq [1,\ldots, n] $. Notice that $ \sigma^{+} \sigma^{z} = -\sigma^{z} \sigma^{+} = \sigma^{+} $ and $ \sigma^{-} \sigma^{z} = -\sigma^{z} \sigma^{-} = - \sigma^{-} $.

\m

\quad Using the form for $ H $ from \cref{eq:hexpand}, define $ T(\vec{y},\vec{v},\vec{w}) = \prod_{i=1}^{n} \left( \sigma_{i}^{z} \right)^{y_{i}} \left( \sigma_{i}^{+} \right)^{v_{i}} \left( \sigma_{i}^{-} \right)^{w_{i}} $ such that $ T^{\dg}(\vec{y},\vec{v},\vec{w}) = T(\vec{y},\vec{w},\vec{v}) $.

\m

\quad Next, consider the commutation of a general $ H $ with the constraint embedding operator for a general constraint $ C $:
\begin{align}
    \left[ H, \hat{C} \right] &= \sum_{ ( \vec{y_{j}}, \vec{v_j}, \vec{w_j} ) \in \Delta( \mathscr{Y}, \; \mathscr{V}, \; \mathscr{W} ) } \left[ \alpha_{j} \prod_{i = 1}^{n} \left( \sigma_{i}^{z} \right)^{y_{ji}} \left( \sigma_{i}^{+} \right)^{v_{ji}} \left( \sigma_{i}^{-} \right)^{w_{ji}}, \sum_{J_{K} \in \mathcal{J}} J_{K} \prod_{k \in K} \sigma_{k}^{z} \right] \nonumber \\
    &\phantom{= } + \sum_{ ( \vec{y_{j}}, \vec{v_j}, \vec{w_j} ) \in \Delta( \mathscr{Y}, \; \mathscr{V}, \; \mathscr{W} ) }  \left[ \alpha_{j}^{\dg} \prod_{i = 1}^{n} \left( \sigma_{i}^{z} \right)^{y_{ji}} \left( \sigma_{i}^{+} \right)^{w_{ji}} \left( \sigma_{i}^{-} \right)^{v_{ji}}, \sum_{J_{K} \in \mathcal{J}} J_{K} \, \prod_{k\in K} \sigma_{k}^{z} \right] \nonumber \\
    &= \sum_{ ( \vec{y_{j}}, \vec{v_j}, \vec{w_j} ) \in \Delta( \mathscr{Y}, \; \mathscr{V}, \; \mathscr{W} ) } \alpha_{j} \, \sum_{J_{K} \in \mathcal{J}} J_{K} \, L_{jK} \nonumber \\
    &\phantom{= } + \sum_{ ( \vec{y_{j}}, \vec{v_j}, \vec{w_j} ) \in \Delta( \mathscr{Y}, \; \mathscr{V}, \; \mathscr{W} ) } \alpha_{j}^{\dg} \, \sum_{J_{K} \in \mathcal{J}} J_{K} \, R_{jK}
\end{align}

where:
\begin{adjustwidth}{-3em}{-3em}
\begin{align}
    L_{jK} &=  \left( \prod_{k \notin K} \left( \sigma_{k}^{z} \right)^{y_{jk}} \left( \sigma_{k}^{+} \right)^{v_{jk}} \left( \sigma_{k}^{-} \right)^{w_{jk}} \right) \left[ \prod_{k \in K} \left( \sigma_{k}^{z} \right)^{y_{jk}} \left( \sigma_{k}^{+} \right)^{v_{jk}} \left( \sigma_{k}^{-} \right)^{w_{jk}} , \prod_{k\in K} \sigma_{k}^{z} \right]  \nonumber \\
    &= \left( \prod_{k \notin K} \left( \sigma_{k}^{z} \right)^{y_{jk}} \left( \sigma_{k}^{+} \right)^{v_{jk}} \left( \sigma_{k}^{-} \right)^{w_{jk}} \right) \, \left( \prod_{k \in K} \left( \sigma_{k}^{z} \right)^{y_{jk}} \left( \sigma_{k}^{+} \right)^{v_{jk}} \left( \sigma_{k}^{-} \right)^{w_{jk}} \sigma_{k}^{z} - \prod_{k \in K} \sigma_{k}^{z} \left( \sigma_{k}^{z} \right)^{y_{jk}} \left( \sigma_{k}^{+} \right)^{v_{jk}} \left( \sigma_{k}^{-} \right)^{w_{jk}} \right) \nonumber \\
    &= \left( \prod_{k \notin K} \left( \sigma_{k}^{z} \right)^{y_{jk}} \left( \sigma_{k}^{+} \right)^{v_{jk}} \left( \sigma_{k}^{-} \right)^{w_{jk}} \right) \, \prod_{k\in K} (-1)^{w_{jk}} \left( \sigma_{k}^{z} \right)^{ y_{jk} + \mu_{jk} } \left( \sigma_{k}^{+} \right)^{v_{jk}} \left( \sigma_{k}^{-} \right)^{w_{jk}} \nonumber \\
    &\phantom{= } -  \left( \prod_{k \notin K} \left( \sigma_{k}^{z} \right)^{y_{jk}} \left( \sigma_{k}^{+} \right)^{v_{jk}} \left( \sigma_{k}^{-} \right)^{w_{jk}} \right) \prod_{k\in K} (-1)^{v_{jk}} \left( \sigma_{k}^{z} \right)^{ y_{jk} + \mu_{jk} } \left( \sigma_{k}^{+} \right)^{v_{jk}} \left( \sigma_{k}^{-} \right)^{w_{jk}} \nonumber \\
    &= \left( \prod_{k \in K} (-1)^{w_{jk}} - \prod_{k \in K} (-1)^{v_{jk}} \right) \prod_{k \in K} \left( \sigma_{k}^{z} \right)^{ \mu_{jk} } \, T(\vec{y_{j}},\vec{v_{j}},\vec{w_{j}}) \nonumber \\
    &= \left( \prod_{k \in K} (-1)^{w_{jk}} - \prod_{k \in K} (-1)^{v_{jk}} \right) T(\vec{y_{j}}+\vec{\mu_{j}},\vec{v_{j}},\vec{w_{j}})
\end{align}
\end{adjustwidth}

with $ \mu_{jk} = (1 - v_{jk})(1 - w_{jk}) $ (identity or spin-z were placed on location $ k $ by the $j$-th basis term), and so $ \vec{\mu_{j}} = ( \mu_{j1}, \ldots, \mu_{jn} )$. Likewise, it can be shown:
\begin{align}
    R_{jK} = \left( \prod_{k \in K} (-1)^{v_{jk}} - \prod_{k \in K} (-1)^{w_{jk}} \right)  T(\vec{y_{j}} + \vec{\mu_{j}}, \vec{w_{j}}, \vec{v_{j}})
\end{align}

and so:
\begin{align}
    \left[ H, \hat{C} \right] &=\phantom{+} \sum_{ ( \vec{y_{j}}, \vec{v_j}, \vec{w_j} ) \in \Delta( \mathscr{Y}, \; \mathscr{V}, \; \mathscr{W} ) } \alpha_{j} \, \sum_{J_{K} \in \mathcal{J}} J_{K} \, \left( \prod_{k \in K} (-1)^{w_{jk}} - \prod_{k \in K} (-1)^{v_{jk}} \right) T(\vec{y_{j}}+\vec{\mu_{j}},\vec{v_{j}},\vec{w_{j}}) \nonumber \\
    &\phantom{= } + \sum_{ ( \vec{y_{j}}, \vec{v_j}, \vec{w_j} ) \in \Delta( \mathscr{Y}, \; \mathscr{V}, \; \mathscr{W} ) } \alpha_{j}^{\dg} \, \sum_{J_{K} \in \mathcal{J}} J_{K} \left( \prod_{k \in K} (-1)^{v_{jk}} - \prod_{k\in K}(-1)^{w_{jk}} \right) T(\vec{y_{j}}+ \vec{\mu_{j}}, \vec{w_{j}}, \vec{v_{j}}) \nonumber \\
    &=\phantom{+} \sum_{ ( \vec{y_{j}}, \vec{v_j}, \vec{w_j} ) \in \Delta( \mathscr{Y}, \; \mathscr{V}, \; \mathscr{W} ) } \alpha_{j} \, \sum_{J_{K} \in \mathcal{J}} J_{K} \, \left( (-1)^{\vec{w_{j}} \cdot \mathbf{1}_{K}} -  (-1)^{\vec{v_{j}} \cdot \mathbf{1}_{K}} \right) T(\vec{y_{j}}+\vec{\mu_{j}},\vec{v_{j}},\vec{w_{j}}) \nonumber \\
    &\phantom{= } + \sum_{ ( \vec{y_{j}}, \vec{v_j}, \vec{w_j} ) \in \Delta( \mathscr{Y}, \; \mathscr{V}, \; \mathscr{W} ) } \alpha_{j}^{\dg} \, \sum_{J_{K} \in \mathcal{J}} J_{K} \left( (-1)^{\vec{v_{j}} \cdot \mathbf{1}_{K}} - (-1)^{\vec{w_{j}} \cdot \mathbf{1}_{K}} \right) T(\vec{y_{j}}+ \vec{\mu_{j}}, \vec{w_{j}}, \vec{v_{j}}), \label{eq:comeq}
\end{align}
where $ \mathbf{1}_{K} $ is the indicator function over the indices $K$ for each term $J_{K} \prod_{k \in K} \sigma_{k}^{z} $ in the constraint:
\begin{align}
\left( \mathbf{1}_{K} \right)_{j} = \begin{cases}
1 \text{ if } j \in K \\
0 \text{ if } j \notin K
\end{cases}
\end{align}

Define $ v(\mathcal{J}, \mathcal{Y},\mathcal{V},\mathcal{W}) \in \mathbb{C}^{2^{3n}} $, which appears like \cref{eq:comeq} except $ T(\vec{y},\vec{v},\vec{w}) $ is replaced by $ \vec{y} \otimes \vec{v} \otimes \vec{w} $. Clearly \cref{eq:comeq} is zero if and only if $ v(\mathcal{J}, \mathcal{Y}, \mathcal{V}, \mathcal{W}) = 0  $. We can define $ u_{j} = \vec{v_{j}} - \vec{w_{j}} $ since $ \vec{v_{j}}, \vec{w_{j}} $ never match on an index.

\quad Given that $ H $ has $ \mathcal{O}\left( \text{poly}(n) \right) $ nonzero terms over the basis, there are at most $ \mathcal{O} \left( |\mathcal{J}| \, \text{poly}(n) \right) $ such terms to check, all of which must be zero.

\quad A class of constraints of particular interest are quadratic constraints, which can be used to describe a large class of optimization problems. Certain NP-Hard problems, for example, are more naturally described with quadratic constraints. In this section we give a sufficient condition for quadratic constraints that generalizes naturally from Theorem~\ref{thm:algconcom}.

\quad Let us check that \cref{eq:comeq} matches our result from Theorem~\ref{thm:algconcom} when considering linear constraints, i.e. $ \hat{C} = \sum_{k=1}^{n} h_{k} \sigma_{k}^{z} $:
\begin{align}
    \left[ H, \hat{C} \right] &=\phantom{+} \sum_{ ( \vec{y_{j}}, \vec{v_j}, \vec{w_j} ) \in \Delta( \mathscr{Y}, \; \mathscr{V}, \; \mathscr{W} ) } \alpha_{j} \, \sum_{J_{K} \in \mathcal{J}} J_{K} \, \left( \prod_{k \in K} (-1)^{w_{jk}} - \prod_{k \in K} (-1)^{v_{jk}} \right) T(\vec{y_{j}}+\vec{\mu_{j}},\vec{v_{j}},\vec{w_{j}}) \nonumber \\
    &\phantom{= } + \sum_{ ( \vec{y_{j}}, \vec{v_j}, \vec{w_j} ) \in \Delta( \mathscr{Y}, \; \mathscr{V}, \; \mathscr{W} ) } \alpha_{j}^{\dg} \, \sum_{J_{K} \in \mathcal{J}} J_{K} \, \left( \prod_{k \in K} (-1)^{v_{jk}} - \prod_{k \in K} (-1)^{w_{jk}} \right) T(\vec{y_{j}}+ \vec{\mu_{j}}, \vec{w_{j}}, \vec{v_{j}}) \nonumber \\
    &=\phantom{+} \sum_{ ( \vec{y_{j}}, \vec{v_j}, \vec{w_j} ) \in \Delta( \mathscr{Y}, \; \mathscr{V}, \; \mathscr{W} ) } \alpha_{j} \, \sum_{k=1}^{n} h_{k} \, \left( (-1)^{w_{jk}} - (-1)^{v_{jk}} \right) T(\vec{y_{j}}+\vec{\mu_{j}},\vec{v_{j}},\vec{w_{j}}) \nonumber \\
    &\phantom{= } + \sum_{ ( \vec{y_{j}}, \vec{v_j}, \vec{w_j} ) \in \Delta( \mathscr{Y}, \; \mathscr{V}, \; \mathscr{W} ) } \alpha_{j}^{\dg} \, \sum_{k=1}^{n} h_{k} \, \left( (-1)^{v_{jk}} - (-1)^{w_{jk}} \right) T(\vec{y_{j}}+ \vec{\mu_{j}}, \vec{w_{j}}, \vec{v_{j}}) \nonumber \\
    &=\phantom{+} 2 \, \sum_{ ( \vec{y_{j}}, \vec{v_j}, \vec{w_j} ) \in \Delta( \mathscr{Y}, \; \mathscr{V}, \; \mathscr{W} ) } \alpha_{j} \sum_{k=1}^{n} h_{k} \left( v_{jk} - w_{jk} \right) T(\vec{y_{j}}, \vec{w_{j}}, \vec{v_{j}}) \nonumber \\
    &\phantom{= } + 2 \, \sum_{ ( \vec{y_{j}}, \vec{v_j}, \vec{w_j} ) \in \Delta( \mathscr{Y}, \; \mathscr{V}, \; \mathscr{W} ) } \alpha_{j}^{\dg} \sum_{k=1}^{n} h_{k} \left( w_{jk} - v_{jk} \right) T(\vec{y_{j}}, \vec{w_{j}}, \vec{v_{j}}) \nonumber \\
    &=\phantom{+} 2 \, \sum_{ ( \vec{y_{j}}, \vec{v_j}, \vec{w_j} ) \in \Delta( \mathscr{Y}, \; \mathscr{V}, \; \mathscr{W} ) } \alpha_{j} \, \vec{h} \cdot \left( \vec{v_{j}} - \vec{w_{j}} \right) T(\vec{y_{j}}, \vec{w_{j}}, \vec{v_{j}}) \nonumber \\
    &\phantom{= } - 2 \, \sum_{ ( \vec{y_{j}}, \vec{v_j}, \vec{w_j} ) \in \Delta( \mathscr{Y}, \; \mathscr{V}, \; \mathscr{W} ) } \alpha_{j}^{\dg} \, \vec{h} \cdot \left( \vec{v_{j}} - \vec{w_{j}} \right) T(\vec{y_{j}}, \vec{w_{j}}, \vec{v_{j}})
\end{align}

\quad Now, let us consider the case with quadratic and linear terms in the constraint. For such a constraint, the constraint embedding operator has the form $ \hat{C} = \sum_{i} h_{i} \sigma_{i}^{z} + \sum_{ij} J_{ij} \sigma_{i}^{z} \sigma_{j}^{z} $ and so:
\begin{align}
    \left[ H, \hat{C} \right] &=\phantom{+} \sum_{ ( \vec{y_{j}}, \vec{v_j}, \vec{w_j} ) \in \Delta( \mathscr{Y}, \; \mathscr{V}, \; \mathscr{W} ) } \alpha_{j} \, \sum_{J_{K} \in \mathcal{J}} J_{K} \, \left( \prod_{k \in K} (-1)^{w_{jk}} - \prod_{k \in K} (-1)^{v_{jk}} \right) T(\vec{y_{j}}+\vec{\mu_{j}},\vec{v_{j}},\vec{w_{j}}) \nonumber \\
    &\phantom{= } + \sum_{ ( \vec{y_{j}}, \vec{v_j}, \vec{w_j} ) \in \Delta( \mathscr{Y}, \; \mathscr{V}, \; \mathscr{W} ) } \alpha_{j}^{\dg} \, \sum_{J_{K} \in \mathcal{J}} J_{K} \, \left( \prod_{k \in K} (-1)^{v_{jk}} - \prod_{k \in K} (-1)^{w_{jk}} \right) T(\vec{y_{j}}+ \vec{\mu_{j}}, \vec{w_{j}}, \vec{v_{j}}) \nonumber \\
    &=\phantom{+} 2 \, \sum_{ ( \vec{y_{j}}, \vec{v_j}, \vec{w_j} ) \in \Delta( \mathscr{Y}, \; \mathscr{V}, \; \mathscr{W} ) } \alpha_{j} \sum_{k=1}^{n} h_{k} \left( v_{jk} - w_{jk} \right) T(\vec{y_{j}}, \vec{w_{j}}, \vec{v_{j}}) \nonumber \\
    &\phantom{= } + 2 \, \sum_{ ( \vec{y_{j}}, \vec{v_j}, \vec{w_j} ) \in \Delta( \mathscr{Y}, \; \mathscr{V}, \; \mathscr{W} ) } \alpha_{j}^{\dg} \sum_{k=1}^{n} h_{k} \left( w_{jk} - v_{jk} \right) T(\vec{y_{j}}, \vec{w_{j}}, \vec{v_{j}}) \nonumber \\
    &\phantom{= } + \sum_{ ( \vec{y_{j}}, \vec{v_j}, \vec{w_j} ) \in \Delta( \mathscr{Y}, \; \mathscr{V}, \; \mathscr{W} ) } \alpha_{j} \sum_{kl} J_{kl} \left( (-1)^{w_{jk} + w_{jl}} - (-1)^{v_{jk} + v_{jl}} \right) \left( \sigma_{k}^{z} \right)^{\mu_{jk}} \left( \sigma_{l}^{z} \right)^{\mu_{jl}} T(\vec{y_{j}}, \vec{w_{j}}, \vec{v_{j}}) \nonumber \\
    &\phantom{= } + \sum_{ ( \vec{y_{j}}, \vec{v_j}, \vec{w_j} ) \in \Delta( \mathscr{Y}, \; \mathscr{V}, \; \mathscr{W} ) } \alpha_{j}^{\dg} \sum_{kl}^{n} J_{kl} \left( (-1)^{v_{jk} + v_{jl}} - (-1)^{w_{jk} + w_{jl}} \right) \left( \sigma_{k}^{z} \right)^{\mu_{jk}} \left( \sigma_{l}^{z} \right)^{\mu_{jl}} T(\vec{y_{j}}, \vec{w_{j}}, \vec{v_{j}})
\end{align}

\quad To get a sufficient condition, assume terms associated with $ \{ h_{i} \} $ and $ \{ J_{ij} \} $ are zero respectively. The first condition is the previous result $ \vec{h} \cdot \left( \vec{v_{j}} - \vec{w_{j}} \right) = 0 $. A sufficient condition on the second is
\begin{align}
\sum_{kl} J_{kl} \left( (-1)^{w_{jk} + w_{jl}} - (-1)^{v_{jk} + v_{jl}} \right) &= \sum_{kl} J_{kl} \left( (1 - 2 \, w_{jk})(1 - 2 \, w_{jl}) - (1 - 2 \, v_{jk})(1 - 2 \, v_{jl}) \right) \nonumber \\
&= (\mathbf{1} - 2 \, \vec{w}_{j})^{T} \cdot J^M  \cdot (\mathbf{1} - 2 \, \vec{w}_{j}) - (\mathbf{1} - 2 \, \vec{v}_{j})^{T} \cdot J^M  \cdot (\mathbf{1} - 2 \, \vec{v}_{j}) \nonumber \\
&= 0,
\end{align}

where $J^{M} \in \R^{n \times n}$ is the symmetric matrix of coefficients $(J^{M})_{jk} = J_{jk} $.

\begin{thm}[Sufficient Condition for Quadratic] \label{algconcom}
	A Hermitian Matrix $ H $ commutes with an embedding of a quadratic constraint $ C $ if $ \vec{h} \cdot \vec{v} = \vec{h} \cdot \vec{w} $ and $ (\mathbf{1} - 2 \, \vec{v})^{T} \cdot J^{M} \cdot (\mathbf{1} - 2 \, \vec{v}) = (\mathbf{1} - 2 \, \vec{w})^{T} \cdot J^{M} \cdot (\mathbf{1} - 2 \, \vec{w}) $ for all $ \vec{v}, \vec{w} $.
\end{thm}

\quad While this condition is sufficient, it is not necessary for $ H $ to commute with the constraint embedding operator. We consider a simple concrete counterexample. Consider constraints associated with the maximum independent set problem, where we wish to maximize a set of vertices $ S \subseteq V $ such that no $ v \in S $ has an edge to any other $ w \in S $. This can be represented with the quadratic constraints:
\begin{align}
    \forall v \in V, \; \sum_{(v,w) \in E} x_{v} \, x_{w} = 0.
\end{align}

\quad Clearly the Hamiltonian $ (\sigma_{v}^{+} + \sigma_{v}^{-}) \prod_{(v,w) \in E} \sigma_{w}^{0} = \left( \sigma_{v}^{+} + \sigma_{v}^{-} \right) \prod_{(v,w) \in E} \left( \mathbbm{1} + \sigma_{w}^{z} \right) $ commutes with the embedded constraint, but does not satisfy this sufficient condition (see Sec.~\ref{sec:algecond}). Note that the expressions found in Sec.~\ref{sec:algecond} are more useful for this case, since $ \sigma^{0} $ and $ \sigma^{1} $ are more naturally expressed in that formulation.

\section{Algebraic Condition for Commutation and Anticommutation of Linear Constraint Drivers} \label{sec:alg_driver_com}

\newcommand{\Ga}{T(\z,\z,\e_1, \e_2) + \text{h.c.}}
\newcommand{\Gb}{T(\z,\z,\e_2, \e_3) + \text{h.c.}}
\newcommand{\Gc}{T(\z,\z,\e_3, \e_4) + \text{h.c.}}
\newcommand{\Gd}{T(\z,\z,\e_{12}, \e_{34}) + \text{h.c.}}
\newcommand{\Ge}{T(\z,\z,\e_{34}, \e_{56}) + \text{h.c.}}
\newcommand{\Gab}{T(\z,\z,\e_{1},\e_{3}) + \text{h.c.}}

A specific class of interest is the commutation and anticommutation of driver terms that have linear constraints imposed on them. They have a natural reduced representation over \textit{only} $\{ \sigma^{+}, \sigma^{-} \}^{\otimes n} $. Pertinent to this discussion are the identities related to the left and right multiplication of $\sigma^{+}$ and $\sigma^{-}$, specifically that $ \sigma^{+} \sigma^{+} = \sigma^{-} \sigma^{-} = 0 $ while $ \sigma^{+} \sigma^{-} = \sigma^{1} $ and $ \sigma^{-} \sigma^{+} = \sigma^{0} $.

\subsection{Motivating Examples}\label{subsec:lincom_exams}

We begin with a few motivating examples. Recall that $\z$ is the null vector and one-hot vector $\e_{j} = (\underbrace{0, \ldots, 0}_{1:j-1}, 1,\underbrace{0,\ldots,0}_{j+1:n})$. Consider
\begin{align}
\Ga  &= \sigma_{1}^{+} \sigma_{2}^{-} + \sigma_{1}^{-} \sigma_{2}^{+} \\
\Gb  &= \sigma_{2}^{+} \sigma_{3}^{-} + \sigma_{2}^{-} \sigma_{3}^{+}
\end{align}

We have
\begin{align}
\left( \Ga \right) \, \left( \Gb \right) &= (\sigma_{1}^{+} \sigma_{2}^{-} + \sigma_{1}^{-} \sigma_{2}^{+}) (\sigma_{2}^{+} \sigma_{3}^{-} + \sigma_{2}^{-} \sigma_{3}^{+}) \nonumber  \\
&= \sigma_{1}^{+} \sigma_{2}^{0} \sigma_{3}^{-} + \sigma_{1}^{-} \sigma_{2}^{1} \sigma_{3}^{+} \\
\left( \Gb \right) \, \left( \Ga \right) &= (\sigma_{2}^{+} \sigma_{3}^{-} + \sigma_{2}^{-} \sigma_{3}^{+}) (\sigma_{1}^{+} \sigma_{2}^{-} + \sigma_{1}^{-} \sigma_{2}^{+}) \nonumber \\
&= \sigma_{1}^{+} \sigma_{2}^{1} \sigma_{3}^{-} + \sigma_{1}^{-} \sigma_{2}^{0} \sigma_{3}^{+} \\
\left[ \left( \Ga \right), \left( \Gb \right) \right] &= \sigma_{1}^{+} ( \sigma_{2}^{0} - \sigma_{2}^{1} ) \sigma_{3}^{-} + \sigma_{1}^{-} ( \sigma_{2}^{1} - \sigma_{2}^{0} ) \sigma_{3}^{+} \nonumber \\
&= \sigma_{1}^{+} \sigma_{2}^{z} \sigma_{3}^{-} - \sigma_{1}^{-} \sigma_{2}^{z} \sigma_{3}^{+} \\
\left\{ \left( \Ga \right), \left( \Gb \right) \right\} &= \sigma_{1}^{+} ( \sigma_{2}^{0} + \sigma_{2}^{1} ) \sigma_{3}^{-} + \sigma_{1}^{-} ( \sigma_{2}^{0} + \sigma_{2}^{1} ) \sigma_{3}^{+} \nonumber \\
&= \sigma_{1}^{+} \sigma_{3}^{-} + \sigma_{1}^{-} \sigma_{3}^{+} \nonumber \\
&= T(\z,\z,\e_1,\e_3) + \text{h.c.}
\end{align}

Consider
\begin{align}
\Gc &= \sigma_{3}^{+} \sigma_{4}^{-} + \sigma_{3}^{-} \sigma_{4}^{+}
\end{align}

We have
\begin{align}
\left( \Ga \right) \, \left( \Gc \right) &= \left( \Gc \right) \left( \Ga \right)  \nonumber \\
&= \left\{ \left( \Ga \right), \left( \Gc \right) \right\} / 2 \nonumber \\
&=  ( \sigma_{1}^{+} \sigma_{2}^{-} + \sigma_{1}^{-} \sigma_{2}^{+} ) ( \sigma_{3}^{+} \sigma_{4}^{-} + \sigma_{3}^{-} \sigma_{4}^{+} ) \nonumber \\
&= \sigma_{1}^{+} \sigma_{2}^{-} \sigma_{3}^{+} \sigma_{4}^{-} + \sigma_{1}^{+} \sigma_{2}^{-} \sigma_{3}^{-} \sigma_{4}^{+} + \sigma_{1}^{-} \sigma_{2}^{+} \sigma_{3}^{+} \sigma_{4}^{-} + \sigma_{1}^{-} \sigma_{2}^{+} \sigma_{3}^{-} \sigma_{4}^{+}  \nonumber \\
&= T(\z,\z,\e_{13},\e_{24}) + T(\z,\z,\e_{14},\e_{23}) + \text{h.c.} \\
[ \left( \Ga \right), \left( \Gc \right) ] &= 0
\end{align}

To see how this generalizes for higher order terms, we use the notation $\e_{jk} = \e_{j} + \e_{k} $. Consider the $4$-local terms
\begin{align}
\Gd &= \sigma_{1}^{+} \sigma_{2}^{+} \sigma_{3}^{-} \sigma_{4}^{-} + \sigma_{1}^{-} \sigma_{2}^{-} \sigma_{3}^{+} \sigma_{4}^{+} \\
\Ge &= \sigma_{3}^{+} \sigma_{4}^{+} \sigma_{5}^{-} \sigma_{6}^{-} + \sigma_{3}^{-} \sigma_{4}^{-} \sigma_{5}^{+} \sigma_{6}^{+}
\end{align}

This yields
\begin{align}
\left( \Gd \right) \, \left( \Ge \right) &= \left( \sigma_{1}^{+} \sigma_{2}^{+} \sigma_{3}^{-} \sigma_{4}^{-} + \sigma_{1}^{-} \sigma_{2}^{-} \sigma_{3}^{+} \sigma_{4}^{+} \right) \left( \sigma_{3}^{+} \sigma_{4}^{+} \sigma_{5}^{-} \sigma_{6}^{-} + \sigma_{3}^{-} \sigma_{4}^{-} \sigma_{5}^{+} \sigma_{6}^{+} \right) \nonumber \\
&= \sigma_{1}^{+} \sigma_{2}^{+} \sigma_{3}^{0} \sigma_{4}^{0} \sigma_{5}^{-} \sigma_{6}^{-} + \sigma_{1}^{-} \sigma_{2}^{-} \sigma_{3}^{1} \sigma_{4}^{1} \sigma_{5}^{+} \sigma_{6}^{+} \nonumber \\
&= T(\e_{34}, \z, \e_{12}, \e_{56}) + T(\z, \e_{34}, \e_{56}, \e_{12}) \\
\left( \Ge \right) \, \left( \Gd \right) &= \left( \sigma_{3}^{+} \sigma_{4}^{+} \sigma_{5}^{-} \sigma_{6}^{-} + \sigma_{3}^{-} \sigma_{4}^{-} \sigma_{5}^{+} \sigma_{6}^{+} \right) \left( \sigma_{1}^{+} \sigma_{2}^{+} \sigma_{3}^{-} \sigma_{4}^{-} + \sigma_{1}^{-} \sigma_{2}^{-} \sigma_{3}^{+} \sigma_{4}^{+} \right) \nonumber \\
&= \sigma_{1}^{+} \sigma_{2}^{+} \sigma_{3}^{1} \sigma_{4}^{1} \sigma_{5}^{-} \sigma_{6}^{-} + \sigma_{1}^{-} \sigma_{2}^{-} \sigma_{3}^{0} \sigma_{4}^{0} \sigma_{5}^{+} \sigma_{6}^{+}  \nonumber \\
&= T(\z, \e_{34}, \e_{12}, \e_{56}) + T(\e_{34}, \z, \e_{56}, \e_{12}) \\
\left\{ \left( \Gd \right), \left( \Ge \right) \right\} &= \sigma_{1}^{+} \sigma_{2}^{+} ( \sigma_{3}^{0} \sigma_{4}^{0} + \sigma_{3}^{1} \sigma_{4}^{1} ) \sigma_{5}^{-} \sigma_{6}^{-} + \sigma_{1}^{-} \sigma_{2}^{-} ( \sigma_{3}^{0} \sigma_{4}^{0} + \sigma_{3}^{1} \sigma_{4}^{1} ) \sigma_{5}^{+} \sigma_{6}^{+} \nonumber \\
&= T(\e_{34}, \z, \e_{12}, \e_{56}) + T(\z, \e_{34}, \e_{12}, \e_{56}) + \text{h.c.} \\
[ \left( \Gd \right), \left( \Ge \right) ] &= \sigma_{1}^{+} \sigma_{2}^{+} ( \sigma_{3}^{0} \sigma_{4}^{0} - \sigma_{3}^{1} \sigma_{4}^{1} ) \sigma_{5}^{-} \sigma_{6}^{-} - \sigma_{1}^{-} \sigma_{2}^{-} ( \sigma_{3}^{0} \sigma_{4}^{0} - \sigma_{3}^{1} \sigma_{4}^{1} ) \sigma_{5}^{+} \sigma_{6}^{+} \nonumber \\
&= T(\e_{34}, \z, \e_{12}, \e_{56}) - T(\z, \e_{34}, \e_{12}, \e_{56}) \nonumber \\
&\phantom{= } - T(\e_{34}, \z, \e_{56}, \e_{12}) + T(\z, \e_{34}, \e_{56}, \e_{12})
\end{align}

\pagebreak
\subsection{Formula for Multiplication, Commutation, and Anticommutation}

\quad Any simple linear constraint driver term has the form (up to coefficient $\alpha$ and its complex conjugate):
\begin{align}
    \hat{d_{1}} = T(\z, \z, \vec{v_{1}}, \vec{w_{1}}) + T( \z, \z, \vec{w_{1}}, \vec{v_{1}})
\end{align}

Then, we consider the multiplication of two linear driver terms:
\begin{align}
\hat{d_1} \hat{d_2} &= \phantom{+}
\left( T(\z, \z, \vec{v_1}, \vec{w_1}) + \text{h.c.} \right) \left( T(\z, \z, \vec{v_2}, \vec{w_2}) + \text{h.c.} \right) \nonumber \\
&=
\phantom{+}
\left( \prod_{j=1}^{n} \left( \sigma^{+} \right)^{v_{1j}} \left( \sigma^{-} \right)^{w_{1j}}  + \text{h.c.} \right)
\left( \prod_{j=1}^{n} \left( \sigma^{+} \right)^{v_{2j}} \left( \sigma^{-} \right)^{w_{2j}} + \text{h.c.} \right)  \nonumber \\
&= \phantom{+}
\prod_{j=1}^{n} \left( \sigma^{+} \right)^{v_{1j}+v_{2j}} \left( \sigma^{-} \right)^{w_{1j}+w_{2j}}
+ \prod_{j=1}^{n} \left( \sigma^{+} \right)^{v_{1j} + w_{2j}} \left( \sigma^{-} \right)^{w_{1j} + v_{2j}} \nonumber \\
&\phantom{= } +
\prod_{j=1}^{n} \left( \sigma^{+} \right)^{w_{1j}+v_{2j}} \left( \sigma^{-} \right)^{v_{1j}+w_{2j}}
+ \prod_{j=1}^{n} \left( \sigma^{+} \right)^{w_{1j}+w_{2j}} \left( \sigma^{-} \right)^{v_{1j}+v_{2j}} \nonumber \\
&= \phantom{+} (1 - \Theta(\vec{v_1} \cdot \vec{v_2}))(1 - \Theta(\vec{w_1} \cdot \vec{w_2})) T(\vec{w_1} \circ \vec{v_2}, \vec{v_1} \circ \vec{w_2},\vec{v_1}' +\vec{v_2}', \vec{w_{1}}' + \vec{w_{2}}')  \nonumber \\
&\phantom{= } + (1 - \Theta(\vec{v_1} \cdot \vec{w_2})) (1 - \Theta(\vec{w_1} \cdot \vec{v_2})) T(\vec{w_1} \circ \vec{w_2}, \vec{v_1} \circ \vec{v_2}, \vec{v_1}'' + \vec{w_2}'', \vec{w_{1}}'' + \vec{v_{2}}'') \nonumber \\
&\phantom{= } + (1 - \Theta(\vec{w_1} \cdot \vec{v_2})) (1 - \Theta(\vec{v_1} \cdot \vec{w_2})) T(\vec{v_1} \circ \vec{v_2}, \vec{w_1} \circ \vec{w_2}, \vec{w_1}'' + \vec{v_2}'', \vec{v_{1}}'' + \vec{w_{2}}'') \nonumber \\
&\phantom{= } + (1 - \Theta(\vec{w_1} \cdot \vec{w_2})) (1 - \Theta(\vec{v_1} \cdot \vec{v_2})) T(\vec{v_1} \circ \vec{w_2}, \vec{w_1} \circ \vec{v_2}, \vec{w_1}' +\vec{w_2}', \vec{v_{1}}' + \vec{v_{2}}') \nonumber \\
&= \begin{cases}
T(\vec{w_1} \circ \vec{v_2}, \vec{v_1} \circ \vec{w_2},\vec{v_1}' + \vec{v_2}', \vec{w_{1}}'+ \vec{w_{2}}')
+ T(\vec{v_1} \circ \vec{w_2}, \vec{w_1} \circ \vec{v_2}, \vec{w_1}' + \vec{w_2}', \vec{v_{1}}' + \vec{v_{2}}'), \\
\phantom{========} \text{if } \vec{v_1} \cdot \vec{v_2} + \vec{w_1} \cdot \vec{w_2} = 0  \\
0, \phantom{=======}\text{otherwise}
\end{cases} \\
&\phantom{=} + \begin{cases}
T(\vec{w_1} \circ \vec{w_2}, \vec{v_1} \circ \vec{v_2}, \vec{v_1}'' + \vec{w_2}'', \vec{w_{1}}'' + \vec{v_{2}}'') + T(\vec{v_1} \circ \vec{v_2}, \vec{w_1} \circ \vec{w_2}, \vec{w_1}'' + \vec{v_2}'', \vec{v_{1}}'' + \vec{w_{2}}''), \\
\phantom{========} \text{if } \vec{v_1} \cdot \vec{w_2} + \vec{w_1} \cdot \vec{v_2} = 0  \\
0, \phantom{=======}\text{otherwise}
\end{cases}
\end{align}

With the primed vectors denoting that we subtract the overlap with the opposing dissimilar vector and double primed vectors denoting that we subtract the overlap with the opposing similar vector, such that
\begin{align}
\vec{v_{1}}'  &= \vec{v_{1}} - \vec{v_{1}} \circ \vec{w_{2}} , &
\vec{w_{1}}'  &= \vec{w_{1}} - \vec{w_{1}} \circ \vec{v_{2}} , \nonumber \\
\vec{v_{2}}'  &= \vec{v_{2}} - \vec{v_{2}} \circ \vec{w_{1}} , &
\vec{w_{2}}'  &= \vec{w_{2}} - \vec{w_{2}} \circ \vec{v_{1}} , \nonumber \\
\vec{v_{1}}'' &= \vec{v_{1}} - \vec{v_{1}} \circ \vec{v_{2}} , &
\vec{w_{1}}'' &= \vec{w_{1}} - \vec{w_{1}} \circ \vec{w_{2}} , \nonumber \\
\vec{v_{2}}'' &= \vec{v_{2}} - \vec{v_{2}} \circ \vec{v_{1}} , &
\vec{w_{2}}'' &= \vec{w_{2}} - \vec{w_{2}} \circ \vec{w_{1}} .
\end{align}

\pagebreak

\quad Then commutation leads to the skew-Hermitian term:
\begin{align}
[ \hat{d_1},  \hat{d_2} ]
&= \hat{d_1} \hat{d_2} - \hat{d_2} \hat{d_1} \nonumber \\
&= \begin{cases}
\left( T(\vec{w_1} \circ \vec{v_2}, \vec{v_1} \circ \vec{w_2}, \vec{v_1}' + \vec{v_2}', \vec{w_{1}}' + \vec{w_{2}}')
+ T(\vec{v_1} \circ \vec{w_2}, \vec{w_1} \circ \vec{v_2}, \vec{w_1}' + \vec{w_2}', \vec{v_{1}}' + \vec{v_{2}}') \right)
- \text{h.c.}, \\
\phantom{========}\text{if } \vec{v_1} \cdot \vec{v_2} + \vec{w_1} \cdot \vec{w_2} = 0, \\
0, \phantom{=======} \text{otherwise}
\end{cases} \\
&\phantom{=} +
\begin{cases}
\left(
T(\vec{w_1} \circ \vec{w_2}, \vec{v_1} \circ \vec{v_2}, \vec{v_1}'' + \vec{w_2}'', \vec{w_{1}}'' + \vec{v_{2}}'')
+ T(\vec{v_1} \circ \vec{v_2}, \vec{w_1} \circ \vec{w_2}, \vec{w_1}'' + \vec{v_2}'', \vec{v_{1}}'' + \vec{w_{2}}'')
\right) - \text{h.c.}, \\
\phantom{========}\text{if } \vec{v_1} \cdot \vec{w_2} + \vec{w_1} \cdot \vec{v_2} = 0, \\
0, \phantom{=======} \text{otherwise.}
\end{cases}
\end{align}

While anticommutation leads to the Hermitian term:
\begin{align}
\left\{ \hat{d_{1}}, \hat{d_{2}} \right\}
&= \hat{d_1} \hat{d_2} + \hat{d_2} \hat{d_1} \nonumber \\
&= \begin{cases}
\left(
T(\vec{w_1} \circ \vec{v_2}, \vec{v_1} \circ \vec{w_2}, \vec{v_1}' + \vec{v_2}', \vec{w_{1}}' + \vec{w_{2}}')
+ T(\vec{v_1} \circ \vec{w_2}, \vec{w_1} \circ \vec{v_2}, \vec{w_1}' + \vec{w_2}', \vec{v_{1}}' + \vec{v_{2}}')
\right)
+ \text{h.c.}, \\
\phantom{========}\text{if } \vec{v_1} \cdot \vec{v_2} + \vec{w_1} \cdot \vec{w_2} = 0, \\
0, \phantom{=======}\text{otherwise}
\end{cases} \nonumber \\
&\phantom{=} +
\begin{cases}
\left(
T(\vec{w_1} \circ \vec{w_2}, \vec{v_1} \circ \vec{v_2}, \vec{v_1}'' + \vec{w_2}'', \vec{w_{1}}'' + \vec{v_{2}}'')
+ T(\vec{v_1} \circ \vec{v_2}, \vec{w_1} \circ \vec{w_2}, \vec{w_1}'' + \vec{v_2}'', \vec{v_{1}}'' + \vec{w_{2}}'')
\right) + \text{h.c.}, \\
\phantom{========}\text{if } \vec{v_1} \cdot \vec{w_2} + \vec{w_1} \cdot \vec{v_2} = 0, \\
0, \phantom{=======}\text{otherwise.}
\end{cases} \label{eq:lin_anticom}
\end{align}

The two terms generated by this anticommutation formula are utilized in the generator reduction algorithm discussed in Sec.~\ref{sec:check_gens}.

\subsection{Reduced Linear Terms}\label{subsec:redlinterms}

Linear Constraint terms do not require $\sigma^{0}$ or $\sigma^{1}$ (see \cref{thm:linconcom}) and so primitives such as Alg.~\ref{alg:find_lincom_terms} do not unnecessarily use them. However, as shown in this section, anticommutation leads to terms that do have such basis terms for some qubits, specifically when there is some overlap between the anticommutants. As an example from this section, recall:
\begin{align}
\Gd &= \sigma_{1}^{+} \sigma_{2}^{+} \sigma_{3}^{-} \sigma_{4}^{-} + \sigma_{1}^{-} \sigma_{2}^{-} \sigma_{3}^{+} \sigma_{4}^{+} \\
\Ge &= \sigma_{3}^{+} \sigma_{4}^{+} \sigma_{5}^{-} \sigma_{6}^{-} + \sigma_{3}^{-} \sigma_{4}^{-} \sigma_{5}^{+} \sigma_{6}^{+} \\
\left\{ \left( \Gd \right), \left( \Ge \right) \right\} &= T(\e_{34}, \z, \e_{12}, \e_{56}) + T(\z, \e_{34}, \e_{12}, \e_{56}) + \text{h.c.} \\
&= \sigma_{1}^{+} \sigma_{2}^{+} ( \sigma_{3}^{0} \sigma_{4}^{0} + \sigma_{3}^{1} \sigma_{4}^{1} ) \sigma_{5}^{-} \sigma_{6}^{-} + \sigma_{1}^{-} \sigma_{2}^{-} ( \sigma_{3}^{0} \sigma_{4}^{0} + \sigma_{3}^{1} \sigma_{4}^{1} ) \sigma_{5}^{+} \sigma_{6}^{+}
\end{align}

This is different from the earlier $2$-local examples from \cref{subsec:lincom_exams} which end up not needing such terms (because $\sigma^{0} + \sigma^{1} = \mathbbm{1}_{2}$). This has an important implication for generator reduction algorithms like Alg.~\ref{alg:check_gens}, since we would like to reduce generators that are not necessary. For example, $ \Ga $ and $ \Gb $ can generate $T(\z,\z,\e_{1},\e_{3}) + \text{h.c.}$ and so we can reduce this term. However, when dealing with many terms of higher locality weight than $2$, we would in general expect that terms with overlap simply cannot help reduce since all the generators have no vectors in the first two slots.

If we wish to use this for reduction of linear terms, we have to make a \textit{reduction assumption}. Specifically, given the two terms
\begin{align}
\hat{d_{1}} &= T(\z, \z, \vec{v_{1}}, \vec{w_{1}}) + T( \z, \z, \vec{w_{1}}, \vec{v_{1}}), \\
\hat{d_{2}} &= T(\z, \z, \vec{v_{2}}, \vec{w_{2}}) + T( \z, \z, \vec{w_{2}}, \vec{v_{2}}),
\end{align}
the anticommutator term is given by \cref{eq:lin_anticom}, which leads to the \textit{reduced terms} (since $\sigma^{0}$ and $\sigma^{1}$ are in the kernel of the commutator with the constraints):
\begin{align}
T( \z, \z, \vec{v_1}' + \vec{v_2}', \vec{w_{1}}' + \vec{w_{2}}') &+ \text{h.c.}  \nonumber \\
T( \z, \z, \vec{v_1}'' + \vec{w_2}'', \vec{w_{1}}'' + \vec{v_{2}}'') &+ \text{h.c.} \label{eq:redcomterm}
\end{align}

As such, we can add this assumption to algorithms like Alg.~\ref{alg:check_gens} to heuristically prune further generators at the potential cost of reducing the reachability of the generator set.

\section{A Sound Algorithm for Generator Reduction}\label{sec:check_gens}

\begin{figure}[!t]
\centering
\setlength{\intextsep}{0.5em}
\begin{minipage}{1.0\linewidth}
\begin{algorithm}[H]
\caption{\algcheckgen$(h, i, u, G, n)$}\label{alg:check_gens}
\begin{algorithmic}[1]
\Statex \textbf{Inputs: } current term $ h = T(\vec{x}_h, \vec{y}_h, \vec{v}_h, \vec{w}_h) $, current index $ i $, target term $ u = T(\z, \z, \vec{v}_u, \vec{w}_u) $, generator list $ G $, number of variables $ n $
\If{$\vec{x}_{h} \neq \z$ or $\vec{y}_{h} \neq \z $} \Comment{Since $\vec{x}_u = \vec{y}_u = \z $, end branch}
    \State \textbf{return False.}
\EndIf
\If{$ i > n $}
    \State \textbf{return True.} \Comment{$ h $ matched $ u $ on all indices.}
\EndIf
\If{($ u $ has a term on index $ i $) \textbf{and} ($ h $ has a term that agrees with $ u $ on index $ i $)}
    \State \textbf{return} \algcheckgen$(h, i+1, u, G, n)$. \Comment{$h$ may become $u$, move onto index $i+1$.}
\EndIf
\If{($ u $ has a term on index $ i $) \textbf{and} ($ h $ has a term that disagrees with $ u $ on index $ i $)}
    \State \textbf{return False.}
\EndIf
\If{($ u $ does not have a term on index $ i $) \textbf{and} ($ h $ has a term on index $ i $)}
    \For{$g \in G$ such that the first non-zero entry of $g$ is on index $i$}
        \State update $h$ to the matching candidate in expansion of $ \{ h + h^{\dg} , g + g^{\dg} \} $ \Comment{Using \cref{eq:lin_anticom}}
        \If{$h$ has no term on $i$ and \algcheckgen$(h,i+1,u,G,n)$}
            \State \textbf{return True.}
        \EndIf
        \State revert $h$
    \EndFor
\EndIf

\If{($ u $ has a term on index $ i $) \textbf{and} ($ h $ does not have a term on index $ i $)}
\For{$ g \in G $ such that the first non-zero entry of $ g $ is on index $ i $}
    \State update $ h $ to the matching candidate in expansion of $ \{ h + h^{\dg} , g + g^{\dg} \} $ \Comment{Using \cref{eq:lin_anticom}}
    \If{$h$ has a term that agrees with $u$ on index $i$}
    \If{\algcheckgen$(h,i+1,u,G,n)$}
        \State \textbf{return True.} \Comment{Covered all indices of term $u$.}
    \EndIf
    \EndIf
    \State revert $ h $
\EndFor
\EndIf
\If{($ u $ has no term on index $ i $) \textbf{and} ($ h $ has no term on index $ i $)}
    \If{\algcheckgen$(h,i+1,u,G,n)$}
        \State \textbf{return True.}
    \EndIf
\EndIf
\State \textbf{return False.}
\end{algorithmic}
\end{algorithm}
\end{minipage}
\vspace{-1em}
\caption*{{Alg. \ref{alg:check_gens}}: \textbf{Reducing Generators.} This algorithm checks if a given term, represented by a vector $ u $, can be generated by a collection of terms $ G $ (also represented by vectors) by attempting to generate candidate terms $ h $ that could match $ u $ through backtracking. Starting from $i=1$, the algorithm recursively compares the current term $h$ with the target term $u$ and updates $h$ using generators $g \in G$ so as to cover or remove support at index $i$. At each recursive step, when the expansion of an anticommutator produces multiple candidate terms, the algorithm follows the candidate that agrees with the already fixed indices $j < i$ and has the required behavior at the current index $i$.}
\setlength{\intextsep}{\ointextsep}
\end{figure}

Alg.~\ref{alg:check_gens} is a sound but incomplete backtracking algorithm for generator reduction, although it can be modified based on \cref{subsec:redlinterms} to be more aggressive and therefore approximate. We describe the algorithm in words. The algorithm explores possible anticommutator generated terms through backtracking, such that for any recursive branch the current expression is $h$. The algorithm continues to update $h$ in an attempt to match $u$ until either it manages to do so or recognizes it can no longer do so. It does this by considering each qubit index. Let
\begin{align}
u &= T(\z, \z, \vec{v}_u, \vec{w}_u) , \\
h &= T(\z, \z, \vec{v}_h, \vec{w}_h) , \\
g &= T(\z, \z, \vec{v}_g, \vec{w}_g) , \\
h' &= T(\vec{x}_{h'}, \vec{y}_{h'}, \vec{v}_{h'}, \vec{w}_{h'}) .
\end{align}

We consider a current qubit index $i$, we have $4$ possible scenarios for $u$ and $h$:
\begin{enumerate}[label=\arabic*.]
    \item $u$ and $h$ both have no term on $i$. In this case, $u$ and $h$ agree and we move onto the next index $i+1$ as a continuing recursive branch, returning the return value of that branch.
    \item $u$ has no term, but $h$ does. Without loss of generality, let $\vec{v}_h[i] = 1$. Then we consider $g$ such that $\vec{v}_{g}[i] = 1$ so that we generate a term $h'$ such that $\vec{v}_{h'}[i] = 0$. Note that \cref{eq:lin_anticom} for $ \left\{ h + h^{\dg}, g + g^{\dg} \right\} $ generates up to two candidate terms, and at most one of them matches this condition. For example, $\Ga $ and $ \Gb $ generate $ \Gab $ Note that if $h'$ has a nonzero vectors $\vec{x}_{h'}$ or $\vec{y}_{h'}$, it will be pruned in the next recursion.
    \item $u$ has a term, but $h$ does not. In this case, we look at $g \in G$ such that it could cover $i$ and doesn't cover any indices less than $i$ (we would have considered such a $g$ earlier). At this point $h = T(\z, \z, \vec{v}_h, \vec{w}_h)$ with $\vec{v}_h[i] = \vec{w}_h[i] = 0$. Let $u = T(\z,\z,\vec{v}_u, \vec{w}_u) $. Without loss of generality, suppose $\vec{v}_{u}[i] = 1 $ and so $\vec{w}_{u}[i] = 0$. Then $g = T(\z,\z,\vec{v}_g, \vec{w}_g)$ must be such that $\vec{v}_{g}[i] = 1$ to cover $i$ congruently to $u$. In particular, the new term is $ h' = T(\vec{x}_{h'}, \vec{y}_{h'}, \vec{v}_{h'}, \vec{w}_{h'}) $ with $\vec{v}_{h'}[i] = 1$.
    \item $u$ has a term and $h$ does as well. In this case, either they match and we continue by returning the return value of the branch with index $i+1$. If they do not match, we return False.
\end{enumerate}

First, we show that Alg.~\ref{alg:check_gens} terminates. The first call would be \algcheckgen$(T(\z, \z, \z, \z),1,u,G,n)$ and would return either True or False. Note that steps 1-3 prune terms with $\vec{x}_{h} \neq \z $ or $ \vec{y}_{h} \neq \z $, but an approximate version of Alg.~\ref{alg:check_gens} will simply set them to $\z$.

\begin{thm}[Termination of Alg.~\ref{alg:check_gens}]
Alg.~\ref{alg:check_gens} eventually terminates.
\end{thm}

\begin{proof}[Proof Sketch]
In every recursive call, we increment the index from $i$ to $i+1$ and so the algorithm branch will terminate when $i=n+1$ if a recursive branch reaches that far and earlier if not.
\end{proof}

\begin{thm}[Soundness of Alg.~\ref{alg:check_gens}]
Alg.~\ref{alg:check_gens} returns True only if $u$ can be generated by $G$.
\end{thm}

\begin{proof}[Proof Sketch]
If the algorithm returns True, it must be that a branch reached return True. Then at this point $i=n+1$ and for each index less than $n+1$, the term $h$ agreed with $u$. Then $u=h$ and $h$ was generated through $G$.
\end{proof}

Lastly, we note that Alg.~\ref{alg:check_gens} is not complete in full generality.

\begin{remark}[Incompleteness of Alg.~\ref{alg:check_gens}]
Alg.~\ref{alg:check_gens} is sound but incomplete in general. The source of incompleteness is that the algorithm restricts attention to candidate terms that remain linear, and prunes branches once nonzero first-slot or second-slot vectors appear. While this keeps the recursion within the scope of \cref{eq:lin_anticom}, it may discard branches that could contribute to a valid generator construction in the full anticommutator algebra. A complete version of the algorithm would require the more general update rule for anticommutators of terms of the form $T(\vec{x}, \vec{y}, \vec{v}, \vec{w}) + \mathrm{h.c.}$ with linear generators.
\end{remark}

\quad Note again that simply not reducing the generators and sequencing all $\mathcal{T}_{k}$ is an acceptable approach that does not harm separability or reachability.

\section{Multiplication Tables for Single Qubit Operators}\label{app:single_qubit_relations}

\subsection{For the Operators $\{ \mathbbm{1}_{2}, \sigma^{0}, \sigma^{1}, \sigma^{-}, \sigma^{+} \}$}

Relevant to results derived in Sec.~\ref{sec:algecond} and Sec.~\ref{sec:alg_driver_com} are multiplication identities for the single qubit operators
\begin{align}
\sigma^{0} &= \ketbra{0}{0} = \left( \mathbbm{1} + \sigma^{z} \right) / 2,      \nonumber\\
\sigma^{1} &= \ketbra{1}{1} = \left( \mathbbm{1} - \sigma^{z} \right) / 2,      \nonumber\\
\sigma^{-} &= \ketbra{0}{1} = \left( \sigma^{x} + i \, \sigma^{y} \right) / 2,  \nonumber\\
\sigma^{+} &= \ketbra{1}{0} = \left( \sigma^{x} - i \, \sigma^{y} \right) / 2.
\end{align}

Here we list the left-multiplication and right-multiplication identities for each matrix. Only one of $x,y,v,w$ is $1$ while the others are $0$.

\begin{align}
    \sigma^{0} &\text{ left}                \nonumber \\
    \sigma^{0} \sigma^{0} &= \sigma^{0}     \nonumber \\
    \sigma^{0} \sigma^{1} &= 0              \nonumber \\
    \sigma^{0} \sigma^{-} &= \sigma^{-}     \nonumber \\
    \sigma^{0} \sigma^{+} &= 0              \nonumber \\
    \sigma^{0} \left( \sigma^{0} \right)^{x} \left( \sigma^{1} \right)^{y} \left( \sigma^{+} \right)^{v} \left( \sigma^{-} \right)^{w}
    &= \left( 1 - \Theta(y+v) \right) \left( \sigma^{0} \right)^{x} \left( \sigma^{1} \right)^{y} \left( \sigma^{+} \right)^{v} \left( \sigma^{-} \right)^{w} \nonumber \\
    \sigma^{0} &\text{ right}               \nonumber \\
    \sigma^{0} \sigma^{0} &= \sigma^{0}     \nonumber \\
    \sigma^{1} \sigma^{0} &= 0              \nonumber \\
    \sigma^{-} \sigma^{0} &= 0              \nonumber \\
    \sigma^{+} \sigma^{0} &= \sigma^{+}     \nonumber \\
    \left( \sigma^{0} \right)^{x} \left( \sigma^{1} \right)^{y} \left( \sigma^{+} \right)^{v} \left( \sigma^{-} \right)^{w} \sigma^{0} &= \left( 1 - \Theta(y+w) \right) \left( \sigma^{0} \right)^{x} \left( \sigma^{1} \right)^{y} \left( \sigma^{+} \right)^{v} \left( \sigma^{-} \right)^{w}
\end{align}

\begin{align}
    \sigma^{1} &\text{ left}                \nonumber \\
    \sigma^{1} \sigma^{0} &= 0              \nonumber \\
    \sigma^{1} \sigma^{1} &= \sigma^{1}     \nonumber \\
    \sigma^{1} \sigma^{-} &= 0              \nonumber \\
    \sigma^{1} \sigma^{+} &= \sigma^{+}     \nonumber \\
    \sigma^{1} \left( \sigma^{0} \right)^{x} \left( \sigma^{1} \right)^{y} \left( \sigma^{+} \right)^{v} \left( \sigma^{-} \right)^{w}
    &= \left( 1 - \Theta(x+w) \right) \left( \sigma^{0} \right)^{x} \left( \sigma^{1} \right)^{y} \left( \sigma^{+} \right)^{v} \left( \sigma^{-} \right)^{w} \nonumber \\
    \sigma^{1} &\text{ right}               \nonumber \\
    \sigma^{0} \sigma^{1} &= 0              \nonumber \\
    \sigma^{1} \sigma^{1} &= \sigma^{1}     \nonumber \\
    \sigma^{-} \sigma^{1} &= \sigma^{-}     \nonumber \\
    \sigma^{+} \sigma^{1} &= 0              \nonumber \\
    \left( \sigma^{0} \right)^{x} \left( \sigma^{1} \right)^{y} \left( \sigma^{+} \right)^{v} \left( \sigma^{-} \right)^{w} \sigma^{1}
    &= \left( 1 - \Theta(x+v) \right) \left( \sigma^{0} \right)^{x} \left( \sigma^{1} \right)^{y} \left( \sigma^{+} \right)^{v} \left( \sigma^{-} \right)^{w}
\end{align}

\begin{align}
    \sigma^{+} &\text{ left}                \nonumber \\
    \sigma^{+} \sigma^{0} &= \sigma^{+}     \nonumber \\
    \sigma^{+} \sigma^{1} &= 0              \nonumber \\
    \sigma^{+} \sigma^{-} &= \sigma^{1}     \nonumber \\
    \sigma^{+} \sigma^{+} &= 0              \nonumber \\
    \sigma^{+} \left( \sigma^{0} \right)^{x} \left( \sigma^{1} \right)^{y} \left( \sigma^{+} \right)^{v} \left( \sigma^{-} \right)^{w}
    &= \left( 1 - \Theta(y+v) \right) \left( \sigma^{1} \right)^{y+w} \left( \sigma^{+} \right)^{v+x}                             \nonumber \\
    \sigma^{+} &\text{ right}               \nonumber \\
    \sigma^{0} \sigma^{+} &= 0              \nonumber \\
    \sigma^{1} \sigma^{+} &= \sigma^{+}     \nonumber \\
    \sigma^{-} \sigma^{+} &= \sigma^{0}     \nonumber \\
    \sigma^{+} \sigma^{+} &= 0              \nonumber \\
    \left( \sigma^{0} \right)^{x} \left( \sigma^{1} \right)^{y} \left( \sigma^{+} \right)^{v} \left( \sigma^{-} \right)^{w} \sigma^{+}
    &= \left( 1 - \Theta(x+v) \right) \left( \sigma^{0} \right)^{x+w} \left( \sigma^{+} \right)^{v+y}
\end{align}

\begin{align}
    \sigma^{-} &\text{ left}                \nonumber \\
    \sigma^{-} \sigma^{0} &= 0              \nonumber \\
    \sigma^{-} \sigma^{1} &= \sigma^{-}     \nonumber \\
    \sigma^{-} \sigma^{-} &= 0              \nonumber \\
    \sigma^{-} \sigma^{+} &= \sigma^{0}     \nonumber \\
    \sigma^{-} \left( \sigma^{0} \right)^{x} \left( \sigma^{1} \right)^{y} \left( \sigma^{+} \right)^{v} \left( \sigma^{-} \right)^{w}
    &= \left( 1 - \Theta(x+w) \right) \left( \sigma^{0} \right)^{x+v} \left( \sigma^{-} \right)^{y+w} \nonumber \\
    \sigma^{-} &\text{ right}               \nonumber \\
    \sigma^{0} \sigma^{-} &= \sigma^{-}     \nonumber \\
    \sigma^{1} \sigma^{-} &= 0              \nonumber \\
    \sigma^{-} \sigma^{-} &= 0              \nonumber \\
    \sigma^{+} \sigma^{-} &= \sigma^{1}     \nonumber \\
    \left( \sigma^{0} \right)^{x} \left( \sigma^{1} \right)^{y} \left( \sigma^{+} \right)^{v} \left( \sigma^{-} \right)^{w} \sigma^{-}
    &= \left( 1 - \Theta(y+w) \right) \left( \sigma^{1} \right)^{y+v} \left( \sigma^{-} \right)^{w+x}
\end{align}

\pagebreak
\subsection{For the Operators $\{ \mathbbm{1}_{2}, \sigma^{z}, \sigma^{-}, \sigma^{+} \}$}

Relevant to results derived in App.~\ref{app:suffquad} are multiplication identities for the single qubit operators
\begin{align}
\mathbbm{1} &= \ketbra{0}{0} + \ketbra{1}{1} , \nonumber \\
\sigma^{z}  &= \ketbra{0}{0} - \ketbra{1}{1} , \nonumber \\
\sigma^{-} &= \ketbra{0}{1} = \left( \sigma^{x} + i \, \sigma^{y} \right) / 2 ,  \nonumber\\
\sigma^{+} &= \ketbra{1}{0} = \left( \sigma^{x} - i \, \sigma^{y} \right) / 2 .
\end{align}

Here we list the left-multiplication and right-multiplication identities for each matrix. Assuming at most one of $x,v,w$ is $1$ and the rest $0$.

\begin{align}
    \sigma^{z} &\text{ left}                \nonumber \\
    \sigma^{z} \sigma^{z} &=  \mathbbm{1}   \nonumber \\
    \sigma^{z} \sigma^{-} &=  \sigma^{-}    \nonumber \\
    \sigma^{z} \sigma^{+} &= -\sigma^{+}    \nonumber \\
    \sigma^{z} \left( \sigma^{z} \right)^{x} \left( \sigma^{+} \right)^{v} \left( \sigma^{-} \right)^{w}
    &= \left( -1 \right)^{v} \left( \sigma^{+} \right)^{v} \left( \sigma^{-} \right)^{w} \left( \sigma^{z} \right)^{x+1} \nonumber \\
    \sigma^{z} &\text{ right}               \nonumber \\
    \sigma^{z} \sigma^{z} &=  \mathbbm{1}   \nonumber \\
    \sigma^{-} \sigma^{z} &= -\sigma^{-}    \nonumber \\
    \sigma^{+} \sigma^{z} &=  \sigma^{+}    \nonumber \\
    \left( \sigma^{z} \right)^{x} \left( \sigma^{+} \right)^{v} \left( \sigma^{-} \right)^{w} \sigma^{z} &= \left( -1 \right)^{w} \left( \sigma^{+} \right)^{v} \left( \sigma^{-} \right)^{w} \left( \sigma^{z} \right)^{x+1}
\end{align}

\end{document}